\documentclass[12pt, a4paper]{article} 
\pdfoutput=1

\usepackage[includeheadfoot,
            marginratio={1:1,2:3}, 
            width=412pt, 
            height=688pt,]{geometry}

\usepackage[dvipsnames]{xcolor}
\usepackage{amsmath}
\usepackage{amsfonts}
\usepackage{amssymb}
\usepackage{amsthm}
\usepackage{mathtools}
\usepackage{mathalfa}
\usepackage{graphicx,caption,subfigure}
\usepackage{tikz} 
\usetikzlibrary{decorations.pathreplacing,calc}
\usepackage{booktabs}
\usepackage{adjustbox}
\usepackage[utf8]{inputenc}
\usepackage[splitrule]{footmisc} 
\usepackage{footnotebackref}

\newcommand{\tikzmark}[2]{\tikz[remember picture,baseline,inner sep=0pt,outer sep=0pt,anchor=base] \node (#1) {\ensuremath{#2}};}
 
\usepackage{placeins}
\usepackage{empheq}
\usepackage{paralist}
\usepackage{cite}
\usepackage[normalem]{ulem}
\usepackage{soul}
\usepackage{float}
\usepackage{subfigure}
\usepackage{cancel}
\usepackage{enumitem}
\usepackage{spectralsequences}

\definecolor{MPP_green}{RGB}{0, 108, 102} 
\definecolor{MPP_green_dark}{cmyk}{0.85, 0.17, 0.45, 0.40} 
\definecolor{MPP_green_light}{cmyk}{0.25, 0.0, 1.0, 0.0} 
\definecolor{MPP_blue_dark}{cmyk}{0.89, 0.47, 0.25, 0.51} 
\definecolor{MPP_blue_light}{RGB}{0, 177, 234} 
\definecolor{MPP_orange}{cmyk}{0.0, 0.60, 1.0, 0.0} 

\usepackage{tensor}
\usepackage{hyperref}
\hypersetup{colorlinks,linkcolor={MPP_orange},citecolor={MPP_green},urlcolor={MPP_green}, linktocpage = true}

 \newcommand{\nc}{\newcommand}
 \nc{\lb}{\llbracket}
 \nc{\rb}{\rrbracket}
 \nc{\gl}{\llbracket}
 \nc{\gr}{\rrbracket}
 \nc{\del}{\partial}
 \nc{\tri}{\hspace{-3.5pt}\vartriangle\hspace{-3.5pt}}
 \nc{\blacktri}{\blacktriangle}

 \nc{\eq}[1]{\begin{equation}
                      \begin{split} #1 \end{split}
                      \end{equation}}
 \nc{\ov}{\overline}

 \nc{\fa}{\hat}
 \nc{\fb}{\MakeUppercase}
 \nc{\fc}{\tilde }
 \nc{\Lie}{{\cal L}} 
 \nc{\lambdabar}{{\mkern0.75mu\mathchar '26\mkern -9.75mu\lambda}}

\newtheorem{theorem}{Theorem}[section]
\newtheorem{lemma}[theorem]{Lemma}

\usepackage[outline]{contour}

\contournumber{25}
\contourlength{0.1pt}

\SseqNewClassPattern{linearnew}{
(0,0);
(-0.13,0)(0.13,0);
(-0.16,0)(0,0)(0.16,0);
(-0.24,0)(-0.08,0)(0.08,0)(0.24,0);
(-0.32,0)(-0.16,0)(0,0)(0.16,0)(0.32,0);
(-0.40,0)(-0.24,0)(-0.08,0)(0.08,0)(0.24,0)(0.40,0);
(-0.48,0)(-0.32,0)(-0.16,0)(0,0)(0.16,0)(0.32,0)(0.48,0);
}

\allowdisplaybreaks[2]
\numberwithin{equation}{section}

\DeclareUnicodeCharacter{2212}{-}
\begin{document}

\vspace*{-1.5cm}
\begin{flushright}
  {\small
  MPP-2024-159\\
  }
\end{flushright}

\vspace{1.5cm}
\begin{center}
  {\Large 
     Spin cobordism and the gauge group of type I/heterotic string theory} 
\vspace{0.4cm}

\end{center}

\vspace{0.35cm}
\begin{center}
{Christian Knei\ss l}
\end{center}

\

\vspace{0.1cm}
\begin{center} 
\emph{
Max-Planck-Institut f\"ur Physik (Werner-Heisenberg-Institut), \\[.1cm] 
   Boltzmannstr. 8, 85748 Garching, Germany 
   \\ [.1cm] \href{mailto:ckneissl@mpp.mpg.de}{ckneissl@mpp.mpg.de} } 
   \\[0.1cm] 
 \vspace{0.3cm} 
\end{center} 

\vspace{0.5cm}

\begin{abstract}
Cobordism offers a unique perspective into the non-perturbative sector of string theory by demanding the absence of higher form global symmetries for quantum gravitational consistency.
In this work we compute the spin cobordism groups of the classifying space of $Spin(32)/\mathbb{Z}_2$ relevant to describing type I/heterotic string theory and explore their (shared) non-perturbative sector. To facilitate this we leverage our knowledge of type I D-brane physics behind the related ko-homology.
The computation utilizes several established tools from algebraic topology, the focus here is on two spectral sequences. 
First, the Eilenberg-Moore spectral sequence is used to obtain the cohomology of the classifying space of the $Spin(32)/\mathbb{Z}_2$ with $\mathbb{Z}_2$ coefficients. This will enable us to start the Adams spectral sequence for finally obtaining our result, the spin cobordism groups. 
We conclude by providing a string theoretic interpretation to the cobordism groups.
\end{abstract}

\newpage

\tableofcontents

\setcounter{tocdepth}{2}

\newpage

\section{Introduction}

By now it is a well-established assertion in the string literature that both the type I and one of the two supersymmetric heterotic string theories in ten dimensions actually share the same gauge group, denoted in the string theory literature as $Spin(32)/\mathbb{Z}_2$, out of the various Lie groups with a  $\mathfrak{so}(32)$ Lie algebra.
While consistency of the heterotic string worldsheet theory immediately singles out this Lie group~\cite{Gross:1984dd},
for type I string theory the story is a bit more intricate as perturbatively the symmetry group is $O(32)/\mathbb{Z}_2$~\cite{Witten:1998cd}.
The story changes when accounting for the D9-branes cancelling the Ramond-Ramond tadpole (as well as the NS-NS tadpole) leading to a SO(32) gauge symmetry group. 
However, this is still not the full picture.

The discovery of S-duality between the $Spin(32)/\mathbb{Z}_2$ heterotic string (HO string) and type I string theory~\cite{Witten:1995ex, Polchinski:1995df}, i.e.\ they provide the weak/strong coupling description of the same complete theory, suggests that both theories have to be built upon the same gauge group~\cite{Witten:1998cd}\footnote{It's quite insightful to contrast S-duality between the two $Spin(32)/\mathbb{Z}_2$ string theories with Montonen-Olive duality \cite{Goddard:1976qe, Montonen:1977sn}, since the strong coupling dual gauge group of the latter duality is the Langlands-dual of the gauge group one started with. Therefore, our physical expectation that this also holds true in string theory matches up nicely with the mathematical fact that $Spin(32)/\mathbb{Z}_2$ is in fact Langlands-self-dual, see for example appendix D of~\cite{Goddard:1976qe}.}.
Subsequent analysis of nonperturbative non-BPS objects in type I showed that they are responsible for amending the perturbative gauge group $O(32)/\mathbb{Z}_2$ to $Spin(32)/\mathbb{Z}_2$ on the type I side~\cite{Witten:1998cd}.

The notation of the gauge group as $Spin(32)/\mathbb{Z}_2$ can be quite confusing as for $n=4k$ (with $k \in \mathbb{N}$) the center $Z$ of the Lie group $Spin(n)$ is $\mathbb{Z}_2 \oplus \mathbb{Z}_2$ and thereby there are three possible quotients. One of them is $SO(n)$, which is sometimes erroneously identified with the gauge group of type I/HO string theory.
The other two quotient groups are in fact isomorphic and are commonly referred to as the $SemiSpin(n)$ group in the mathematical literature, usually abbreviated as $Ss(n)$, which is the notation we will stick to throughout the paper.
Importantly, this group has fundamentally different properties than $SO(n)$. For example, the two groups are of different homotopy type~\cite{BAUM1965305}. A lot of other subtle differences, especially in the context of string dualities, were highlighted in~\cite{McInnes:1999va, McInnes:1999pt}.
In the following we want to actually calculate the cobordism groups relevant to non-perturbative objects in the type I/HO string theory considering the correct gauge group.

\subsection{Cobordism and quantum gravity}

To outline why we are interested in certain cobordism groups let's look at the basics of cobordism theory through the lens of quantum gravity.
The underlying equivalence relation placing two manifolds in the same cobordism class is from a physical point of view nothing else than 
the ability of quantum gravity to change topology\footnote{An excellent summary of the necessity to consider topology change to be a fundamental property of any complete theory of quantum gravity was given in this recent talk by McNamara~\cite{McNamara@Swamplandia}.}.
Additionally, we have the disjoint union adding two different manifolds 
\begin{equation}
    [M] + [N] = [M \cup N]\,,
\end{equation}
which entails that there exists a topology changing process joining two compact manifolds on which our quantum gravity theory is defined. 

Furthermore, we can impose physical restrictions through the structure $\xi$ of
our respective cobordism group $\Omega^{\xi}$, which means that we consider equivalence classes of manifolds which are connected through a cascade of topology changes preserving this structure.
An example, which will play an important role in the following, is spin structure,
which from a physical standpoint originates from requiring fermionic fields (without coupling to an $U(1)$-field or considering time-reversal symmetry) to be well-defined throughout every topology changing operation. 

Finally, we can consider cobordism groups $\Omega^{\xi}(X)$ of a fixed manifold $X$, where we are now only considering the pairs $[M, f]$ and $[N, g]$ with continuous maps $f: M \to X$ and $g: N \to X$ to be cobordant if the property of having a continuous map to $X$ is preserved by the cobordism $W$ between $M$ and $N$ as well. 

Together with the Cobordism Conjecture, which we explain right after, this restriction to fix a kind of ``background" manifold $X$ results in the same pattern one would have expected from conventional dimensional reduction on this ``background" manifold~\cite{Blumenhagen:2022bvh}. In the next section we will take a look at a different class of ``background" manifolds with another distinct physical interpretation. 

Now, suppose one has calculated all of these specific cobordism groups modelling our theory what is the interpretation of these groups?
As part of the Swampland Program~\cite{Vafa:2005ui} (see~\cite{Palti:2019pca, Grana:2021zvf, vanBeest:2021lhn, Agmon:2022thq} for reviews) nontrivial cobordism groups have been linked to  higher-form global symmetries~\cite{McNamara:2019rup}. The idea that global symmetries are incompatible with quantum gravity is a conjecture with a long history~\cite{Misner:1957mt, Banks:1988yz} and has been formalized, e.g.~\cite{Banks:2010zn, Gaiotto:2014kfa}, and strengthened recently via a proof in Anti-de-Sitter space through the AdS/CFT correspondence~\cite{Harlow:2018jwu, Harlow:2018tng}.
The consequence of non-vanishing cobordism groups representing higher form global symmetries
is that the cobordism groups endowed with the full set of restrictions of some complete theory of quantum gravity have to be trivial
\begin{equation}
    \Omega_n^{QG} = 0\,.
\end{equation}
While it might be impractical to determine these final cobordism groups,
it can be very fruitful to compartmentalize this task. To this end we can focus on specific aspects of the theory we are working with thereby exploring how the full theory takes care of these obstacles.
For example one might want to specify certain duality groups as explored in~\cite{Debray:2021vob, Dierigl:2022reg, Debray:2023yrs} or the gauge group of the respective string theory~\cite{Debray:2023rlx, Basile:2023knk} combined with a tangential spin or string structure. 

Moreover, there seems to be a general organisational principle at play here. Namely, the physically sensible structures in string theory seemingly all branching off from the structures organized in the Whitehead tower of the classifying space of the orthogonal group $BO(n)$, which has been discussed a lot more in depth in~\cite{Andriot:2022mri} and also features prominently in the understanding of D-brane and O-plane Wess-Zumino topological couplings from the cobordism perspective of quantum gravity~\cite{Basile:2024}. 

Now, there are two options to get rid off the global symmetry indicated by a non-trivial cobordism group: breaking or gauging.
Breaking it entails including defects such that the modified cobordism group is mapped to the trivialized version:
\begin{equation}
    \Omega_n^{\widetilde{QG}} \to \Omega_n^{\widetilde{QG} + defects} = 0\,,
\end{equation}
where we use the same notation as in~\cite{McNamara:2019rup} to label any incomplete quantum gravity structure leading to a nontrivial cobordism group as $\widetilde{QG}$.
Gauging on the other hand means turning on gauge fields corresponding to the symmetry causing total charge cancellation, i.e.\ only the class with vanishing cobordism invariant(s) $[M] = 0 \in \Omega_n^{\widetilde{QG}}$ is a consistent quantum gravity configuration.
Moreover, triviality of cobordism groups can be utilized to check global anomaly cancellation for the theory at hand, see for example~\cite{Garcia-Etxebarria:2018ajm} for a nice review on this subject, and provide another major motivation to study cobordism groups, also divorced from quantum gravity.  

\subsection{Incorporating gauge groups into cobordism considerations}

As already touched upon in the previous section the concept we want to introduce is classifying spaces $BG$ as ``background gauge" manifolds for the cobordism groups, i.e.\ we would like to study $\Omega_n^{\xi}(X = BG)$.
This is quite distinct from the case of a finite dimensional background manifold $X$, which models dimensional reduction~\cite{Blumenhagen:2022bvh}, as classifying spaces are commonly infinite-dimensional.

More specifically, the map we are fixing to be preserved is a classifying map $M \to BG$,
which assigns a G-principal bundle to the manifolds in our equivalence classes within each cobordism group due to the bijection between the homotopy class of the classifying map and the isomorphism class of numerable G-principal bundles~\cite{dieck2008algebraic}.
From a physics point of view the G-principal bundle is the topological backbone 
of the gauge theory we want to fix. In our case this is the unique gauge group arising in type I/HO string theory.

It should be emphasized that the cobordism groups of some classifying space should be considered a particular, intermediate step in the approximation of the ``final" quantum gravity structure, each step revealing a different aspect of quantum gravity. In the case at hand we gain insight into the D-brane bound states in correspondence with the presence of a Yang-Mills theory, introduced from the type I perspective by the tadpole cancelling stack of 32 D9-branes and one O9-plane.

In what follows we will utilize the Adams spectral sequence to calculate the spin cobordism groups below dimension 13 of the classifying space of $Ss(32)$.
Before we start with the actual spectral sequences let us introduce the mathematical concept of spectral sequences and our main actors, the Adams and Eilenberg-Moore spectral sequence, in section \ref{sec:math_background}, again from the viewpoint of a physicist.
Afterwards we dive into the computation of $\Omega_{n \leq 12}^{Spin}(BSs(32))$ in section \ref{sec:calculation_spin_bordism}. Finally, we interpret the result from a string theoretic point of view, especially in the context of the Cobordism Conjecture, in section \ref{sec:physics_interpretation}. 

\section[Mathematical background on spectral sequences]{\texorpdfstring{Mathematical background on spectral\\ sequences}{Mathematical background on spectral sequences}}
\label{sec:math_background}
Spectral sequences come in a wide variety of variants and applications. So let's start by recollecting the fundamental, shared properties loosely following the nice presentation in~\cite{McCleary_2000}. Since spectral sequences take such an important place in algebraic topology, there are many more excellent textbook accounts, see for example~\cite{hatcher2004spectral, mosher2008cohomology}. For the Adams spectral sequence specifically~\cite{Beaudry:2018ifm} provides a brilliant introduction. For the ample applications of spectral sequences in physics, in particular high energy physics, please take a look at for example~\cite{Witten:1985bt, Stong:1985vj, Freed:2016rqq, Garcia-Etxebarria:2018ajm, Wan:2018bns, Davighi:2020bvi, Davighi:2020uab, Lee:2020ewl, Davighi:2020kok, Debray:2021rik, Lee:2022spd, Blumenhagen:2022bvh, Davighi:2022icj, McNamara:2022lrw, Debray:2023yrs, Debray:2023rlx, Davighi:2023luh, Debray:2023tdd}.

The aim of deploying a spectral sequence is usually to calculate some object $G_{*}$, which in our case will be a (generalized) (co-)homology satisfying all Eilenberg-Steenrod axioms \cite{eilenberg1945axiomatic}  (see e.g.~\cite{may1999concise} for a nice textbook account) with the exception of the dimensionality axiom in the generalized case. Additionally, $G_{*}$ is filtered, meaning that we can define a series of sub-objects organized in the following way:
\begin{equation}
    G_{*} = F_{0} G_{*} \supset \dots \supset F_{n} G_{*} \supset \dots \supset \{0\}\,.
\end{equation}
This filtration can be used to define the so called associated graded vector space as an approximation to $G_{*}$:
\begin{equation}
    E_{p,q} = F_{p} G_{p+q}/\, F_{p+1} G_{p+q}\,.
\end{equation}
Then we can get $G_{r}$ by summing over $p$ and $q$:
\begin{equation}
    G_{r} = \bigoplus_{p+q \,= \,r} E_{p,\,q}\,.
\end{equation}
Now, a spectral sequence consists of a sequence of differential bigraded vector spaces, which means that we have $r \in \mathbb{N}$ bigraded vector spaces, called pages, $E^{r}_{p, q}$. We will even specialize to first quadrant spectral sequences, i.e.\ $p, q \geq 0$, which will be sufficient for our computations. On top these vector spaces have a differential, i.e.\ a linear mapping within a spectral sequence page $d_r: E^{r}_{p,\,q} \to E^{r}_{p-r,\,q+r-1}$ with $d_r \circ d_r = 0$. 
Then we can calculate the next page in the following way:
\begin{equation}
    E^{r+1}_{p,q} = \frac{\text{ker} \,d_r: E^{r}_{p,\, q} \to E^{r}_{p-r,\,q+r-1}}{\text{im} \,d_r: E^{r}_{p+r,\,q-r+1} \to E^{r}_{p,\,q}}\,.
\end{equation}
In certain instances, e.g.\ the Adams spectral sequence for unoriented cobordism~\cite{McNamara:2022lrw} or the Atiyah-Hirzebruch spectral sequence for complex K-theory~\cite{Maldacena:2001xj},
the authors were able to assign a clear physical interpretation to the differentials. 

In the case of the Atiyah-Hirzebruch spectral sequence for complex K-theory they enforce some physical condition, for example the first differential $d_3$ is eliminating groups, which classify charges of type II D-brane configurations not satisfying Freed-Witten anomaly cancellation\cite{Freed:1999vc, Maldacena:2001xj}. 

Coming back to our introduction to spectral sequences, suppose $G_{*}$ has a filtration and we converge to a ``infinity page", which is our approximation to $G_{*}$:
\begin{equation}
    E^{\infty}_{p,\,q} = F_{p} G_{p+q}/\, F_{p+1} G_{p+q}\,.
\end{equation}
Finally, since going through infinitely many pages isn't tractable, we want to focus on spectral sequences collapsing at some page $r=N$, such that $d_r = 0$ for $r \geq N$. Then 
\begin{equation}
    E^{N}_{*,\,*} \cong E^{N+1}_{*,\,*} \cong \dots \cong E^{\infty}_{*,\,*}\,.
\end{equation}
Given, of course, the fixed input by choosing a certain $G_{*}$ we have now reached a stable configuration for the objects classified by $G_{*}$.
From a physics point of view we have taken all possible decay channels (given the input) into account. 
With this short introduction to the concept of a spectral sequence we will now move on to the Adams Spectral sequence, which is the main tool we will be working with.

\subsection{The Adams spectral sequence}

The tool of choice for us to determine $\Omega_n^{Spin}(BSs(32))$ will be the Adams spectral sequence. While this spectral sequence was invented to compute the stable homotopy groups of the spheres, it is also particularly useful for the calculation of certain generalized homology theories of classifying spaces, like connective K- or KO-homology. The connective versions, i.e.\ only non-vanishing for $n \geq 0$, will be denoted by lowercase k- or ko-homology.

In general terms  Adams~\cite{Adams1957/58} devised the following spectral sequence:
\begin{equation}
    E_{2}^{s,t} = Ext^{s,t}_{\mathcal{A}}(H^*(X, \mathbb{Z}_p), \mathbb{Z}_p) \Rightarrow \pi_{t-s}(X)_{\widehat{p}}\,,
\end{equation}
where $\mathcal{A}$ denotes the Steenrod algebra. While we will refrain from giving an extensive introduction to the Steenrod algebra (and its subalgebras), we will collect some useful information about the Steenrod squares and the algebra they generate in the appendix.
Now, it turns out we can just focus just on the 2-torsion part as we will show later on. Here, by utilizing the Anderson-Brown-Peterson decomposition~\cite{ABP} we can express spin cobordism in terms of connective ko-homology at the prime $p = 2$:
\begin{equation}
\label{ABP}
    \Omega_n^{Spin}(X)_{\widehat{2}} = ko_n(X)_{\widehat{2}} \, \oplus \,  ko_{n-8}(X)_{\widehat{2}} \, \oplus \,  ko_{n-10}\langle 2 \rangle(X)_{\widehat{2}} \oplus \,\dots\,.
\end{equation}
This means that as a first step we need to calculate the connective ko-homology of $BSs(32)$.
In particular, the calculation of real connective k-homology $ko_{*}(pt) = \pi_{*}(ko)$ or more generally $ko_{*}(X) = \pi_{*}(ko \wedge X)$ simplifies due to work of Stong~\cite{Stong63}:
\begin{equation}
    H^{*}(ko, \mathbb{Z}_2) \cong \mathcal{A} \otimes_{\mathcal{A}_1} \mathbb{Z}_2\,.
\end{equation}
We get by change of rings, see for example~\cite{Beaudry:2018ifm}:
\begin{equation}
    E_{2}^{s,t} = \text{Ext}^{s,t}_{\mathcal{A}_1}(H^*(X, \mathbb{Z}_2), \mathbb{Z}_2) \Rightarrow ko_{t-s}(X)_{\widehat{2}}\,.
\end{equation}

Unfortunately, the mod 2 cohomology of the classifying space of the 
$SemiSpin$ group $H^*(BSs(4n),\mathbb{Z}_2)$ (specifically for $BSs(32)$) is not fully known. However, since we are only interested in the cobordism groups up to dimension $n=12$, it will suffice to calculate $H^*(BSs(4n),\mathbb{Z}_2)$ up to degree $* \leq 13$. In particular, to set up the second page of the Adams spectral sequence we will need to determine the Steenrod operations within $H^*(BSs(4n), \mathbb{Z}_2)$ in order to get the structure of $H^*(BSs(4n), \mathbb{Z}_2)$ in terms of $\mathcal{A}_{1}$-modules. Before we tackle this question let us introduce a couple more useful concepts related to the Adams spectral sequence, which will be of importance in the following.

\subsubsection*{From \texorpdfstring{$H^*(X,\mathbb{Z}_2)$}{H*(X,Z/2)} to \texorpdfstring{$\mathrm{Ext}^{s,t}_{\mathcal{A}_1}(H^*(X, \mathbb{Z}_2), \mathbb{Z}_2)$}{Ext(H*(X, Z/2), Z/2)} - The second page of the Adams spectral sequence}
Suppose that we have determined the Steenrod algebra for $H^*(X,\mathbb{Z}_2)$. In our case the subalgebra $\mathcal{A}_1$ is sufficient.
Then we can fill the second page of the Adams spectral sequence, which means that we want to determine
$\mathrm{Ext}_{\mathcal{A}_1}^{s,t}(H^*(X,\mathbb{Z}_2), \mathbb{Z}_2))$. Hence, we need to recollect some facts about $\mathrm{Ext}^{s,t}_{\mathcal{R}}(M,N)$ and especially how to get 
$\mathrm{Ext}^{s,t}_{\mathcal{A}_1}(H^{*}(X, \mathbb{Z}_2),\mathbb{Z}_2)$
from the $\mathcal{A}_1$-module structure of $H^{*}(X, \mathbb{Z}_2)$.
Again,~\cite{Beaudry:2018ifm} contains a more comprehensive discussion of this topic as we will keep this part concise.
$\mathrm{Ext}^{s,t}_{\mathcal{R}}(M,N)$ can be understood as equivalence classes of extensions with $s \geq 1$. An element of $\mathrm{Ext}^{s,t}_{\mathcal{R}}(M,N)$ then represents the following extension:
\begin{equation}
\label{extension_Ext}
    0 \to \Sigma^t N \to P_1 \to \dots \to P_s \to M \to 0\,,
\end{equation}
where $\Sigma^t N$ denotes the $t$-th reduced suspension of $N$ and both $M$ and $N$ are $\mathcal{R}$-modules.
Let us first introduce the two most important classes in $\mathrm{Ext}_{\mathcal{A}_1}(\mathbb{Z}_2, \mathbb{Z}_2)$ that are conventionally used to depict the second page $\mathrm{Ext}_{\mathcal{A}_1}^{s,t}(H^*(X,\mathbb{Z}_2), \mathbb{Z}_2))$.
These are $h_0 = \mathrm{Ext}^{1,1}_{\mathcal{A}_1}(\mathbb{Z}_2,\mathbb{Z}_2)$
\begin{equation}
\label{h0}
    0 \to \Sigma \mathbb{Z}_2 \to \Sigma^{-1} H^*(\mathbb{RP}^2, \mathbb{Z}_2) \to \mathbb{Z}_2 \to 0
\end{equation}
and $h_1 = \mathrm{Ext}^{1,2}_{\mathcal{A}_1}(\mathbb{Z}_2,\mathbb{Z}_2)$
\begin{equation}
\label{h1}
    0 \to \Sigma^2 \mathbb{Z}_2 \to \Sigma^{-2} H^*(\mathbb{CP}^2, \mathbb{Z}_2) \to \mathbb{Z}_2 \to 0\,.
\end{equation}
The usual convention is to use coordinates $(t-s, s)$ for all of the pages of the Adams spectral sequence, such that each connective ko-homology group $ko_{t-s}(X)$ can be read off as a column on the infinity page.
Each $\mathbb{Z}_2$ summand within $\mathrm{Ext}_{\mathcal{A}_1}^{s,t}(H^*(X,\mathbb{Z}_2), \mathbb{Z}_2))$ amounts to a dot on the second page,
while $h_0$'s are vertical lines raising $s$ and $t$ by 1 and $h_1$'s are diagonal lines raising $s$ by 1, while $t$ gets raised by 2.
There are two more classes of $\mathrm{Ext}_{\mathcal{A}_1}(\mathbb{Z}_2, \mathbb{Z}_2)$, which are customarily not depicted to avoid cluttering the Adams pages, namely $v$ of degree $(7, 3)$, whose action raises $s$ by 3 and $t-s$ by 4, and $w$ of degree $(12, 4)$ raising $s$ by 4 and $t-s$ by 8.
All of these actions have interpretations, e.g.\ as multiplication by 2 in the case of $h_0$ or multiplication by certain manifolds of appropriate degree in $t-s$. 
We will go into more detail once we actually utilize these properties.

There are now two convenient pathways to determine $\mathrm{Ext}^{s,t}_{\mathcal{R}}(M,N)$: Minimal resolutions and long exact sequences. Since for most of the $\mathcal{A}_1$-modules $\mathcal{M}$ we have encountered in $H^n(BSs(32), \mathbb{Z}_2)$ their $\mathrm{Ext}_{\mathcal{A}_1}^{s,t}(\mathcal{M}, \mathbb{Z}_2)$ can be found explicitly in the literature, we will only explain the method via long exact sequences for the concrete example of $\tilde{R}_2$, which we are going to come across later. A lot more information on minimal resolutions can be found for example in Beaudry-Campbell~\cite{Beaudry:2018ifm}, which we will follow in the presentation of the long exact sequence method.

Below we include a depiction of $\tilde{R}_2$. Since $\mathcal{A}_1$ is only generated by $Sq^1$ and $Sq^2$ the different nodes are connected by just two types of lines: straight ones raising degree by 1 corresponding to $Sq^1$ and curved lines representing $Sq^2$s. \\
\begin{sseqdata}[name= tildeR2, classes = {fill}, no axes,
class pattern = linearnew]


\class[black](0,0)
\class[black](0,1)
\class[black](0,2)
\class[black](0,3)
\class[black](2,3)
\class[black](2,4)
\class[black](2,5)

\structline[black](0,0)(0,1)
\structline[black](0,2)(0,3)
\structline[black, bend left = 30](0,0)(0,2)
\structline[black, in = -160, out = 20](0,1)(2,3)
\structline[black, in = -160, out = 20](0,2)(2,4)
\structline[black](2,3)(2,4)
\structline[black, in = -160, out = 20](0,3)(2,5)

\end{sseqdata}

\begin{figure}[H]
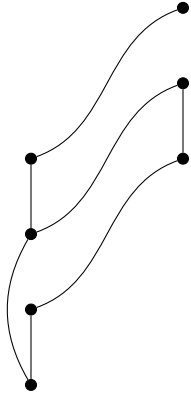

    \centering
    \printpage[ name = tildeR2, page = 2] 
    \caption{The $\mathcal{A}_1$-module - $\tilde{R}_2$}
    \label{fig:A1_module_R2}
\end{figure}

Consider the exact sequence of (in our specific case) $\mathcal{A}_1$-modules:
\begin{equation}
\label{exactseqtildeR2}
   0 \to \Sigma^3 C\eta \to \tilde{R}_2 \to J \to 0\,,
\end{equation}
which we can depict as \newline
\begin{sseqdata}[name= extensiontildeR2, classes = {fill}, no axes,
class pattern = linearnew]


\class[black](0,3)
\class[black](0,5)
\structline[black, bend left = 30](0,3)(0,5)


\class[black](2,3)
\class[black](2,5)

\class[black](3,0)
\class[black](3,1)
\class[black](3,2)
\class[black](3,3)
\class[black](3,4)
\structline[black, bend left = 30](2,3)(2,5)
\structline[black](3,0)(3,1)
\structline[black](2,3)(3,2)
\structline[black, bend left = 30](3,0)(3,2)
\structline[black, bend left = 30](3,2)(3,4)
\structline[black, bend right = 30](3,1)(3,3)
\structline[black](3,3)(3,4)


\class[black](5,0)
\class[black](5,1)
\class[black](5,2)
\class[black](5,3)
\class[black](5,4)
\structline[black](5,0)(5,1)
\structline[black, bend left = 30](5,0)(5,2)
\structline[black, bend left = 30](5,2)(5,4)
\structline[black, bend right = 30](5,1)(5,3)
\structline[black](5,3)(5,4)

\draw [-To, MPP_orange, line width=0.75pt](0,3) -- (2,3);
\draw [-To, MPP_orange, line width=0.75pt](0,5) -- (2,5);
\end{sseqdata}

\begin{figure}[H]
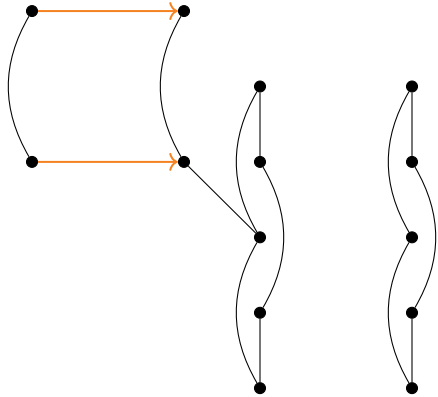

    \centering
    \printpage[ name = extensiontildeR2, page = 2] 
    \caption{Exact sequence for $\tilde{R}_2$}
    \label{fig:exseqR2}
\end{figure}

Now, this exact sequence leads to this dual long exact sequence 
\newpage
\begin{align}
    \dots \longrightarrow &\mathrm{Ext}_{\mathcal{A}_{1}}^{s,t}(J, \mathbb{Z}_2)) \longrightarrow \mathrm{Ext}_{\mathcal{A}_{1}}^{s,t}(\tilde{R}_2, \mathbb{Z}_2)) \longrightarrow \tikzmark{Pos1}{\mathrm{Ext}_{\mathcal{A}_{1}}^{s,t}(\Sigma^3 C\eta, \mathbb{Z}_2))} \\
    & \nonumber \\
    & \nonumber \\
    &\tikzmark{Pos2}{\mathrm{Ext}_{\mathcal{A}_{1}}^{s+1,t}(J, \mathbb{Z}_2))} \longrightarrow \mathrm{Ext}_{\mathcal{A}_{1}}^{s+1,t}(\tilde{R}_2, \mathbb{Z}_2))  \longrightarrow \dots \nonumber
     \begin{tikzpicture}[overlay, remember picture, shorten < = 2.5pt, shorten > = 2.5pt]
    \draw[-to, line width=0.75pt] (Pos1) -- (Pos2) node [pos=.64, above, sloped] (TextNode1) {$\delta$};
     \end{tikzpicture}
\end{align}
By splicing our exact sequence \eqref{exactseqtildeR2} into the extension corresponding to \\
$\mathrm{Ext}_{\mathcal{A}_{1}}^{s,t}(\mathcal{A}_1/\mathcal{E}_1, \mathbb{Z}_2)$ we extend it to
\begin{align}
 0 \to \Sigma^t \mathbb{Z}_2 \to P_1 \to \dots \to &\tikzmark{Pos1}{P_s} \hspace{2.4 cm} \tikzmark{Pos3}{\tilde{R}_2} = P_{s+1} \to J \to 0\,. \\
 & \nonumber \\
 &\qquad \tikzmark{Pos2}{\Sigma^3 C\eta}\nonumber 
 \begin{tikzpicture}[overlay, remember picture, shorten < = 2.5pt, shorten > = 2.5pt]
     \draw[-to, line width=0.75pt] (Pos1) -- (Pos2);
     \draw[-to, line width=0.75pt] (Pos1) -- (Pos3);
     \draw[-to, line width=0.75pt] (Pos2) -- (Pos3);
 \end{tikzpicture}
\end{align}
Therefore, we can understand  $\delta$ as mapping elements in $\mathrm{Ext}_{\mathcal{A}_{1}}^{s,t}(\Sigma^3 C\eta, \mathbb{Z}_2)$ to their boundary representative in $\mathrm{Ext}_{\mathcal{A}_{1}}^{s+1,t}(J, \mathbb{Z}_2)$.
Since we know $\mathrm{Ext}_{\mathcal{A}_{1}}^{s,t}(\Sigma^3 C\eta, \mathbb{Z}_2)$ and $\mathrm{Ext}_{\mathcal{A}_{1}}^{s,t}(J, \mathbb{Z}_2)$, see for example~\cite{Beaudry:2018ifm},
we can get $\mathrm{Ext}_{\mathcal{A}_{1}}^{s,t}(\tilde{R}_2, \mathbb{Z}_2)$ as well. To represent the $\mathrm{Ext}$-functors we transition to the so called Adams charts, which use the same conventions as the second page of the Adams spectral sequence. A key difference is that we want to track the 
effect of our boundary map $\delta$, which can be understood as a differential $d_{1}$. This is, of course, by definition never part of a second page of a spectral sequence. Differentials $d_r$ in an Adams chart (on a page $E_r$ in the spectral sequence) are understood to go from a node $(s, t-s)$ to another one at $(s+r, t-s-1)$. So in our example possible differentials
go from nodes coming from $\Sigma^3 C \eta$ to nodes coming from $J$.
They eliminate all of the nodes connected via a differential. 
Keeping this in mind we get the following Adams chart, where we encircled all of the nodes, which are not affected by a differential: \\
\\
\begin{sseqdata}[
name = extensionR2,
Adams grading, classes = fill,
x range = {0}{13}, y range = {0}{7},
x tick step = 1,
run off differentials = {->},
xscale = 0.75,
class pattern = linearnew
]

\class[MPP_blue_dark, circlen = 2](0,0)
\class[MPP_blue_dark](2,1)
\DoUntilOutOfBoundsThenNMore{2}{
\class[MPP_blue_dark](\lastx,\lasty+1)
\structline[MPP_blue_dark]
}
\class[MPP_blue_dark, circlen = 2](6,2)
\DoUntilOutOfBoundsThenNMore{2}{
\class[MPP_blue_dark](\lastx,\lasty+1)
\structline[MPP_blue_dark]
}
\class[MPP_blue_dark, circlen = 2](7,3)
\structline[MPP_blue_dark](6,2)(7,3)
\class[MPP_blue_dark](8,4)
\structline[MPP_blue_dark](7,3)(8,4)

\class[MPP_blue_dark](10,5)
\DoUntilOutOfBoundsThenNMore{2}{
\class[MPP_blue_dark](\lastx,\lasty+1)
\structline[MPP_blue_dark]
}

\class[MPP_blue_light](3,0)
\d1
\DoUntilOutOfBounds{
\class[MPP_blue_light](\lastx,\lasty+1)
\structline[MPP_blue_light]
\d1
}
\class[MPP_blue_light, circlen = 2](5,1)
\extension[dashed](5,1)(6,2)
\DoUntilOutOfBounds{
\class[MPP_blue_light, circlen = 2](\lastx,\lasty+1)
\structline[MPP_blue_light]
}

\class[MPP_blue_light](7,2)
\d1
\DoUntilOutOfBounds{
\class[MPP_blue_light](\lastx,\lasty+1)
\structline[MPP_blue_light]
\d1
}

\class[MPP_blue_light](9,3)
\d1
\DoUntilOutOfBounds{
\class[MPP_blue_light, circlen = 2](\lastx,\lasty+1)
\structline[MPP_blue_light]
}

\class[MPP_blue_light](11,4)
\d1
\DoUntilOutOfBounds{
\class[MPP_blue_light](\lastx,\lasty+1)
\structline[MPP_blue_light]
\d1
}

\class[MPP_blue_light, circlen = 2](13,5)
\DoUntilOutOfBounds{
\class[MPP_blue_light, circlen = 2](\lastx,\lasty+1)
\structline[MPP_blue_light]
}

\end{sseqdata}

\begin{figure}[H]
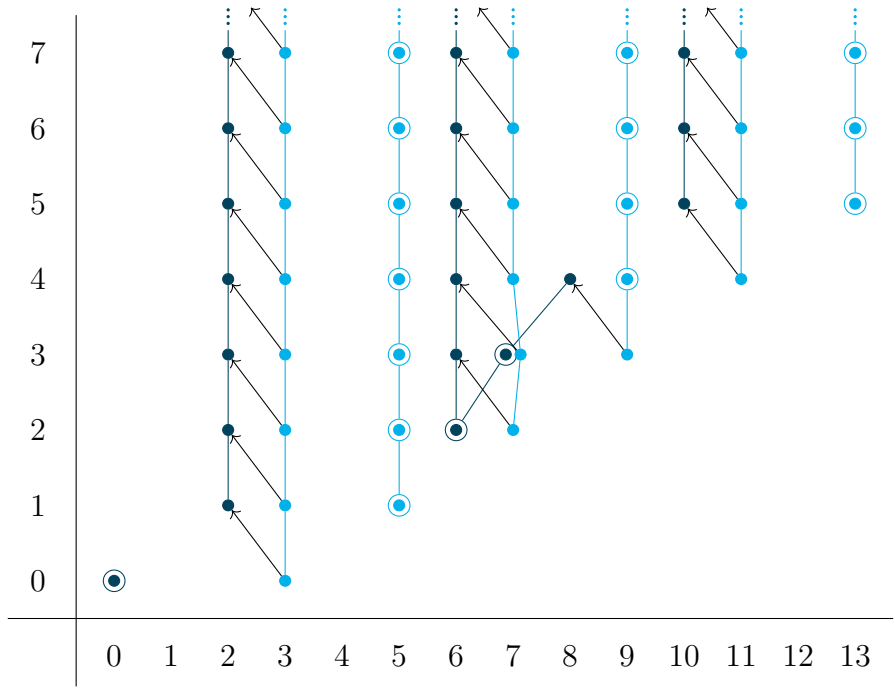

    \centering
    \printpage[ name = extensionR2, page = 1 ]
    \caption{Adams chart for $\tilde{R}_2$ extension including $\delta$}
    \label{fig:R2_Adamschart1}
\end{figure}

For the final result below we added an extension (dashed) not seen through this specific exact sequence. There are two different methods to see this extension. Either by working with minimal resolutions, which can be a bit tedious, or alternatively, we can look at the following exact sequence involving $\tilde{R}_2$: 
\begin{equation}
    0 \to \Sigma^6 \mathbb{Z}_2 \to \mathcal{A}_1 \to \tilde{R}_2 \to 0 \,.
\end{equation}
Since we know the Adams charts for $\mathbb{Z}_2$ and $\mathcal{A}_1$, one can get $\tilde{R}_2$ as well, albeit without the extension problem. 

\begin{figure}[H]
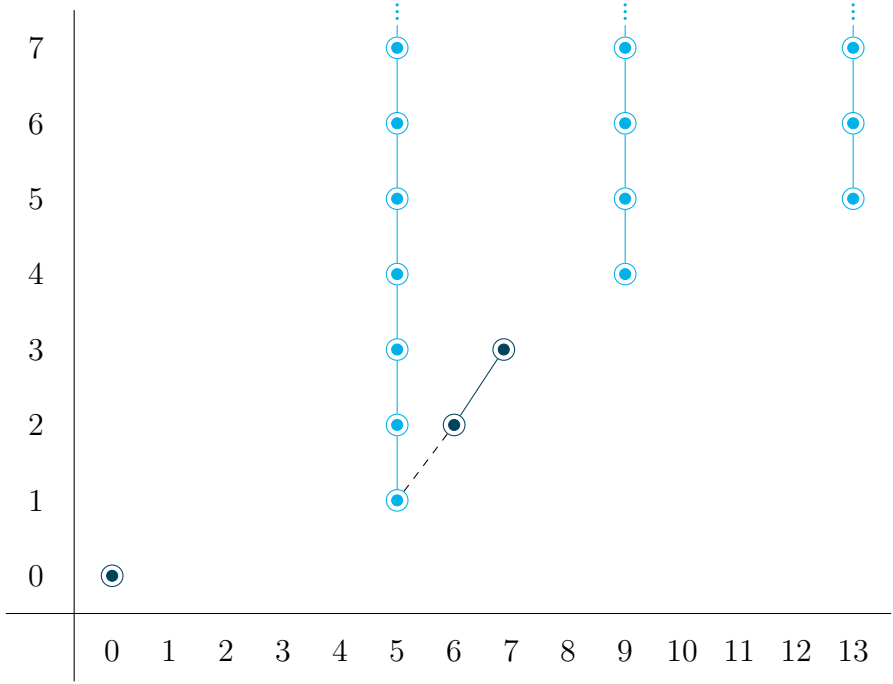

    \centering
    \printpage[ name = extensionR2, page = \infty]  
    \caption{Final Adams chart for $\tilde{R}_2$ extension}
    \label{fig:R2_Adamschart2}
\end{figure}

\subsection{The Eilenberg-Moore spectral sequence}
\label{EilenbergMoore}

Since we need the mod 2 cohomology of the classifying space of the $SemiSpin(32)$ group $BSs(32)$ to fill out the second page of the Adams spectral sequence, which has not been fully determined in the literature yet, we need a tool to compute this mod 2 cohomology from something we already know. Luckily, the so called Eilenberg-Moore spectral sequence~\cite{moore1959algebre, rothenberg1965cohomology, eilenberg1966homology} provides a way to compute the mod 2 cohomology of the classifying space of a compact Lie group $BG$ from the mod 2 cohomology of the compact Lie group:
\begin{equation}
\label{EMSS}
    E_2 = \mathrm{Cotor}_{H^*(Ss(32),\mathbb{Z}_2)}(\mathbb{Z}_2, \mathbb{Z}_2) \Longrightarrow H^*(BSs(n),\mathbb{Z}_2)\,.
\end{equation}
First, let's quickly introduce the $\mathrm{Cotor}$ functor $\mathrm{Cotor}_{R}(M, N)$ following~\cite{may1999concise}.
The construction of the $\mathrm{Cotor}$ functor is similar to the $\mathrm{Ext}$ functor, which we have just encountered as a primary ingredient of the Adams spectral sequence.
We start with the cotensor product for a right comodule A and a left comodule B over C. With a coalgebra $C$ we define the cotensor product between $A$ and $B$ as:
\begin{equation}
    A \,\square_C \, B \coloneqq ker(\phi_{A} \otimes id_{B} - id_{A} \otimes \phi_{B})\,,
\end{equation}
where $\phi_{A}: A \to A \otimes C$ and $\phi_{B}: B \to C \otimes B$ are the respective comultiplications.
Then considering the following injective resolution for A with comodules $A_{*}$ over C
\begin{equation}
    A \to A_0 \to A_1 \to \dots A_{n} \to \dots
\end{equation}
we can introduce the following series of cotensor products:
\begin{equation}
    A_{0} \, \square_C \, B \to A_{1} \, \square_C \, B \to \dots \to A_{n} \, \square_C \, B \to \dots
\end{equation}
Finally, the $\mathrm{Cotor}$ functor is now the cohomology of this resolution
\begin{equation}
    \mathrm{Cotor}^{n}_{C}(A, B) = H^{n}(A_{0} \, \square_C \, B \to A_{1} \, \square_C \, B \to \dots \to A_{n} \, \square_C \, B \to \dots)\,.
\end{equation}
For the spectral sequence itself there arises another index $q$ from the fact that $A$, $B$ and $C$ can be graded, such that $\bigoplus_q \mathrm{Cotor}_{C}^{p, q}(A, B) = \mathrm{Cotor}_{C}^{p}(A, B)$. In our specific spectral sequence \eqref{EMSS} $C = H^*(Ss(32),\mathbb{Z}_2)$ introduces the grading.
While the Eilenberg-Moore spectral sequence is a fair bit more general, we will focus on a particular application, which was expounded
in~\cite{rothenberg1965cohomology}. Rothenberg and Steenrod proved that there is a convergent spectral sequence:
\begin{equation}
\label{EMSS2}
    E_2 = \mathrm{Cotor}_{H^*(G,\mathbb{Z}_2)}(\mathbb{Z}_2, \mathbb{Z}_2) \Longrightarrow H^*(BG,\mathbb{Z}_2)\,.
\end{equation}
The underlying idea here is that we can build an injective resolution for a classifying space $BG$ through the group $G$ itself, by constructing $G$-invariant closed subspaces of $BG$.
We will close our short introduction to the Eilenberg-Moore spectral sequence here and point to chapter 7 and 8 in~\cite{McCleary_2000} for a much deeper discussion of the topic.

\section{The calculation of \texorpdfstring{$\Omega_n^{Spin}(BSs(32))$}{spin cobordism of BSs(32)}}
\label{sec:calculation_spin_bordism}
After our short introduction to the main mathematical tools we are going to use we will now work through the computation of $\Omega_n^{Spin}(BSs(32))$. 
This section is divided into the two main steps of the calculation: Determining $H^n(BSs(4n),\mathbb{Z}_2)$ via the Eilenberg-Moore spectral sequence and then followed up by computing $\Omega_n^{Spin}(BSs(32))$.
\subsection{Determining \texorpdfstring{$H^n(BSs(4n),\mathbb{Z}_2)$}{Hn(BSs(4n),Z/2)}}

To access the second page of the Adams spectral sequence we will now partially follow and extend the calculation of $H^{* \leq 11}(BSs(n),\mathbb{Z}_2)$ demonstrated by~\cite{TachikawaMO}. In particular the computation exploited the fact that while we do not know $H^n(BSs(n),\mathbb{Z}_2)$, we actually do know $H^n(Ss(n),\mathbb{Z}_2)$~\cite{Ishitoya:1976pf}.  
The aforementioned calculation consists of two subsequent spectral sequences. First, the May spectral sequence~\cite{may1964cohomology} sets up the next spectral sequence
\begin{equation}
    \mathrm{Cotor}_{A'}(\mathbb{Z}_2, \mathbb{Z}_2) \Longrightarrow \mathrm{Cotor}_{H^*(Ss(32), \mathbb{Z}_2)}(\mathbb{Z}_2, \mathbb{Z}_2)\,,
\end{equation}
where $A'$ is is a Hopf algebra such that it is isomorphic as an algebra with $A'$ such that every generator is primitive.
Then, the result is fed as the second page into the Eilenberg-Moore spectral sequence \eqref{EilenbergMoore}~\cite{moore1959algebre, rothenberg1965cohomology, eilenberg1966homology}
\begin{equation}
    E_2 = \mathrm{Cotor}_{H^*(Ss(32), \mathbb{Z}_2)}(\mathbb{Z}_2, \mathbb{Z}_2) \Longrightarrow H^*(BSs(n),\mathbb{Z}_2)\,.
\end{equation}
However, we will follow a different route to determine $\mathrm{Cotor}_{H^*(Ss(32),\mathbb{Z}_2)}(\mathbb{Z}_2, \mathbb{Z}_2)$, namely the ``twisted tensor product" method of~\cite{Kono75}.

We start with $H^*(Ss(32), \mathbb{Z}_2)$ as a Hopf algebra as determined in~\cite{Ishitoya:1976pf}.
Up to degree $n = 15$ $H^n(Ss(n),\mathbb{Z}_2)$ the authors showed that as an algebra it is isomorphic to
\begin{equation}
    \Delta(w_{3}, w_{5}, w_{6}, w_{7}, w_{9}, w_{10}, w_{11}, w_{12}, w_{13}, w_{14}) \otimes \mathbb{Z}_2[\bar{v}]\,.
\end{equation}
Importantly, there are a couple nontrivial coproducts in the degree range we are interested in:
\begin{equation}
    \bar{\psi}(w_7) = \bar{v} \otimes w_6 + \bar{v}^2 \otimes w_5 + \bar{v}^4 \otimes w_3\,,
\end{equation}
\begin{equation}
    \bar{\psi}(w_{11}) = \bar{v} \otimes w_{10} + \bar{v}^2 \otimes w_9 + \bar{v}^8 \otimes w_3\,,
\end{equation}
\begin{equation}
    \bar{\psi}(w_{13}) = \bar{v} \otimes w_{12} + \bar{v}^4 \otimes w_9 + \bar{v}^8 \otimes w_5\,,
\end{equation}
\begin{equation}
    \bar{\psi}(w_{14}) = \bar{v}^2 \otimes w_{12} + \bar{v}^4 \otimes w_{10} + \bar{v}^8 \otimes w_6\,.
\end{equation}

As our first step of determining $\mathrm{Cotor}_{H^*(Ss(32), \mathbb{Z}_2)}(\mathbb{Z}_2, \mathbb{Z}_2)$, we define the $\mathbb{Z}_2$-submodule $L$ of $H^*(Ss(n),\mathbb{Z}_2)$ generated by 
\begin{equation}
    \{\bar{v}, \bar{v}^2, \bar{v}^4, \bar{v}^8, w_{3}, w_{5}, w_{6}, w_{7}, w_{9}, w_{10}, w_{11}, w_{12}, w_{13}, w_{14}\}\,.
\end{equation}
Then we have the projection $\theta: H^*(Ss(32),\mathbb{Z}_2) \to L$, the inclusion $\iota: L \to H^*(Ss(32),\mathbb{Z}_2)$ and the suspension $s$ uplifting the elements of $L$ to:
\begin{equation}
    sL = \{a_2, a_3, a_5, a_9, b_4, b_5, b_{7}, c_{8}, b_{10}, b_{11}, c_{12}, b_{13}, c_{14}, c_{15}\}\,.
\end{equation}
Now we extend the maps $\theta$ and $\iota$ to $\bar{\theta} = s \circ \theta: H^*(Ss(32),\mathbb{Z}_2) \to sL$ and $\bar{\iota} = \iota \circ s^{-1}$.
Next, we construct $\bar{X}$ as $\bar{X} \coloneqq T(sL)/I$, where $T(sL)$ is the tensor algebra with the natural product $\psi$ and $I$ is the two-sided ideal of $T(sL)$ generated by $Im(\psi \circ (\bar{\theta} \otimes \bar{\theta}) \circ \phi) \circ Ker(\bar{\theta})$.
Consequently, $\bar{X}$ is given by
\begin{equation}
    \bar{X} = \mathbb{Z}_2[a_{i}, b_{j}, c_{k}]/I\,,
\end{equation}
where $I$  is generated by
\begin{align}
    &[a_i, a_j], \,[b_i, b_j], \,[a_i, b_j], \,[b_i, c_j], \,[a_i,c_j] \text{ for } (i,j) \in \{(5,12),(3,14),(2,15)\} \nonumber\\
    &[a_2,c_8] + a_3 b_7, \,[a_3,c_8] + a_5 b_6, \, [a_5,c_8] + a_9 b_4, \nonumber\\
    &[a_2,c_{12}] + a_3 b_{11}, \,[a_3,c_{12}] + a_5 b_{10}, \\
    &[a_2,c_{14}] + a_3 b_{13}, \,[a_5,c_{14}] + a_9 b_{10}, \nonumber\\
    &[a_3,c_{15}] + a_5 b_{13}, \,[a_5,c_{15}] + a_9 b_{11},\nonumber
\end{align}
[,] denotes the commutator.
From here we construct our twisted tensor product. Now, we define a differential
$\bar{d}$ as a map $\bar{d} = \psi \circ (\bar{\theta} \otimes \bar{\theta}) \circ \phi\, \bar{\iota}: sL \to T(sL)$ that is uniquely extended to $\bar{d}: T(sL) \to T(sL)$ with $\bar{d}(I) \subset I$, such that $\bar{X}$ becomes a differential algebra. 

Consequently, following~\cite{Kono76} we construct through the triviality of $\bar{d} \circ \bar{\theta} + \psi \circ (\bar{\theta} \otimes \bar{\theta}) \circ \phi = 0$ a twisted tensor product $W = H^*(Ss(32),\mathbb{Z}_2) \otimes \bar{X}$ with respect to $\bar{\theta}$. $W$ is a differential $H^*(Ss(32),\mathbb{Z}_2)$-comodule with the differential:
\begin{equation}
    d = 1 \otimes \bar{d} + (1 \otimes \psi) \circ (1 \otimes \bar{\theta} \otimes 1) \circ (\phi \otimes 1)
\end{equation}
with the product $\psi$ over $T(sL)$ and $\phi$ the product over $H^*(Ss(32),\mathbb{Z}_2)$.
The differential acts in the following way on the elements:
\begin{align}
    dw_i &= a_{i+1} \text{ for all } i \text{ except for }  i = \{7, 11, 13, 14\}\,,\nonumber\\ 
    d\bar{v}^{j} &= b_{j+1} \text{ for } j = \{1, 2, 4, 8\}\,, \nonumber\\
    dw_7 &= c_8 + \bar{v} \otimes a_7 + \bar{v}^2 \otimes a_6 + \bar{v}^4 \otimes a_4\,, \nonumber\\
    dw_{11} &= c_{12} + \bar{v} \otimes a_{11} + \bar{v}^2 \otimes a_{10} + \bar{v}^8 \otimes a_4\,, \\
    dw_{13} &= c_{14} + \bar{v} \otimes a_{13} + \bar{v}^4 \otimes a_{10} + \bar{v}^8 \otimes a_6\,, \nonumber\\
    dw_{14} &= c_{15} + \bar{v}^2 \otimes a_{13} + \bar{v}^4 \otimes a_{11} + \bar{v}^8 \otimes a_7\,. \nonumber
\end{align}
As a result, we get the following action on $\{a_i, b_j, c_k\}$:
\begin{align}
    da_i = \bar{d}a_i &= 0\,,\nonumber\\ 
    db_j = \bar{d}b_j &= 0\,,\nonumber\\
    dc_8 = \bar{d}c_8 &= b_2 a_7 + b_3 a_6 + b_5 a_4\,, \nonumber\\
    dc_{12} = \bar{d}c_{12} &= b_2 a_{11} + b_3 a_{10} + b_9 a_4\,, \\
    dc_{14} = \bar{d}c_{14} &= b_2 a_{13} + b_5 a_{10} + b_9 a_6\,, \nonumber\\
    dc_{15} = \bar{d}c_{15} &= b_3 a_{13} + b_5 a_{11} + b_9 a_7\,. \nonumber
\end{align}
Consequently, in accordance to the procedure outlined by~\cite{Kono75, Kono76} we define weights in $W$, which in our case is 1 for the ``pairs" with respect to the suspension $(w_7,\,c_8),\,(w_{11},\, c_{12}),\, (w_{13},\, c_{14}),\,(w_{14},\, c_{15})$ and zero for all the other pairs.

This allows us to define a filtration with respect to these weights: 
\begin{equation}
   F_r = \{x \in W \text{ with weight } \leq r\}\,. 
\end{equation}
Then we have essentially achieved the same as a spectral sequence and we define $E_{\infty}(W) = \sum_r F_r/F_{r-1}$. The point is that since $\bar{d}(F_r) \subset F_{r}$ the homology groups of $E_{\infty}(W)$ vanish, i.e.\ $E_{\infty}(W)$ is acyclic. Thus, $W$ itself is acyclic, as well. Then $W = H^*(Ss(32),\mathbb{Z}_2) \otimes \bar{X}$ is an acyclic injective comodule resolution of $\mathbb{Z}_2$ over $H^*(Ss(32),\mathbb{Z}_2)$, which is reminiscent of the definition we gave for the $\mathrm{Cotor}$ functor.
The point is that the cohomology of $\bar{X}$ together with the map $\bar{d}$ gives us the $\mathrm{Cotor}$ functor:
\begin{equation}
    \mathrm{Cotor}^{H^*(Ss(32),\mathbb{Z}_2)}(\mathbb{Z}_2, \mathbb{Z}_2) \cong H(\bar{X}:\,\bar{d}) = ker(\bar{d})/im(\bar{d})
\end{equation}
in the notation of\cite{Kono76}.
Therefore, we now get:
\begin{align}
\label{Cotor_HSs(n)}
   \mathrm{Cotor}^{H^*(Ss(32),\mathbb{Z}_2)}(\mathbb{Z}_2, \mathbb{Z}_2) \cong \,&\mathbb{Z}_2\left[\bar{x}_2, \bar{x}_3, \bar{x}_5, \bar{x}_9, \bar{y}_4, \bar{y}_6, \bar{y}_7, \bar{y}_{10}, \bar{y}_{11}, \bar{y}_{13}\right]/ \\
    &\big(\bar{x}_2 \bar{y}_7 + \bar{x}_3 \bar{y}_6 + \bar{x}_5 \bar{y}_4, \,\bar{x}_2 \bar{y}_{11} + \bar{x}_3 \bar{y}_{10} + \bar{x}_9 \bar{y}_4, \,\nonumber \\
    &\bar{x}_2 \bar{y}_{13} + \bar{x}_5 \bar{y}_{10} + \bar{x}_9 \bar{y}_6,\,
     \bar{x}_3 \bar{y}_{13} + \bar{x}_5 \bar{y}_{11} + \bar{x}_9 \bar{y}_7\big)\,. \nonumber
\end{align}
With $\mathrm{Cotor}^{H^*(Ss(32),\mathbb{Z}_2)}(\mathbb{Z}_2, \mathbb{Z}_2)$ determined we can proceed with the main goal of this section resolving the Eilenberg-Moore spectral sequence \eqref{EilenbergMoore}.
Actually, we can make the task a lot easier as all the generators can be written as Steenrod squares acting on either $\bar{x}_2$ or $\bar{y}_4$:
\begin{align}
    &\bar{x}_3 = Sq^1 \bar{x}_2, \,\bar{x}_5 = Sq^2 Sq^1 \bar{x}_2, \,\bar{x}_9 = Sq^4 Sq^2 Sq^1 \bar{x}_2, \nonumber \\
    &\bar{y}_6 = Sq^2 \bar{y}_4, \, \bar{y}_7 = Sq^3 \bar{y}_4, \, \bar{y}_{10} = Sq^4 Sq^2 \bar{y}_4, \\ 
    &\bar{y}_{11} = Sq^5 Sq^2 \bar{y}_4, \, \bar{y}_{13} = Sq^6 Sq^3 \bar{y}_4. \nonumber
\end{align}
Since $\bar{x}_i \in E_2^{1, i-1}$ and $y_j \in E_2^{1, j-1}$ we can see that there are no nontrivial differentials acting on $\bar{x}_2$ and $\bar{y}_4$:
\begin{align}
    d_r(\bar{x}_2): E_{r}^{1,1} \to E_{r}^{1+r,1-(r-1)} = 0\,, \\
    d_r(\bar{y}_4): E_{r}^{1,3} \to E_{r}^{1+r,3-(r-1)} = 0\,.
\end{align}
As a consequence all of the elements up to the degree we are studying are actually permanent cycles as both $\bar{x}_2$ and $\bar{y}_4$ are and subsequently every element related by a cohomology operation (in this case Steenrod squares) to a permanent cycle. Therefore, our spectral sequence collapses.
Now, proceeding to $H^*(BSs(n),\mathbb{Z}_2)$, let's define the elements. Analogous to~\cite{TachikawaMO} we choose $y_4$ instead of $y_4 + x_2^2$ as the representative of $\bar{y}_4$ besides $x_2$ as the representative of $\bar{x}_2$, the other representatives
in the same order as in \eqref{Cotor_HSs(n)}:
\begin{align}
    &x_3 = Sq^1 x_2, \,x_5 = Sq^2 x_3, \,x_9 = Sq^4 x_5, \nonumber \\
    &y_6 = Sq^2 y_4, \, y_7 = Sq^1 y_6, \, y_{10} = Sq^4 y_6, \\ 
    &y_{11} = Sq^1 y_{10}, \, y_{13} = Sq^2 y_{11}. \nonumber
\end{align}
At this point let us denote the action of the Steenrod squares, which is going to be very important for the Adams spectral sequence: \\
\begin{center}
\label{HBSs(4n)_Sqkaction}
\captionsetup{type=table}
\begin{tabular}{c|c|c|c|c}\toprule
    &$Sq^1$ & $Sq^2$ & $Sq^3$ & $Sq^4$ \\\bottomrule 
    $x_2$ &$x_3$ & $x_2^2$ &  &  \\
    $x_3$ & / &$x_5$ & $x_3^2$ &  \\
    $x_5$ & $x_3^2$ & / & / & $x_9$ \\
    $x_9$ & $x_5^2$ & / & / & /  \\\bottomrule
    $y_4$ & / & $y_6$ & $y_7$ & $y_4^2$ \\
    $y_6$ & $y_7$ & / & / & $y_{10}$ \\
    $y_7$ & / & / & / & $y_{11}$ \\
    $y_{10}$ & $y_{11}$ & / & / & / \\
    $y_{11}$ & / & $y_{13}$ & / & / \\
    $y_{13}$ & / & / &  &  \\\bottomrule
\end{tabular}
\captionof{table}{Elements of $H^n(BSs(4n), \mathbb{Z}_2)$ and their transformation under Steenrod squares.}
\end{center}
As alluded to before obtaining the higher relations $r_i$ is a lot simpler as they are related to $r_1$ by Steenrod operations. The first relation is given by:
\begin{equation}
\label{first_relation}
    r_1 = x_2 y_7 + x_3 y_6 + \tilde{x}_5 y_4\,,
\end{equation}
where we defined $\tilde{x}_5 = x_5 + x_2 x_3$.
Now, we can bootstrap the next relations, while crosschecking that Steenrod squares acting on the relations actually vanish. The first few Steenrod operations acting on $r_1$ are pretty simple:
\begin{align}
    &Sq^1(r_1) = x_3 y_7 + x_3 y_7 = 0\,, \nonumber \\
    &Sq^2(r_1) = x_2^2 y_7 + x_5 y_6 + x_2 \tilde{x}_5 y_4 + \tilde{x}_5 y_6 = x_2 r_1 = 0\,, \\
    &Sq^3(r_1) = x_5 y_7 + x_3^2 y_6 + \tilde{x}_5 y_7 + (x_3 x_5 + x_2 x_3^2) y_4 = x_3 r_1 = 0\,. \nonumber
\end{align}
Finally with $Sq^4$ we reach the next relation $r_2$:
\begin{equation}
\label{second_relation}
    r_2 = Sq^4(r_1) = x_2 y_{11} + x_3 y_{10} + \tilde{x}_5 y_4^2 + x_3^2 y_7 + x_2 \tilde{x}_5 y_6 + \tilde{x}_9 y_4\,,
\end{equation}
where we defined $\tilde{x}_9 = x_9 + x_2^2 x_5 + x_3^3$. This definition is very convenient, since we have $Sq^4(\tilde{x}_5) = \tilde{x}_9$ and $Sq^1(\tilde{x}_9) = \tilde{x}_5^2$.
So let's look at the next couple of Steenrod squares building onto $r_2$, as well:
\begin{align}
\label{third_relation}
    &Sq^1(r_2) = x_3 y_{11} + x_3 y_{11} + \tilde{x}_5 r_1 = 0\,, \nonumber \\
    &r_3 = Sq^2(r_2) = x_2 y_{13} + \tilde{x}_5 y_{10} + x_3 x_5 y_7 + x_2^2 \tilde{x}_5 y_6 + \tilde{x}_9 y_6 + x_2 \tilde{x}_9 y_4\,.
\end{align}
Subsequently, we get $r_4$ from $r_3$ (or as $Sq^3(r_2)$ from $r_2$):
\begin{align}
\label{fourth_relation}
    r_4 = Sq^1(r_3) = x_3 y_{13} &+ \tilde{x}_5 y_{11} + x_2^2 \tilde{x}_5 y_7 + x_3^3 y_7 + \tilde{x}_9 y_7  \\
    &+ \tilde{x}_5^2 y_6 + x_3 \tilde{x}_9 y_4 + x_2 \tilde{x}_5^2 y_4 \,. \nonumber
\end{align}
As a final consistency check we calculate $Sq^1$ and $Sq^2$ of $r_4$ to check that both actually vanish\footnote{There is an additional relation $r_5$ in degree 17, but we don't hit it with $Sq^1(r_4)$.}.
\begin{align}
    &Sq^1(r_4) = \tilde{x}_5^2 y_7 + \tilde{x}_5^2 y_7 + x_3 \tilde{x}_5^2 y_4 + x_3 \tilde{x}_5^2 y_4 = 0\,, \nonumber \\
    &Sq^2(r_4) = \tilde{x}_9 r_1 + \tilde{x}_5^2(x_2 y_6 + x_2^2 y_4 + x_2 y_6 + x_2^2 y_4) = 0\,.
\end{align}
Now, with the action of the Steenrod squares set up and the relations determined we can go ahead and identify the $\mathcal{A}_{1}$-module structure of $H^*(BSs(4n), \mathbb{Z}_2)$. First, from \eqref{HBSs(4n)_Sqkaction} we recognize that the elements of $H^*(BSs(4n), \mathbb{Z}_2)$ split into two different parts 
a ``x-part" consisting of $x_2$, \dots, $x_9$ and a ``y-part" composed of $y_4$, \dots, $y_{13}$ as the Steenrod squares never transforms them into each other.
The third part, a mixed part, of course is comprised of the combination of x- and y-elements. This is where the relations come into play.

Notice that the ``x-part" is up to the degrees we are working with isomorphic to the algebra of $H^*(B^2\mathbb{Z}_2, \mathbb{Z}_2)$ and the ``y-part" can be identified with $H^*(BE_8, \mathbb{Z}_2)$.
The ``x-part" can be understood as coming from the fibration
\begin{equation}
    BSpin(4n) \to BSs(4n) \to B^2\mathbb{Z}_2\,.
\end{equation}
From a string theory perspective the ``y-part" has a very natural interpretation as the lead actor of the T-Duality between the two supersymmetric heterotic string theories with gauge groups $(E_8 \times E_8) \rtimes \mathbb{Z}_2$ and $Ss(32)$.

Interestingly, the coproducts in $H^*(Ss(4n), \mathbb{Z}_2)$ precisely cause the ``y-part" in the range relevant to both string theories to change from being identical to $H^*(BSpin(4n), \mathbb{Z}_2)$ to being identical to $H^*(BE_8, \mathbb{Z}_2)$. 

Let us also point out the close relation to $H^*(BSO(4n), \mathbb{Z}_2)$ as sometimes the gauge group of type I/HO string theory is wrongfully identified as SO(32). Here, additionally to removing the coproducts we would also need to couple the ``x-" and the ``y-part", such that after renaming the ``x-elements": $x_i \to y_i$ ($i \in 2,3,5,9$) we would have the following action of $Sq^{i}$, $i \in 1,2$:
\begin{align}
    Sq^1(y_i) &= (i-1)\, y_{i+1}\,, \\
    Sq^2(y_i) &= \binom{i-1}{2}\, y_{i+2} + y_2 y_i\,.
\end{align}
We leave a string theoretic interpretation of the coproducts and the consequent relations in $H^*(BSs(4n), \mathbb{Z}_2)$ to future work.

Let's start with looking at the $\mathcal{A}_1$-module structure of the ``x-part". Since $\mathbb{Z}_2$ is a discrete group, $B^2\mathbb{Z}_2$ is nothing else than the Eilenberg-Maclane space $K(\mathbb{Z}_2, 2)$.
Up to degree $n=40$ the ko-homology $ko_n(K(\mathbb{Z}_2, 2))$, which is more than sufficient for the string theoretic applications we have in mind, has been calculated in~\cite{wilson1973new}. As already hinted upon before we will use some particular combinations of $x_2$, $x_3$, $x_5$ and $x_9$. Namely, we will use 
\begin{align*}
    &\tilde{x}_5 = x_5 + x_2 x_3\,, \\
    &\tilde{x}_9 = x_9 + x_2^2 x_5 + x_3^3\,, \\
    &\tilde{x}'_9 = x_2^3 x_3 + x_2^2 x_5 + x_3^3\,, \\
    &\tilde{x}_{11} = x_2 x_9 + x_3^2 x_5 + x_2 x_3^3
\end{align*}
matching the definitions of~\cite{wilson1973new}. With the action of $Sq^1$ and $Sq^2$ on the elements of the ``x-part" we get the following $\mathcal{A}_1$-module structure: \\

\begin{sseqdata}[name=Sqxpart, Adams grading, classes = fill, no axes, xscale = 0.6, yscale = 0.6]
\class[MPP_blue_light, "x_2" { right, xshift = -0.16cm, black }, font = \footnotesize](0,0)
\class[MPP_blue_light](0,1)
\class[MPP_blue_light](0,2)
\class[MPP_blue_light](0,3)
\class[MPP_blue_light](0,4)
\structline[MPP_blue_light](0,0)(0,1)
\structline[MPP_blue_light, bend left = 30](0,0)(0,2)
\structline[MPP_blue_light, bend left = 30](0,2)(0,4)
\structline[MPP_blue_light, bend right = 30](0,1)(0,3)
\structline[MPP_blue_light](0,3)(0,4)

\class[MPP_blue_light, "\tilde{x}_5" { right, xshift = -0.16cm, black }, font = \footnotesize](1,3)
\class[MPP_blue_light](1,5)
\class[MPP_blue_light](1,6)
\class[MPP_blue_light](1,8)
\structline[MPP_blue_light, bend left = 30](1,3)(1,5)
\structline[MPP_blue_light](1,5)(1,6)
\structline[ MPP_blue_light, bend left = 30](1,6)(1,8)


\foreach \n in {4} \foreach \m in {3}{
\class[MPP_blue_light, "x_2^3" { right, xshift = -0.16cm, black }, font = \footnotesize](\m,\n)
\class[MPP_blue_light](\m,\n + 1)
\class[MPP_blue_light](\m,\n + 2)
\class[MPP_blue_light](\m,\n + 3)
\class[MPP_blue_light](\m + 2,\n + 3)
\class[MPP_blue_light](\m + 2,\n + 4)
\class[MPP_blue_light](\m + 2,\n + 5)
\class[MPP_blue_light](\m + 2,\n + 6)

\structline[MPP_blue_light](\m,\n)(\m,\n + 1)
\structline[MPP_blue_light, bend left = 30](\m,\n)(\m,\n + 2)
\structline[MPP_blue_light](\m,\n + 2)(\m,\n + 3)
\structline[MPP_blue_light, in = -160, out = 20](\m,\n + 1)(\m + 2,\n + 3)
\structline[MPP_blue_light, in = -160, out = 20](\m,\n + 2)(\m + 2,\n + 4)
\structline[MPP_blue_light](\m + 2,\n + 3)(\m + 2,\n + 4)
\structline[MPP_blue_light, in = -160, out = 20](\m,\n + 3)(\m + 2,\n + 5)
\structline[MPP_blue_light](\m + 2,\n + 5)(\m + 2,\n + 6)
\structline[MPP_blue_light, bend right = 30](\m + 2,\n + 4)(\m + 2,\n + 6)
}

\foreach \n in {8} \foreach \m in {9}{
\class[MPP_blue_light, "x_2 x_3 x_5" { right, xshift = -0.16cm, black }, font = \footnotesize](\m,\n)
\class[MPP_blue_light](\m,\n + 1)
\class[MPP_blue_light](\m,\n + 2)
\class[MPP_blue_light](\m,\n + 3)
\class[MPP_blue_light](\m + 2,\n + 3)
\class[MPP_blue_light](\m + 2,\n + 4)
\class[MPP_blue_light](\m + 2,\n + 5)
\class[MPP_blue_light](\m + 2,\n + 6)

\structline[MPP_blue_light](\m,\n)(\m,\n + 1)
\structline[MPP_blue_light, bend left = 30](\m,\n)(\m,\n + 2)
\structline[MPP_blue_light](\m,\n + 2)(\m,\n + 3)
\structline[MPP_blue_light, in = -160, out = 20](\m,\n + 1)(\m + 2,\n + 3)
\structline[MPP_blue_light, in = -160, out = 20](\m,\n + 2)(\m + 2,\n + 4)
\structline[MPP_blue_light](\m + 2,\n + 3)(\m + 2,\n + 4)
\structline[MPP_blue_light, in = -160, out = 20](\m,\n + 3)(\m + 2,\n + 5)
\structline[MPP_blue_light](\m + 2,\n + 5)(\m + 2,\n + 6)
\structline[MPP_blue_light, bend right = 30](\m + 2,\n + 4)(\m + 2,\n + 6)
}

\foreach \n in {10} \foreach \m in {13}{
\class[MPP_blue_light, "x_2^3 x_3^2" { right, xshift = -0.16cm, black }, font = \footnotesize](\m,\n)
\class[MPP_blue_light](\m,\n + 1)
\class[MPP_blue_light](\m,\n + 2)
\class[MPP_blue_light](\m,\n + 3)
\class[MPP_blue_light](\m + 2,\n + 3)
\class[MPP_blue_light](\m + 2,\n + 4)
\class[MPP_blue_light](\m + 2,\n + 5)
\class[MPP_blue_light](\m + 2,\n + 6)

\structline[MPP_blue_light](\m,\n)(\m,\n + 1)
\structline[MPP_blue_light, bend left = 30](\m,\n)(\m,\n + 2)
\structline[MPP_blue_light](\m,\n + 2)(\m,\n + 3)
\structline[MPP_blue_light, in = -160, out = 20](\m,\n + 1)(\m + 2,\n + 3)
\structline[MPP_blue_light, in = -160, out = 20](\m,\n + 2)(\m + 2,\n + 4)
\structline[MPP_blue_light](\m + 2,\n + 3)(\m + 2,\n + 4)
\structline[MPP_blue_light, in = -160, out = 20](\m,\n + 3)(\m + 2,\n + 5)
\structline[MPP_blue_light](\m + 2,\n + 5)(\m + 2,\n + 6)
\structline[MPP_blue_light, bend right = 30](\m + 2,\n + 4)(\m + 2,\n + 6)
}

\class[MPP_blue_light, "x_2^4" { below, yshift = 0.1cm, black }, font = \footnotesize](6,6)

\class[MPP_blue_light, "\tilde{x}_9" { below, yshift = 0.15cm, xshift = 0.25cm, black }, font = \footnotesize](1,7)
\structline[MPP_blue_light, dashed](1,7)(1,8)

\class[MPP_blue_light, "\tilde{x}'_9" { below, yshift = 0.1cm, black }, font = \footnotesize](7,7)
\class[MPP_blue_light](7,9)
\class[MPP_blue_light](7,10)
\class[MPP_blue_light](7,12)
\structline[MPP_blue_light, bend left = 30](7,7)(7,9)
\structline[MPP_blue_light](7,9)(7,10)
\structline[MPP_blue_light, bend left = 30](7,10)(7,12)

\class[MPP_blue_light, "x_2 \cdot x_2^4 " { below, yshift = 0.1cm, black }, font = \footnotesize](12,8)
\class[MPP_blue_light](12,9)
\class[MPP_blue_light](12,10)
\class[MPP_blue_light](12,11)
\class[MPP_blue_light](12,12)
\structline[MPP_blue_light](12,8)(12,9)
\structline[MPP_blue_light, bend left = 30](12,8)(12,10)
\structline[MPP_blue_light, bend left = 30](12,10)(12,12)
\structline[MPP_blue_light, bend right = 30](12,9)(12,11)
\structline[MPP_blue_light](12,11)(12,12)

\class[MPP_blue_light, "\tilde{x}_{11}" { below, yshift = 0.1cm, black }, font = \footnotesize](8,9)
\class[MPP_blue_light](8,10)
\class[MPP_blue_light](8,11)
\class[MPP_blue_light](8,12)
\class[MPP_blue_light](8,13)
\structline[MPP_blue_light](8,9)(8,10)
\structline[MPP_blue_light, bend left = 30](8,9)(8,11)
\structline[MPP_blue_light, bend left = 30](8,11)(8,13)
\structline[MPP_blue_light, bend right = 30](8,10)(8,12)
\structline[MPP_blue_light](8,12)(8,13)
\structline[MPP_blue_light, dashed](8,11)(7,12)

\class[MPP_blue_light, "\tilde{x}_5 \cdot x_2^4" { below, yshift = 0.1cm, black }, font = \footnotesize](16,11)
\class[MPP_blue_light](16,13)
\class[MPP_blue_light](16,14)
\class[MPP_blue_light](16,16)
\structline[MPP_blue_light, bend left = 30](16,11)(16,13)
\structline[MPP_blue_light](16,13)(16,14)
\structline[MPP_blue_light, bend left = 30](16,14)(16,16)

\class[MPP_blue_light, "\tilde{x}_9 \cdot x_2^4" { below, yshift = 0.15cm, xshift = 0.55cm, black }, font = \footnotesize](16,15)
\structline[MPP_blue_light, dashed](16,15)(16,16)

\class[white, font = \footnotesize](17,15)

\end{sseqdata}

\begin{figure}[H]
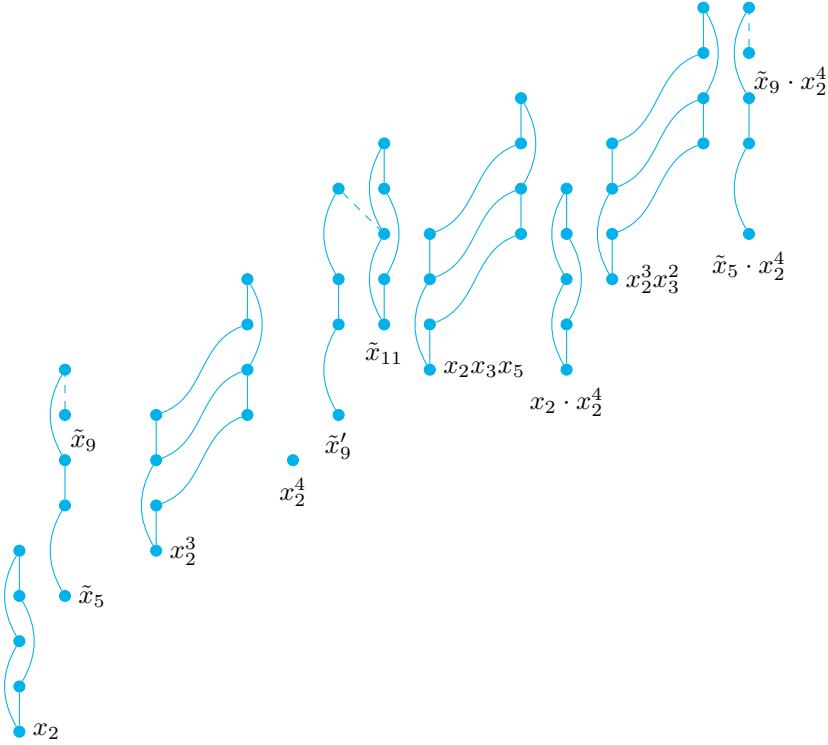

 \centering
 \printpage[ name = Sqxpart]
 \caption{$\mathcal{A}_1$-module structure - x-part}
 \label{fig:Sqxpart}
\end{figure}
To match the notation of \cite{wilson1973new} we included the $Sq^1$ connecting $\tilde{x}_9$ with \\
$Sq^2\,Sq^1\,Sq^2(\tilde{x}_5)$ as a dashed line in the figure above. For the subsequent Adams charts this $Sq^1$ does not makes a difference and we can effectively treat $\mathbb{Z}_2[\tilde{x}_9]$ and $M_0[\tilde{x}_5]$ as separate modules following \cite{wilson1973new}.
In our degree range there is another non-trivial extension, namely a $Sq^3 = Sq^1 Sq^2$ between $\tilde{x}_{11}$ and $Sq^2\,Sq^1\,Sq^2(\tilde{x}_9')$ \cite{wilson1973new}, which we indicate with a dashed line. Once more, we can treat $M_0[\tilde{x}_9']$ and $J[\tilde{x}_{11}]$ as different modules in the Adams charts.

Before we discuss setting up the Adams spectral sequence by filling the second page, we'll look at the $\mathcal{A}_1$-module structure of the ``y-part" and the mixed ``x-y-part".
Whereas the structure of the ``y-part" is simple compared to the other parts as can be seen below, the mixed ``x-y-part" has a highly involved $\mathcal{A}_1$-module structure. \\
\begin{sseqdata}[name = Sqypart, classes = fill, no axes , xscale = 0.6, yscale = 0.6]

\class[ MPP_orange, "y_4" { right, xshift = -0.16cm, black }, font = \footnotesize](0,0)
\class[MPP_orange](0,2)
\class[MPP_orange](0,3)
\structline[MPP_orange, bend left = 30](0,0)(0,2)
\structline[MPP_orange](0,2)(0,3)

\class[ MPP_orange, "y_4^2" { right, xshift = -0.16cm, black }, font = \footnotesize](1,4)

\class[ MPP_orange, "y_4 y_6" { below, yshift = 0.1cm, black }, font = \footnotesize](2,6)
\class[MPP_orange](2,7)
\class[MPP_orange](2,8)
\class[MPP_orange](2,9)
\class[MPP_orange](2,10)
\structline[MPP_orange](2,6)(2,7)
\structline[MPP_orange, bend left = 30](2,6)(2,8)
\structline[MPP_orange, bend left = 30](2,8)(2,10)
\structline[MPP_orange, bend right = 30](2,7)(2,9)
\structline[MPP_orange](2,9)(2,10)

\class[ MPP_orange, "y_{10}" { right, xshift = -0.16cm, black }, font = \footnotesize](3,6)
\class[MPP_orange](3,7)
\class[MPP_orange](3,9)
\structline[MPP_orange, bend left = 30](3,7)(3,9)
\structline[MPP_orange](3,6)(3,7)

\class[ MPP_orange, "y_4^3" { right, xshift = -0.16cm, black }, font = \footnotesize](4,8)
\class[MPP_orange](4,10)
\class[MPP_orange](4,11)
\structline[MPP_orange, bend left = 30](4,8)(4,10)
\structline[MPP_orange](4,10)(4,11)

\end{sseqdata}

\begin{figure}[H]
    \centering
    \printpage[ name = Sqypart, page = 2]
    \caption{$\mathcal{A}_{1}$-module structure - y-part}
    \label{fig:Sqypart}
\end{figure}

After carefully incorporating the relations $r_i$, $i \in 1, \dots, 4$ \eqref{first_relation}-\eqref{fourth_relation} we finally arrive at the following structure: \\
\begin{sseqdata}[name=Sqxypart, classes = {fill}, no axes , xscale = 0.42, yscale = 0.42]


\class[ MPP_green, "x_2^2 y_4" { below, yshift = 0.1cm, black }, font = \scriptsize](3,2)
\class[MPP_green](3,3)
\class[MPP_green](3,4)
\class[MPP_green](5,4)
\class[MPP_green](5,5)
\class[MPP_green](5,6)
\class[MPP_green](5,7)

\structline[MPP_green](3,3)(3,4)
\structline[MPP_green, in = -160, out = 20](3,2)(5,4)
\structline[MPP_green, in = -160, out = 20](3,3)(5,5)
\structline[MPP_green](5,4)(5,5)
\structline[MPP_green, in = -160, out = 20](3,4)(5,6)
\structline[MPP_green](5,6)(5,7)
\structline[MPP_green, bend right = 30](5,5)(5,7)


\class[ MPP_green, "x_2 y_4^2" { below, yshift = 0.1cm, black }, font = \scriptsize](9,4)
\class[MPP_green](9,5)
\class[MPP_green](9,6)
\class[MPP_green](9,7)
\class[MPP_green](9,8)
\structline[MPP_green](9,4)(9,5)
\structline[MPP_green, bend left = 30](9,4)(9,6)
\structline[MPP_green, bend left = 30](9,6)(9,8)
\structline[MPP_green, bend right = 30](9,5)(9,7)
\structline[MPP_green](9,7)(9,8)

\class[ MPP_green, "x_2^4 y_4" { below, yshift = 0.04cm, black }, font = \scriptsize](16,6)
\class[MPP_green](16,8)
\class[MPP_green](16,9)
\structline[MPP_green, bend left = 30](16,6)(16,8)
\structline[MPP_green](16,8)(16,9)


\class[MPP_green, "x_2 y_{10}" { right, xshift = -0.16cm, black }, font = \scriptsize](17,6)
\class[MPP_green](17,7)
\class[MPP_green](17,8)
\class[MPP_green](17,9)
\class[MPP_green](19,9)
\class[MPP_green](19,10)
\class[MPP_green](19,11)
\structline[MPP_green](17,6)(17,7)
\structline[MPP_green, bend left = 30](17,6)(17,8)
\structline[MPP_green](17,8)(17,9)
\structline[MPP_green](19,9)(19,10)
\structline[MPP_green, in = -160, out = 20](17,7)(19,9)
\structline[MPP_green, in = -160, out = 20](17,8)(19,10)
\structline[MPP_green, in = -160, out = 20](17,9)(19,11)

\class[ MPP_green, "\tilde{x}_9' y_4" { below, yshift = 0.1cm, black }, font = \scriptsize](29,7)
\class[MPP_green](29,9)
\class[MPP_green](29,10)
\structline[MPP_green, bend left = 30](29,7)(29,9)
\structline[MPP_green](29,9)(29,10)




\foreach \n in {0} \foreach \m in {0}{
\class[MPP_green, "x_2 y_4" { right, xshift = -0.16cm, black }, font = \scriptsize](\m,\n)
\class[MPP_green](\m,\n + 1)
\class[MPP_green](\m,\n + 2)
\class[MPP_green](\m,\n + 3)
\class[MPP_green](\m + 2,\n + 3)
\class[MPP_green](\m + 2,\n + 4)
\class[MPP_green](\m + 2,\n + 5)
\class[MPP_green](\m + 2,\n + 6)

\structline[MPP_green](\m,\n)(\m,\n + 1)
\structline[MPP_green,bend left = 30](\m,\n)(\m,\n + 2)
\structline[MPP_green](\m,\n + 2)(\m,\n + 3)
\structline[MPP_green, in = -160, out = 20](\m,\n + 1)(\m + 2,\n + 3)
\structline[MPP_green, in = -160, out = 20](\m,\n + 2)(\m + 2,\n + 4)
\structline[MPP_green](\m + 2,\n + 3)(\m + 2,\n + 4)
\structline[MPP_green, in = -160, out = 20](\m,\n + 3)(\m + 2,\n + 5)
\structline[MPP_green](\m + 2,\n + 5)(\m + 2,\n + 6)
\structline[MPP_green, bend right = 30](\m + 2,\n + 4)(\m + 2,\n + 6)
}

\foreach \n in {4} \foreach \m in {6}{
\class[MPP_green, "x_2^3 y_4" { right, xshift = -0.16cm, black }, font = \scriptsize](\m,\n)
\class[MPP_green](\m,\n + 1)
\class[MPP_green](\m,\n + 2)
\class[MPP_green](\m,\n + 3)
\class[MPP_green](\m + 2,\n + 3)
\class[MPP_green](\m + 2,\n + 4)
\class[MPP_green](\m + 2,\n + 5)
\class[MPP_green](\m + 2,\n + 6)

\structline[MPP_green](\m,\n)(\m,\n + 1)
\structline[MPP_green, bend left = 30](\m,\n)(\m,\n + 2)
\structline[MPP_green](\m,\n + 2)(\m,\n + 3)
\structline[MPP_green, in = -160, out = 20](\m,\n + 1)(\m + 2,\n + 3)
\structline[MPP_green, in = -160, out = 20](\m,\n + 2)(\m + 2,\n + 4)
\structline[MPP_green](\m + 2,\n + 3)(\m + 2,\n + 4)
\structline[MPP_green, in = -160, out = 20](\m,\n + 3)(\m + 2,\n + 5)
\structline[MPP_green](\m + 2,\n + 5)(\m + 2,\n + 6)
\structline[MPP_green, bend right = 30](\m + 2,\n + 4)(\m + 2,\n + 6)
}

\foreach \n in {5} \foreach \m in {10}{
\class[MPP_green, "x_2 x_5 y_4" { right, xshift = -0.16cm, yshift = -0.16cm, black }, font = \scriptsize](\m,\n)
\class[MPP_green](\m,\n + 1)
\class[MPP_green](\m,\n + 2)
\class[MPP_green](\m,\n + 3)
\class[MPP_green](\m + 2,\n + 3)
\class[MPP_green](\m + 2,\n + 4)
\class[MPP_green](\m + 2,\n + 5)
\class[MPP_green](\m + 2,\n + 6)

\structline[MPP_green](\m,\n)(\m,\n + 1)
\structline[MPP_green, bend left = 30](\m,\n)(\m,\n + 2)
\structline[MPP_green](\m,\n + 2)(\m,\n + 3)
\structline[MPP_green, in = -160, out = 20](\m,\n + 1)(\m + 2,\n + 3)
\structline[MPP_green, in = -160, out = 20](\m,\n + 2)(\m + 2,\n + 4)
\structline[MPP_green](\m + 2,\n + 3)(\m + 2,\n + 4)
\structline[MPP_green, in = -160, out = 20](\m,\n + 3)(\m + 2,\n + 5)
\structline[MPP_green](\m + 2,\n + 5)(\m + 2,\n + 6)
\structline[MPP_green, bend right = 30](\m + 2,\n + 4)(\m + 2,\n + 6)
}

\foreach \n in {5} \foreach \m in {13}{
\class[MPP_green, "x_2 x_3 y_6" { right, xshift = -0.16cm, yshift = -0.32cm, black }, font = \scriptsize](\m,\n)
\class[MPP_green](\m,\n + 1)
\class[MPP_green](\m,\n + 2)
\class[MPP_green](\m,\n + 3)
\class[MPP_green](\m + 2,\n + 3)
\class[MPP_green](\m + 2,\n + 4)
\class[MPP_green](\m + 2,\n + 5)
\class[MPP_green](\m + 2,\n + 6)

\structline[MPP_green](\m,\n)(\m,\n + 1)
\structline[MPP_green, bend left = 30](\m,\n)(\m,\n + 2)
\structline[MPP_green](\m,\n + 2)(\m,\n + 3)
\structline[MPP_green, in = -160, out = 20](\m,\n + 1)(\m + 2,\n + 3)
\structline[MPP_green, in = -160, out = 20](\m,\n + 2)(\m + 2,\n + 4)
\structline[MPP_green](\m + 2,\n + 3)(\m + 2,\n + 4)
\structline[MPP_green, in = -160, out = 20](\m,\n + 3)(\m + 2,\n + 5)
\structline[MPP_green](\m + 2,\n + 5)(\m + 2,\n + 6)
\structline[MPP_green, bend right = 30](\m + 2,\n + 4)(\m + 2,\n + 6)
}

\foreach \n in {6} \foreach \m in {20}{
\class[MPP_green, "x_2^3 y_6" { right, xshift = -0.16cm, black }, font = \scriptsize](\m,\n)
\class[MPP_green](\m,\n + 1)
\class[MPP_green](\m,\n + 2)
\class[MPP_green](\m,\n + 3)
\class[MPP_green](\m + 2,\n + 3)
\class[MPP_green](\m + 2,\n + 4)
\class[MPP_green](\m + 2,\n + 5)
\class[MPP_green](\m + 2,\n + 6)

\structline[MPP_green](\m,\n)(\m,\n + 1)
\structline[MPP_green, bend left = 30](\m,\n)(\m,\n + 2)
\structline[MPP_green](\m,\n + 2)(\m,\n + 3)
\structline[MPP_green, in = -160, out = 20](\m,\n + 1)(\m + 2,\n + 3)
\structline[MPP_green, in = -160, out = 20](\m,\n + 2)(\m + 2,\n + 4)
\structline[MPP_green](\m + 2,\n + 3)(\m + 2,\n + 4)
\structline[MPP_green, in = -160, out = 20](\m,\n + 3)(\m + 2,\n + 5)
\structline[MPP_green](\m + 2,\n + 5)(\m + 2,\n + 6)
\structline[MPP_green, bend right = 30](\m + 2,\n + 4)(\m + 2,\n + 6)
}
\foreach \n in {6} \foreach \m in {23}{
\class[MPP_green, "x_2 x_3^2 y_4" { right, xshift = -0.16cm, yshift = -0.16cm, black }, font = \scriptsize](\m,\n)
\class[MPP_green](\m,\n + 1)
\class[MPP_green](\m,\n + 2)
\class[MPP_green](\m,\n + 3)
\class[MPP_green](\m + 2,\n + 3)
\class[MPP_green](\m + 2,\n + 4)
\class[MPP_green](\m + 2,\n + 5)
\class[MPP_green](\m + 2,\n + 6)

\structline[MPP_green](\m,\n)(\m,\n + 1)
\structline[MPP_green, bend left = 30](\m,\n)(\m,\n + 2)
\structline[MPP_green](\m,\n + 2)(\m,\n + 3)
\structline[MPP_green, in = -160, out = 20](\m,\n + 1)(\m + 2,\n + 3)
\structline[MPP_green, in = -160, out = 20](\m,\n + 2)(\m + 2,\n + 4)
\structline[MPP_green](\m + 2,\n + 3)(\m + 2,\n + 4)
\structline[MPP_green, in = -160, out = 20](\m,\n + 3)(\m + 2,\n + 5)
\structline[MPP_green](\m + 2,\n + 5)(\m + 2,\n + 6)
\structline[MPP_green, bend right = 30](\m + 2,\n + 4)(\m + 2,\n + 6)
}

\foreach \n in {6} \foreach \m in {26}{
\class[MPP_green, "x_2 y_4 y_6" { right, xshift = -0.16cm, yshift = -0.16cm, black }, font = \scriptsize](\m,\n)
\class[MPP_green](\m,\n + 1)
\class[MPP_green](\m,\n + 2)
\class[MPP_green](\m,\n + 3)
\class[MPP_green](\m + 2,\n + 3)
\class[MPP_green](\m + 2,\n + 4)
\class[MPP_green](\m + 2,\n + 5)
\class[MPP_green](\m + 2,\n + 6)

\structline[MPP_green](\m,\n)(\m,\n + 1)
\structline[MPP_green, bend left = 30](\m,\n)(\m,\n + 2)
\structline[MPP_green](\m,\n + 2)(\m,\n + 3)
\structline[MPP_green, in = -160, out = 20](\m,\n + 1)(\m + 2,\n + 3)
\structline[MPP_green, in = -160, out = 20](\m,\n + 2)(\m + 2,\n + 4)
\structline[MPP_green](\m + 2,\n + 3)(\m + 2,\n + 4)
\structline[MPP_green, in = -160, out = 20](\m,\n + 3)(\m + 2,\n + 5)
\structline[MPP_green](\m + 2,\n + 5)(\m + 2,\n + 6)
\structline[MPP_green, bend right = 30](\m + 2,\n + 4)(\m + 2,\n + 6)
}


\end{sseqdata}

\begin{figure}[H]
    \centering
    \printpage[ name = Sqxypart, page = 2]   
    \caption{$\mathcal{A}_{1}$-module structure - mixed x-y-part}
    \label{fig:Sqxypart}
\end{figure}

\subsection{Determining the spin cobordism groups of \texorpdfstring{$BSs(32)$}{BSs(32)} up to degree 12}
\begin{theorem}
The 2-completed ko-homology of $BSs(32)$ up to degree 12 takes the following form:
\begin{table}[h!]
\begin{tabular}{ c | c c c c c c c c c }
n & 0 & 1 & 2 &  3 & 4 & 5 & 6 & 7 & 8 \\
\midrule
 $ko_n (BSs(32))$ &$\mathbb{Z}$ & $\mathbb{Z}_2$  & $2\, \mathbb{Z}_2$  & 0 & $2\,\mathbb{Z}\oplus \mathbb{Z}_2$ & 0& $2 \,\mathbb{Z}_2$  &0 & $4\,\mathbb{Z} \oplus \mathbb{Z}_8$ \\
\end{tabular}
\vspace{1cm}
\begin{tabular}{ c | c c c c }
n  & 9 & 10 & 11 & 12 \\
\midrule
 $ko_n (BSs(32))$ &$4\,\mathbb{Z}_2$ & $8\,\mathbb{Z}_2$  & $3 \,\mathbb{Z}_2$  & $6 \,\mathbb{Z} \oplus 7\,\mathbb{Z}_2 \oplus \mathbb{Z}_8$
\end{tabular}
\captionof{table}{ko-homology groups $ko_{n} (BSs(32))_{\widehat{2}}$.} 
\end{table}
\end{theorem}
\begin{proof}
    As outlined earlier the tool of our choice to calculate $ko_{n} (BSs(32))_{\widehat{2}}$ is the Adams spectral sequence. With the $\mathcal{A}_1$-module structure of $H^*(BSs(32), \mathbb{Z}_2)$ we can easily write down the second page $E_2$ of the spectral sequence. 
    Essential to this proof are the two maps we already alluded to in the discussion of the $\mathcal{A}_1$-module structure, namely $ko_{*}(BSs(32)) \to ko_{*}(B^2\mathbb{Z}_2)$ from the fibration $BSpin(n) \to BSs(n) \to B^2\mathbb{Z}_2$ and $ko_{n < 16}(BSs(32)) \to ko_{n < 16}(BE_8)$ from the inclusion $Ss(16) \xhookrightarrow{} E_8$. Again, the first one is highlighted by blue color and the second one by orange color in the following.
    The Adams spectral sequences for both ko-homologies are well studied, i.e.\ differentials and extensions are solved (or solvable) up to at least dimension 12, and due to the naturality of the Adams spectral sequence we can connect differentials that arise in our spectral sequence to known ones via the aforementioned maps.
    To make the spectral sequence easier to follow we discuss the spectral sequences for $ko_{*}(B^2\mathbb{Z}_2)_{\widehat{2}}$ and $ko_{*}(BE_8)_{\widehat{2}}$ separately.

\subsubsection{The Adams spectral sequence for \texorpdfstring{$ko_{*}(B^2\mathbb{Z}_2)_{\widehat{2}}$}{ko*(BB Z/2) at prime 2}}
The spectral sequence for $ko_{*}(B^2\mathbb{Z}_2)_{\widehat{2}} = ko_{*}(K(\mathbb{Z}_2,2))_{\widehat{2}}$, where $K(\mathbb{Z}_2,2)$ is the Eilenberg-Maclane space of the pair $(\mathbb{Z}_2,2)$\footnote{The equivalence between classifying spaces $B^n G$ of a discrete group $G$ and Eilenberg-Maclane spaces $K(G, n)$ is for example discussed in chapter 16.5. of~\cite{may1999concise}.}, was studied in~\cite{wilson1973new} and we will briefly recount their results for this spectral sequence.
We start by transferring the $\mathcal{A}_1$-modules, which are well known in the literature, see for example~\cite{Beaudry:2018ifm}, into an Adams chart. We get the following second page, where we have encircled nodes stemming from a full $\mathcal{A}_1$-module, which will not partake in neither a non-trivial differential nor an extension due to Margolis' theorem~\cite{Margolis74}:
\\
\begin{sseqdata}[
name = xpart,
Adams grading, classes = fill,
x range = {0}{13}, y range = {0}{8},
x tick step = 1,
run off differentials = {->},
xscale = 0.75,
yscale = 0.80,
class pattern = linearnew
]
\class[MPP_blue_light](2,0)
\class[MPP_blue_light](4,1)
\DoUntilOutOfBoundsThenNMore{2}{
\class[MPP_blue_light](\lastx,\lasty+1)
\structline[MPP_blue_light]
}
\class[MPP_blue_light](5,0)
\d2
\DoUntilOutOfBounds{
\class[MPP_blue_light](\lastx,\lasty+1)
\structline[MPP_blue_light]
\d2
}
\class[MPP_blue_light](8,2)
\DoUntilOutOfBoundsThenNMore{2}{
\class[MPP_blue_light](\lastx,\lasty+1)
\structline[MPP_blue_light]
}

\class[MPP_blue_light](9,3)
\structline[MPP_blue_light](8,2)(9,3)
\class[MPP_blue_light](10,4)
\structline[MPP_blue_light](9,3)(10,4)

\class[MPP_blue_light](9,0)
\d2
\DoUntilOutOfBounds{
\class[MPP_blue_light](\lastx,\lasty+1)
\structline[MPP_blue_light]
\d2
}
\class[MPP_blue_light](10,1)
\structline[MPP_blue_light](9,0)(10,1)
\d2
\class[MPP_blue_light](11,2)
\structline[MPP_blue_light](10,1)(11,2)
\d2

\class[MPP_blue_light](12,5)
\DoUntilOutOfBoundsThenNMore{2}{
\class[MPP_blue_light](\lastx,\lasty+1)
\structline[MPP_blue_light]
}

\class[MPP_blue_light](13,3)
\d2
\DoUntilOutOfBounds{
\class[MPP_blue_light](\lastx,\lasty+1)
\structline[MPP_blue_light]
\d2
}

\class[MPP_blue_light](8,0)
\class[MPP_blue_light](8,1)
\structline[MPP_blue_light](8,0,-1)(8,1,-1)
\class[MPP_blue_light](8,2)
\structline[MPP_blue_light](8,1,-1)(8,2,-1)
\DoUntilOutOfBoundsThenNMore{2}{
\class[MPP_blue_light](\lastx,\lasty+1)
\structline[MPP_blue_light]
}

\class[MPP_blue_light](9,1)
\structline[MPP_blue_light](8,0)(9,1,2)
\class[MPP_blue_light](10,2)
\structline[MPP_blue_light](9,1,2)(10,2)

\class[MPP_blue_light](9,0)
\DoUntilOutOfBounds{
\class[MPP_blue_light](\lastx,\lasty+1)
\structline[MPP_blue_light]
}
\foreach \n in {0,...,8}{
 \d3(9,\n,-1)(8,\n+3,-1)
}

\class[MPP_blue_light](12,3)
\DoUntilOutOfBoundsThenNMore{2}{
\class[MPP_blue_light](\lastx,\lasty+1)
\structline[MPP_blue_light]
}

\class[MPP_blue_light](11,0)
\class[MPP_blue_light](13,1)
\DoUntilOutOfBounds{
\class[MPP_blue_light](\lastx,\lasty+1)
\structline[MPP_blue_light]
}
\foreach \n in {1,...,8}{
 \d3(13,\n,-1)(12,\n+3,-1)
}

\class[MPP_blue_light](10,0)
\class[MPP_blue_light](12,1)  
\DoUntilOutOfBoundsThenNMore{2}{
\class[MPP_blue_light](\lastx,\lasty+1)
\structline[MPP_blue_light]
}

\class[MPP_blue_light](13,0)
\DoUntilOutOfBounds{
\class[MPP_blue_light](\lastx,\lasty+1)
\structline[MPP_blue_light]
}
\foreach \m in {0,...,8}{
 \d2(13,\m,-1)(12,\m+2,-1)
}

\class[MPP_blue_light, circlen = 2](6,0)
\class[MPP_blue_light, circlen = 2](10,0)
\class[MPP_blue_light, circlen = 2](12,0)

\end{sseqdata}

\begin{figure}[H]
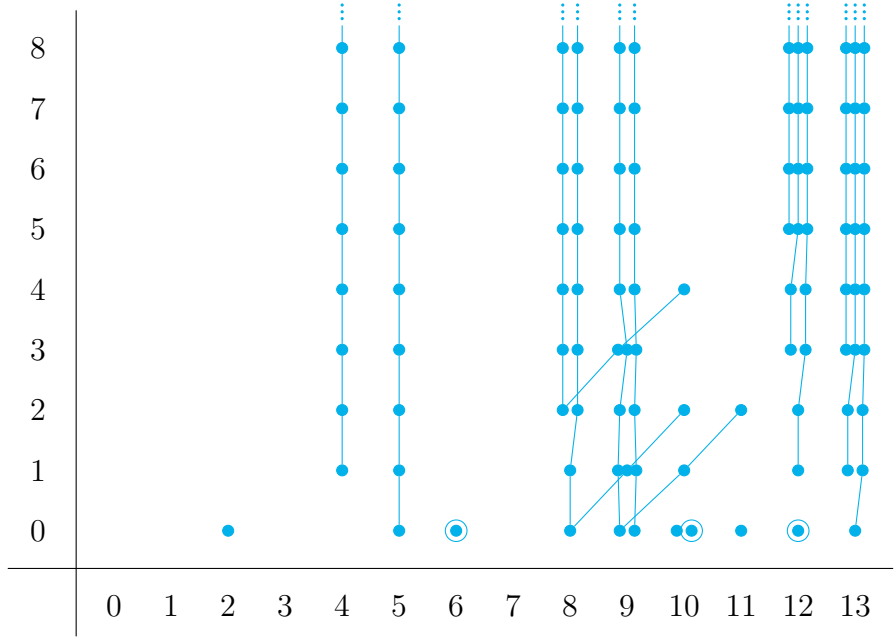

    \centering
    \printpage[ name = xpart, page = 1]    
    \caption{$E_2$ for $ko_{*}(B^2 \mathbb{Z}_2)_{\widehat{2}}$ without differentials}
    \label{fig:xpart}
\end{figure}

Now we can start looking at the differentials. As mentioned before a differential $d_r$ on the $r$-th page goes from a node $(s, t-s)$ to another one at $(s+r, t-s-1)$.
Interestingly, differentials in the Adams spectral sequence have a very peculiar property, namely they are equivariant under elements of $\mathrm{Ext}^{s, t}(\mathbb{Z}_2, \mathbb{Z}_2)$ meaning that in our case acting with the $h_i$'s with $i=0,1$ from before results in:
\begin{equation}
 d_r(h_i x) = h_i d_r(x)\,.
\end{equation}
With the properties of differentials in the Adams spectral sequence taken into account we see that only ``tower killing" differentials can be present going from one tower of nodes to another one.
\begin{lemma}{(Wilson~\cite{wilson1973new})}
    The tower killing differentials up to degree 40 can be identified with Bocksteins. Up to degree 13 we have the following differentials:
    \begin{itemize}
        \item There is a $d_2$ starting from towers in degree $4k+5$ coming from $\Sigma M_0$ and $\Sigma^9 \mathbb{Z}_2[\tilde{x}_9]$ killing the towers in degree $4k+4$ from $\Sigma^2 J[x_2]$. 
        \item There is another $d_2$ between towers generated by the same $\mathcal{A}_1$-modules as above multiplied by $x_2^4$.
        \item A $d_3$ kills the towers in degree $4k+8$ associated to $\Sigma^8 \mathbb{Z}_2[x_2^4]$ and starts from the towers coming from $\Sigma^9 M_0[\tilde{x}'_9]$ and $\Sigma^{11} J[\tilde{x}_{11}]$.
    \end{itemize} 
\end{lemma}

With these results on the differentials we are able to reach the infinity page.
Since there are no higher differentials than $d_3$ up to degree 13, page 4 is already equivalent to the infinity page. The second, third and final fourth page look as follows: 
\begin{figure}[H]
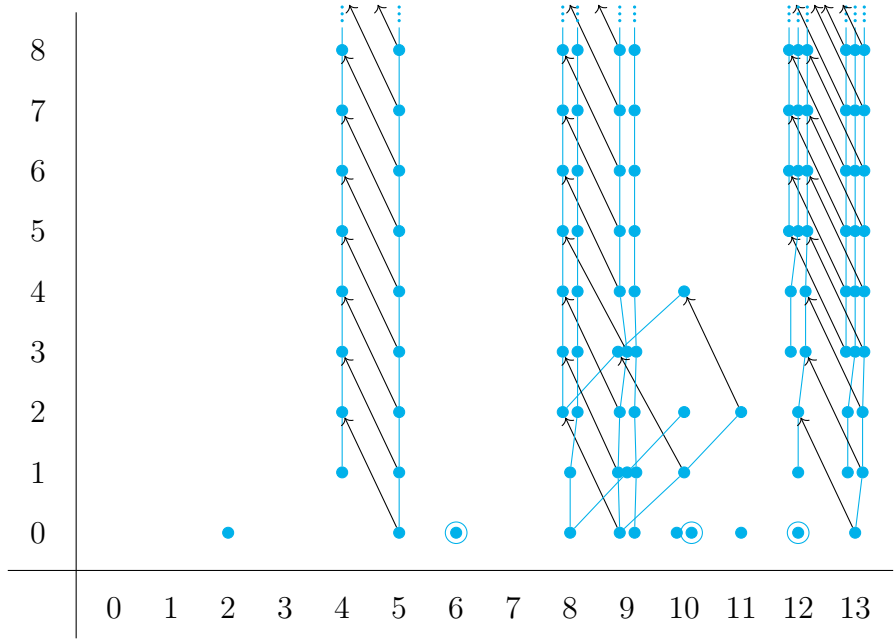

    \centering
    \printpage[ name = xpart, page = 2]   
    \caption{Second page $E_2$ for $ko_{*}(B^2 \mathbb{Z}_2)_{\widehat{2}}$ including $d_2$ differentials}
    \label{fig:xpart2}
\end{figure}
    
\begin{figure}[H]
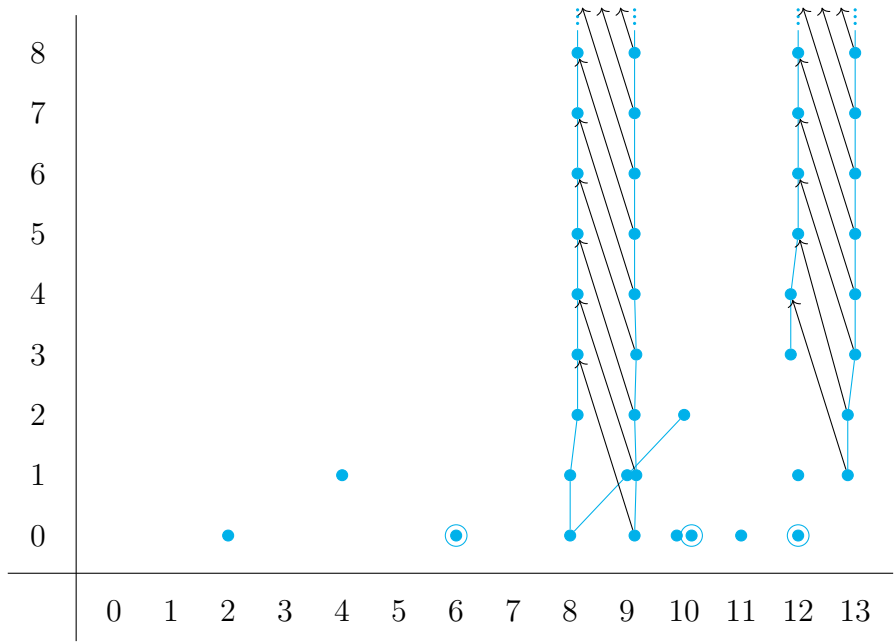

    \centering
    \printpage[ name = xpart, page = 3]   
    \caption{Third page $E_3$ for $ko_{*}(B^2 \mathbb{Z}_2)_{\widehat{2}}$}
    \label{fig:xpart3}
\end{figure}

\begin{figure}[H]
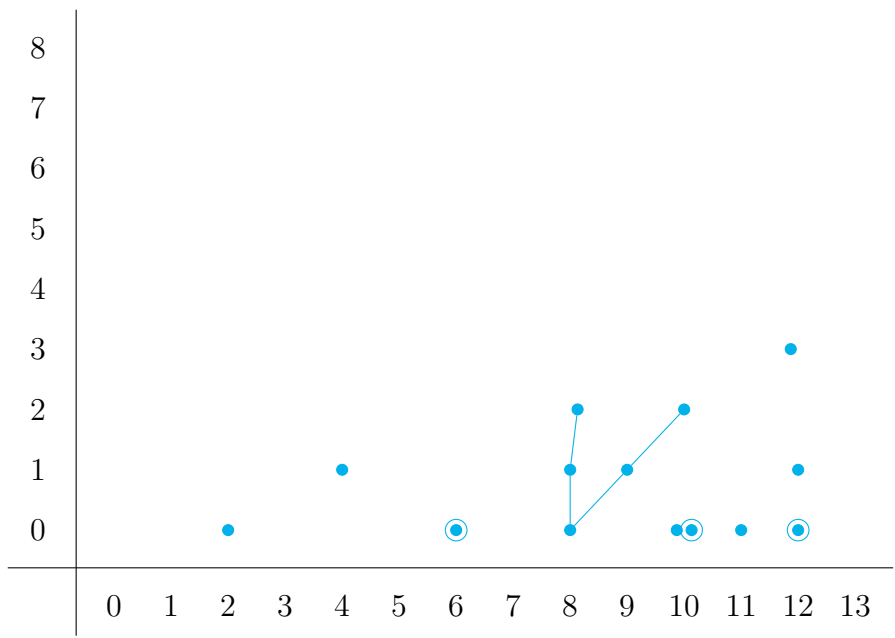

    \centering
    \printpage[ name = xpart, page = 4]   
    \caption{Final page $E_{\infty}$ for $ko_{*}(B^2 \mathbb{Z}_2)_{\widehat{2}}$}
    \label{fig:xpart_final}
\end{figure}

At this point we will not discuss any extension problems, we will do so once we assembled the final page of the full Adams spectral sequence for $ko_{*}(BSs(32))_{\widehat{2}}$.

\subsubsection{The Adams spectral sequence for \texorpdfstring{$ko_{*}(BE_8)_{\widehat{2}}$}{ko*(BE_8) at prime 2}}

\begin{sseqdata}[
name = ypart,
Adams grading, classes = fill,
x range = {0}{13}, y range = {0}{8},
x tick step = 1,
run off differentials = {->},
xscale = 0.95,
yscale = 0.8,
class pattern = linearnew
]

\class[MPP_orange](4,0)
\DoUntilOutOfBoundsThenNMore{2}{
\class[MPP_orange](\lastx,\lasty+1)
\structline[MPP_orange]
}

\class[MPP_orange](8,1)
\DoUntilOutOfBoundsThenNMore{2}{
\class[MPP_orange](\lastx,\lasty+1)
\structline[MPP_orange]
}
\class[MPP_orange](9,2)
\structline[MPP_orange](8,1)(9,2)
\class[MPP_orange](10,3)
\structline[MPP_orange](9,2)(10,3)

\class[MPP_orange](12,4)
\DoUntilOutOfBoundsThenNMore{2}{
\class[MPP_orange](\lastx,\lasty+1)
\structline[MPP_orange]
}

\class[MPP_orange](8,0)
\DoUntilOutOfBoundsThenNMore{2}{
\class[MPP_orange](\lastx,\lasty+1)
\structline[MPP_orange]
}
\class[MPP_orange](9,1)
\structline[MPP_orange](8,0)(9,1)
\class[MPP_orange](10,2)
\structline[MPP_orange](9,1)(10,2)

\class[MPP_orange](12,3)
\DoUntilOutOfBoundsThenNMore{2}{
\class[MPP_orange](\lastx,\lasty+1)
\structline[MPP_orange]
}

\class[MPP_orange](10,0)
\class[MPP_orange](12,1)
\DoUntilOutOfBoundsThenNMore{2}{
\class[MPP_orange](\lastx,\lasty+1)
\structline[MPP_orange]
}

\class[MPP_orange](10,0)
\d2(10,0)(9,2)
\class[MPP_orange](11,1)
\d2(11,1)(10,3)
\structline[MPP_orange](10,0)(11,1)
\class[MPP_orange](13,2)
\DoUntilOutOfBounds{
\class[MPP_orange](\lastx,\lasty+1)
\structline[MPP_orange]
}
\foreach \m in {2,...,7}{
 \d2(13,\m,-1)(12,\m+2,-3)
}
 \d2(13,8,-1)(12,10,-3)

\class[MPP_orange](12,0)
\DoUntilOutOfBoundsThenNMore{2}{
\class[MPP_orange](\lastx,\lasty+1)
\structline[MPP_orange]
}

\end{sseqdata}
As we stated before, the $\mathcal{A}_1$-structure of the ``y-part" is completely equivalent to the one for $BE_8$ in the degrees we are interested in.
Since we can also map the differentials by naturality of the Adams spectral sequence, we again first want to study the isolated case of $ko_{*}(BE_8)_{\widehat{2}}$ to infer a lot about the actual spectral sequence we care about.
To compute the low dimensional spin cobordism or ko-homology groups usually the isomorphism between $BE_8$ and the Eilenberg-Maclane space $K(\mathbb{Z},4)$ in degrees $\leq 15$ is exploited.
The associated Adams spectral sequence for $ko_{*}(K(\mathbb{Z},4))$ was studied in detail in~\cite{francisintegrals}.
Firstly, there are only a few $\mathcal{A}_1$-modules, which are quickly translated into an Adams chart, which looks as follows: 
\begin{figure}[H]
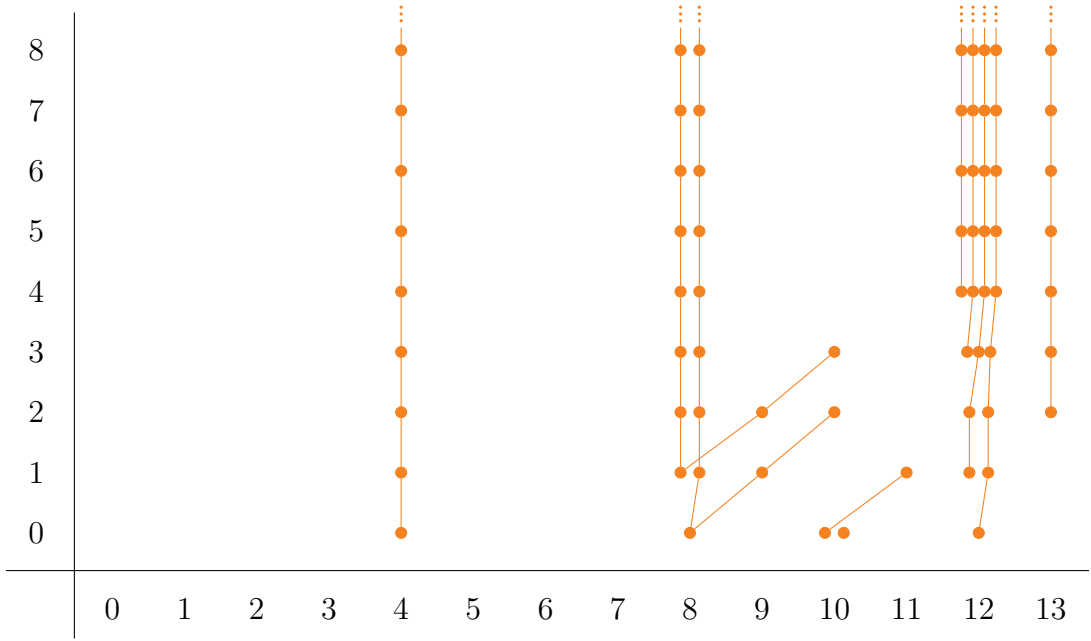

    \centering
    \printpage[ name = ypart, page = 1]
    \caption{Second page $E_2$ for $ko_{*}(BE_8)_{\widehat{2}}$ without differentials}
    \label{fig:ypart1}
\end{figure}
    
As we can see from the second page (without differentials) the $\mathcal{A}_1$-module structure of the ``y-part" doesn't leave much room for differentials. Still, there are two distinct differentials possible, which are in fact realized.
\begin{lemma}{(Francis~\cite{francisintegrals})}
    There is a $d_2$ in degree 10 (and 11 linked by $h_1$ action) and a tower killing differential in degree 13 such that there is no torsion in degree 12.
    Both differentials are linked by a higher cohomology operation, namely a Massey product.
\end{lemma}

Since there are no more possible differentials in this range, we are done. 
\begin{figure}[H]
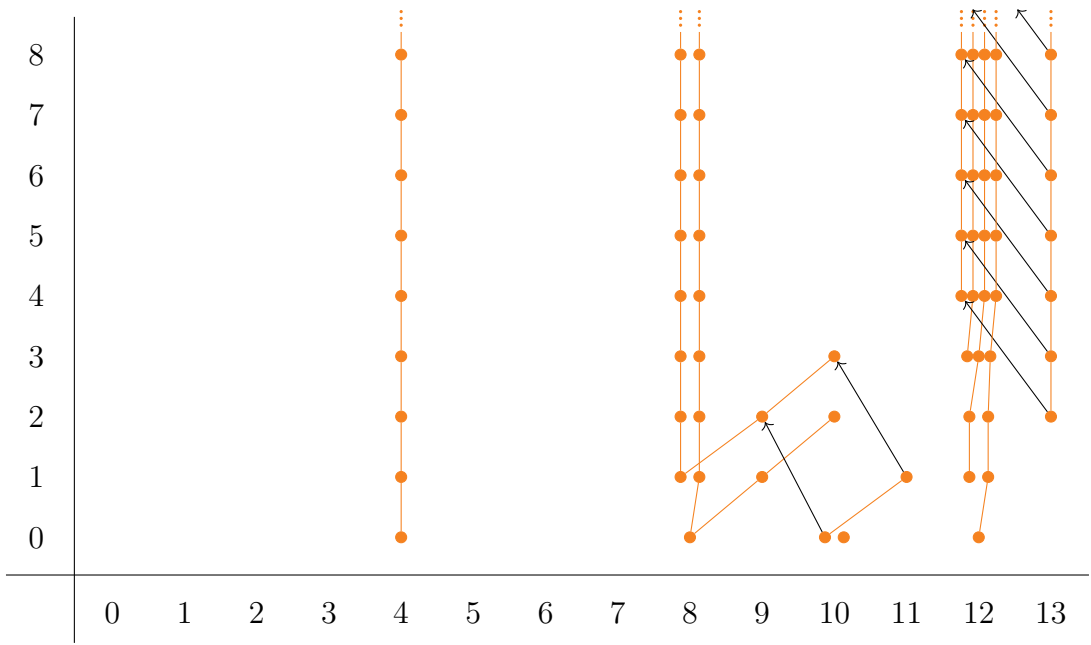

    \centering
    \printpage[ name = ypart, page = 2] 
    \caption{Second page $E_2$ for $ko_{*}(BE_8)_{\widehat{2}}$ with differentials}
    \label{fig:ypart2}
\end{figure}

\begin{figure}[H]
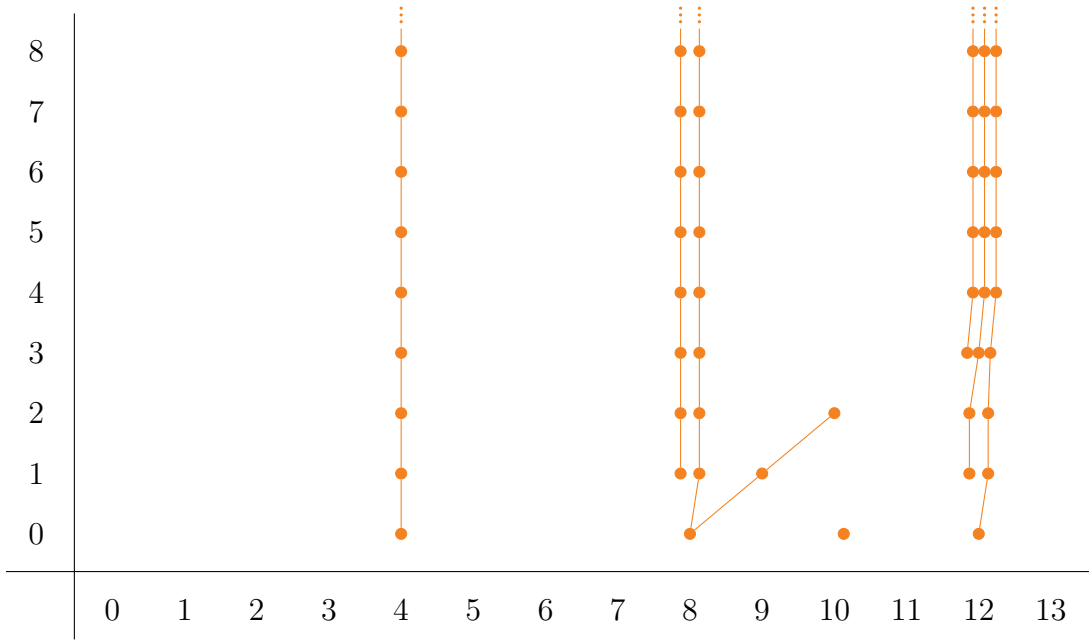

    \centering
    \printpage[ name = ypart, page = 3]
    \caption{Final page $E_{\infty}$ for $ko_{*}(BE_8)_{\widehat{2}}$}
    \label{fig:ypart3}
\end{figure}

\begin{sseqdata}[
name = xypart,
Adams grading, classes = fill,
x range = {0}{13}, y range = {0}{8},
x tick step = 1,
run off differentials = {->},
xscale = 0.95,
yscale = 0.8,
class pattern = linearnew
]


\class[MPP_green, circlen = 2](6,0)
\class[MPP_green, circlen = 2](10,0)
\class[MPP_green, circlen = 2](11,0)
\class[MPP_green, circlen = 2](11,0)
\class[MPP_green, circlen = 2](12,0)
\class[MPP_green, circlen = 2](12,0)
\class[MPP_green, circlen = 2](12,0)

\class[MPP_green](8,0)
\DoUntilOutOfBoundsThenNMore{2}{
\class[MPP_green](\lastx,\lasty+1)
\structline[MPP_green]
}
\class[MPP_green](9,0)
\class[MPP_green](10,1)
\structline[MPP_green](9,0)(10,1)

\class[MPP_green](12,2)
\DoUntilOutOfBoundsThenNMore{2}{
\class[MPP_green](\lastx,\lasty+1)
\structline[MPP_green]
}

\class[MPP_green](10,0)
\class[MPP_green](12,1)
\DoUntilOutOfBoundsThenNMore{2}{
\class[MPP_green](\lastx,\lasty+1)
\structline[MPP_green]
}

\class[MPP_green](12,0)
\DoUntilOutOfBoundsThenNMore{2}{
\class[MPP_green](\lastx,\lasty+1)
\structline[MPP_green]
}

\class[MPP_green](13,0)
\DoUntilOutOfBoundsThenNMore{2}{
\class[MPP_green](\lastx,\lasty+1)
\structline[MPP_green]
}
\foreach \n in {0,...,8}{
 \d3(13,\n,-1)(12,\n+3,-1)
}

\class[MPP_green](12,0)

\end{sseqdata}
\subsubsection{Completing the Adams spectral sequence for \texorpdfstring{$ko_{*}(BSs(32))_{\widehat{2}}$}{ko*(BSs(32)) at prime 2}}
For the final ``x-y-part" the $\mathcal{A}_1$-modules are again well known, see for example~\cite{Beaudry:2018ifm}. Additionally, we demonstrated the computation of the Adams chart for the $\tilde{R}_2$-module in the introduction.
So, we get the following second page without differentials. \\
\begin{figure}[H]
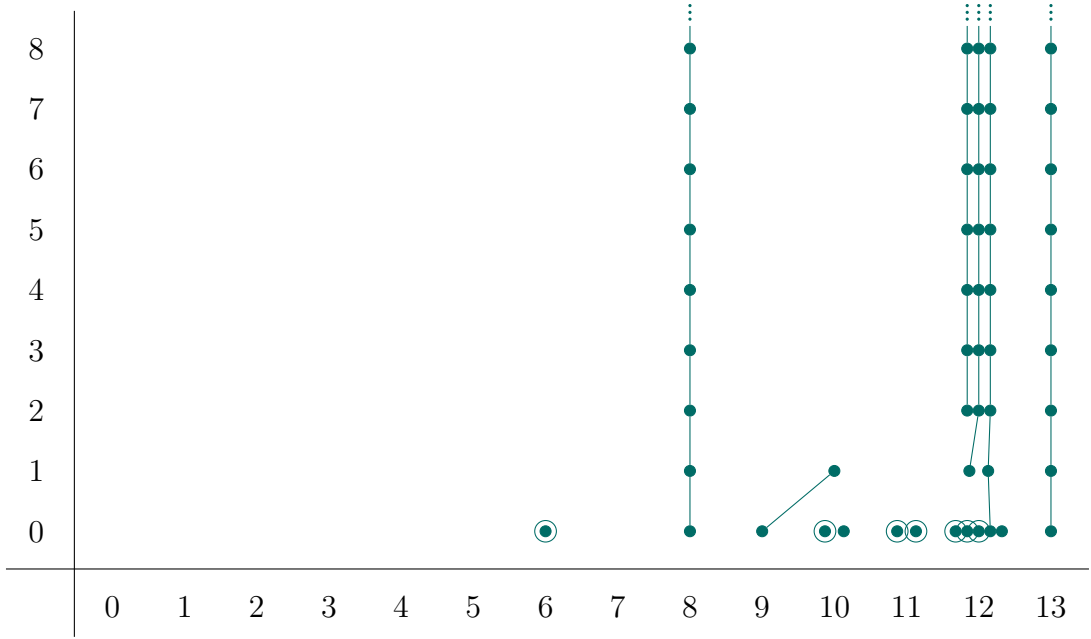

    \centering
    \printpage[ name = xypart,page = 2 ]   
    \caption{Second page $E_2$ for the ``x-y-part" without differentials}
    \label{fig:xypart1}
\end{figure}

Again, the encircled nodes denote $\mathbb{Z}_2$'s stemming from full $\mathcal{A}_1$, which do not participate in any differentials (or non-trivial extensions).
\begin{lemma}
    Up to degree 13 there is only one differential $d_3$ coming from the $\widetilde{Q}[\tilde{x}'_9 y_4]$ 
    \begin{equation}
        d_3(\tilde{x}'_9 \, y_4) = d_3(\tilde{x}'_9) \, y_4\,.
    \end{equation} reducing the $\widetilde{Q}[x_2^4 y_4]$ tower to a $\mathbb{Z}_8$-torsion piece.
\end{lemma}
\begin{proof}
    From the Adams chart and equivariance of differentials under $h_0$- and $h_1$-actions in particular it is an immediate consequence that the only possible differential would come from the $\widetilde{Q}[\tilde{x}'_9 y_4]$ tower in degree 13. 
    Then we can use the fact that the cup product of the cohomology ring of $H^{*}(BSs(32), \mathbb{Z}_2)$ induces a multiplicative structure on the Adams Spectral sequence~\cite{douady11suite}. This entails that we can determine the differential from the $\widetilde{Q}[\tilde{x}'_9 y_4]$ tower via the Leibniz rule induced by the multiplicative structure.
    Of course we have to be careful with the relations \eqref{first_relation}-\eqref{fourth_relation}.
    With this taken into account we get that the differential in question is in fact a $d_3$ as it is the only non-trivial possibility respecting the Leibniz rule
    \begin{equation}
        d_3(\tilde{x}'_9 \, y_4) = d_3(\tilde{x}'_9) \, y_4\,.
    \end{equation}
    Consequently this cuts the $\widetilde{Q}[x_2^4 \,y_4]$ tower in degree 12 to just a $\mathbb{Z}_8$ torsion piece. 
\end{proof}
The resulting pages of the spectral sequence look as follows:
\begin{figure}[H]
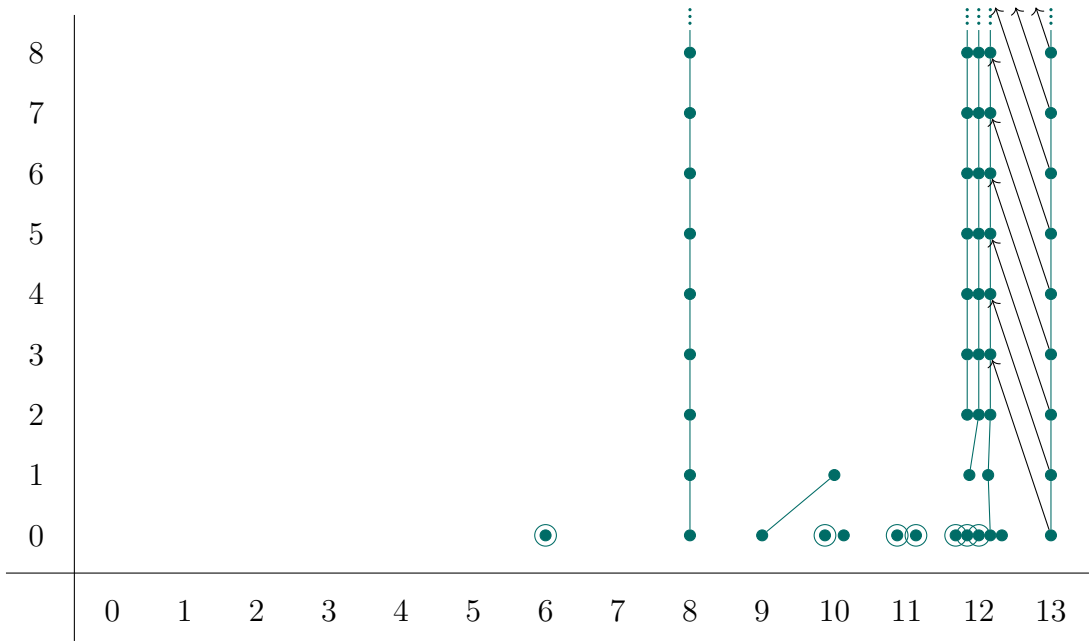

    \centering
    \printpage[ name = xypart, page = 3]   
    \caption{Third page $E_3$ for the ``x-y-part" with differentials}
    \label{fig:xypart2}
\end{figure}
   
\begin{figure}[H]
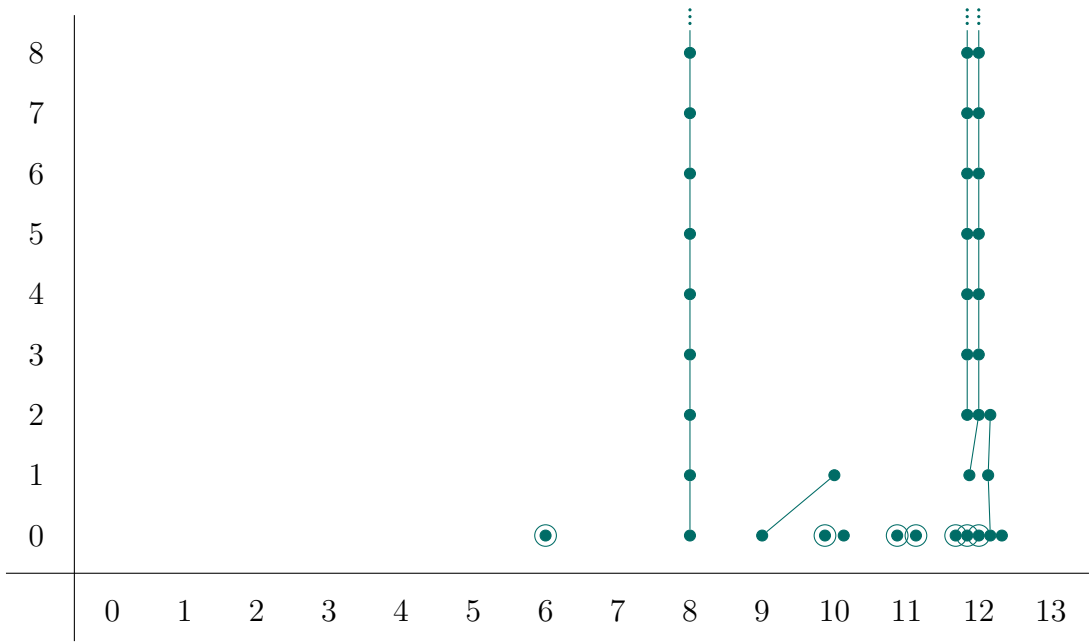

 \centering
 \printpage[ name = xypart, page = 4]    
 \caption{Final page $E_{\infty}$ for the ``x-y-part"}
 \label{fig:xypart3}
\end{figure}

\begin{sseqdata}[
name = finalpart,
Adams grading, classes = fill,
x range = {0}{12}, y range = {0}{8},
x tick step = 1,
run off differentials = {->},
xscale = 1.03972,
yscale = 0.8,
class pattern = linearnew
]
\class[MPP_blue_light](2,0)
\class[MPP_blue_light](4,1)

\class[MPP_blue_light](8,0)
\class[MPP_blue_light](8,1)
\structline[MPP_blue_light](8,0)(8,1)
\class[MPP_blue_light](8,2)
\structline[MPP_blue_light](8,1)(8,2)

\class[MPP_blue_light](9,1)
\structline[MPP_blue_light](8,0)(9,1)
\class[MPP_blue_light](10,2)
\structline[MPP_blue_light](9,1)(10,2)

\class[MPP_blue_light](12,3)

\class[MPP_blue_light](11,0)

\class[MPP_blue_light](10,0)
\class[MPP_blue_light](12,1)  

\class[MPP_blue_light, circlen = 2](6,0)
\class[MPP_blue_light, circlen = 2](10,0)
\class[MPP_blue_light, circlen = 2](12,0)


\class[MPP_orange](4,0)
\DoUntilOutOfBoundsThenNMore{2}{
\class[MPP_orange](\lastx,\lasty+1)
\structline[MPP_orange]
}

\class[MPP_orange](8,1)
\DoUntilOutOfBoundsThenNMore{2}{
\class[MPP_orange](\lastx,\lasty+1)
\structline[MPP_orange]
}

\class[MPP_orange](8,0)
\DoUntilOutOfBoundsThenNMore{2}{
\class[MPP_orange](\lastx,\lasty+1)
\structline[MPP_orange]
}
\class[MPP_orange](9,1)
\structline[MPP_orange](8,0,-1)(9,1,-1)
\class[MPP_orange](10,2)
\structline[MPP_orange](9,1,-1)(10,2,-1)

\class[MPP_orange](12,3)
\DoUntilOutOfBoundsThenNMore{2}{
\class[MPP_orange](\lastx,\lasty+1)
\structline[MPP_orange]
}

\class[MPP_orange](10,0)
\class[MPP_orange](12,1)
\DoUntilOutOfBoundsThenNMore{2}{
\class[MPP_orange](\lastx,\lasty+1)
\structline[MPP_orange]
}

\class[MPP_orange](12,0)
\DoUntilOutOfBoundsThenNMore{2}{
\class[MPP_orange](\lastx,\lasty+1)
\structline[MPP_orange]
}



\class[MPP_green, circlen = 2](6,0)
\class[MPP_green, circlen = 2](10,0)
\class[MPP_green, circlen = 2](11,0)
\class[MPP_green, circlen = 2](11,0)
\class[MPP_green, circlen = 2](12,0)
\class[MPP_green, circlen = 2](12,0)
\class[MPP_green, circlen = 2](12,0)

\class[MPP_green](8,0)
\DoUntilOutOfBoundsThenNMore{2}{
\class[MPP_green](\lastx,\lasty+1)
\structline[MPP_green]
}
\class[MPP_green](9,0)
\class[MPP_green](10,1)
\structline[MPP_green](9,0)(10,1)

\class[MPP_green](12,2)
\DoUntilOutOfBoundsThenNMore{2}{
\class[MPP_green](\lastx,\lasty+1)
\structline[MPP_green]
}

\class[MPP_green](10,0)
\class[MPP_green](12,1)
\DoUntilOutOfBoundsThenNMore{2}{
\class[MPP_green](\lastx,\lasty+1)
\structline[MPP_green]
}

\class[MPP_green](12,0)
\class[MPP_green](12,1)
\structline[MPP_green](12,0,-1)(12,1,-1)
\class[MPP_green](12,2)
\structline[MPP_green](12,1,-1)(12,2,-1)

\class[MPP_green](12,0)

\end{sseqdata}
Eventually, putting all parts together again we obtain the final page for the Adams spectral sequence for $ko_{*}(BSs(32))_{\widehat{2}}$: 
\begin{figure}[H]
    \centering
    \printpage[ name = finalpart, page = 2]    
    \caption{Final page $E_{\infty}$ for $ko_{*}(BSs(32))_{\widehat{2}}$}
    \label{fig:finalpart}
\end{figure}

Reading off the ko-homology of $BSs(32)$ is straightforward at this point:
each node corresponds to a $\mathbb{Z}_2$ and each vertical line, i.e.\ $h_0$, to multiplication by 2. Consequently, the towers are depicting a
$\mathbb{Z}$-summand.
Unfortunately, we are not quite done yet with the determination of the ko-homology groups as there are potentially non-trivial so called hidden extensions.
This means that $\mathbb{Z}_2$-nodes in appropriate degree could either split into a direct sum of each other or be connected by a so far undetected $h_0$.
In the case at hand the extensions all split and we do not uncover any hidden extensions.
Non-trivial extensions are possible in degrees 9, 10 and 12. So let's tackle them case by case.
\begin{lemma}
$ko_{9}(BSs(32)) \cong \textcolor{MPP_blue_light}{\mathbb{Z}_2} \oplus \textcolor{MPP_orange}{\mathbb{Z}_2} \oplus \textcolor{MPP_green}{\mathbb{Z}_2}$\,.
\end{lemma}
\begin{proof}
    All of the nodes in degree 9 are connected via a non-trivial $h_1$-action to nodes in degree 10. Now suppose there is a hidden extension between the green node and either the blue or orange node. Then these two nodes combine to a $\mathbb{Z}_4$. Let's call the generator of this $\mathbb{Z}_4$ $x$ and the image of the $h_1$ action on 2$x$ in degree 10 $y$, such that $h_1\,2x \neq 0$. Moreover, the $h_1$ action lifts to an $\eta$-action, corresponding to multiplication by the ``anti-periodic" circle in ko-homology. However, since $2 \, \eta = 0$, this is contradicting $h_1\,2x = \eta\,2x \neq 0$. Therefore, all $\mathbb{Z}_2$'s must split in $ko_{9}(BSs(32))$.
\end{proof}
\begin{lemma}
$ko_{10}(BSs(32)) \cong \textcolor{MPP_green}{3\,\mathbb{Z}_2} \oplus \textcolor{MPP_orange}{2\,\mathbb{Z}_2} \oplus \textcolor{MPP_blue_light}{3\,\mathbb{Z}_2}$\,.
\end{lemma}
\begin{proof}
    First of all, the Margolis theorem tells us that the two circled nodes do not participate in any non-trivial extensions and just split off.
    Furthermore, since all extensions in degree 9 are split, $\eta$ carries this splitting into degree 10: $\textcolor{MPP_green}{\mathbb{Z}_2} \oplus \textcolor{MPP_orange}{\mathbb{Z}_2} \oplus \textcolor{MPP_blue_light}{\mathbb{Z}_2} \subset ko_{10}(BSs(32))$. Since the remaining nodes can not form non-trivial extension amongst each other, we're done.
\end{proof}
\begin{lemma}
There are no hidden extensions in $ko_{12}(BSs(32))$.
\end{lemma}
\begin{proof}
We start by looking at the three blue nodes coming from $H^*(BSs(32), \mathbb{Z}_2)$. Of course, the circled node splits off completely due to Margolis theorem. To see the splitting between the other two nodes we need to put a bit more work in. In our proof before we have exploited a non-trivial $\eta$-action. 

Here, we will use that the top blue node in degree 12 coming from the $\Sigma^8 \, \mathbb{Z}_2 [x_2^4]$-module is connected to a node stemming from the same module in degree 20 by a $\omega$-action, corresponding to multiplication by a so-called Bott-manifold $B_8$. Just like the node in degree 12 the one in degree 20 is the only surviving after taking the tower-killing differential $d_3$ into account. The same can not be said about the other blue node, which leads to the following situation:

Let's denote the generator of the bottom $\mathbb{Z}_2$ $x$ and the generator of the top blue node $y$. Assuming the extension doesn't split we have $2x = y$. Now, since $\omega \cdot y \neq 0$, we get $\omega \cdot 2x \neq 0$. However, this is a contradiction, because the $\omega$-action on $x$ vanishes. Therefore, the extension must split.

Since we can map the blue nodes isomorphically into the Adams spectral sequence for $ko_{*}(B^2\mathbb{Z}_2)$, the blue nodes have to split in the case of $BSs(32)$ as well.
The remaining possible hidden extension is between the bottom single green $\mathbb{Z}_2$ node and one of the two $\mathbb{Z}$-towers.

Again, we will make use out of non-trivial $\omega$-actions. Namely, both towers are linked by $\omega$-actions to corresponding towers in degree 20.
Even though we have not determined the differentials in that degree we know that there cannot be any non-trivial differential acting on these towers as $d_r(\omega \,x) = \omega \,d_r(x)$ and there is no differential acting on our towers in degree 12. Moreover, the $\omega$-action doesn't just imply the existence of towers in degree 20, but since there is no non-trivial $\omega$-action on the green bottom node (not circled), which can be determined from the short exact sequence used to calculate the Adams chart for the $\tilde{R}_2$,
we know that both possible extensions do in fact split. 
\end{proof}

Now that we have covered the 2-torsion part for $ko_{*}(BSs(32))$ the natural next step is to calculate the odd-torsion Adams spectral sequence. However, as we show next this not necessary, in fact $ko_{*}(BSs(32))_{\widehat{2}} = ko_{*}(BSs(32))$ for $* \leq 12$.
\begin{lemma}
    Neither $ko_{*}(BSs(32))$ nor $\Omega^{Spin}_{*}(BSs(32))$ contains any odd torsion up to at least degree 12 and therefore 
    $ko_{*}(BSs(32))_{\widehat{2}} = ko_{*}(BSs(32))$ \\ as well as $\Omega^{Spin}_{*}(BSs(32))_{\widehat{2}} = \Omega^{Spin}_{*}(BSs(32))$ below at least degree 13.
\end{lemma}
\begin{proof}
    One way to see this is to deploy another type of spectral sequence capable of computing groups of generalized homology theories like connective ko-homology, namely the Atiyah-Hirzebruch spectral sequence (AHSS)~\cite{atiyah1961vector}
    \begin{equation}
        E^2_{p,q} = H_{p}(B,G_{q}(F)) \Rightarrow G_{p+q}(X)\,.
    \end{equation}
    for a Serre fibration $X \to B$ with fiber $F$ and a generalized homology theory like cobordism or K-homology, i.e.\ satisfying all the Eilenberg-Steenrod axioms except for the dimension axiom.
    
    In particular we will use the following Serre fibration
    \begin{equation}
    \label{doublecoverBss}
          BSpin(4n) \to BSs(4n) \to B^2\mathbb{Z}_2 \,,
    \end{equation}
    which can be seen as a consequence of the Puppe sequence for the fibration $B\mathbb{Z}_2 \to BSpin(4n) \to BSs(4n)$.
    Now, since the differentials are group homomorphisms and since $E^{r+1} = \frac{Ker(d^r)}{Im(d^r)}$, they can not generate odd torsion groups from 2-torsion groups on subsequent pages except for differentials between different $\mathbb{Z}$ entries.
    
    Therefore, if there is no odd torsion on the second page of the aforementioned spectral sequence and we can exclude differentials between free Abelian entries, the infinity page isn't going to contain odd torsion either.
    
    Both $ko_{*}(BSpin(n))$ and $\Omega^{Spin}_{*}(BSpin(n))$ do not contain any odd torsion, see e.g.~\cite{Lee:2022spd} for this statement.
    Therefore, as long as the integral homology of $B^2\mathbb{Z}_2$ does not contain odd torsion groups the second page of the AHSS is free of odd torsion.
    Indeed, the integral homology groups for $B^2\mathbb{Z}_2 = K(\mathbb{Z}_2, 2)$ have been computed up to degree 200 in~\cite{clement2002integral}, showing absence of odd torsion groups below at least degree 13.
    
    Additionally, the integral homology groups of $K(\mathbb{Z}_2,2)$ do not contain any free Abelian groups in degree $n \geq 1$, thereby excluding any differentials between different $\mathbb{Z}$ on any page of the spectral sequence.
\end{proof}
\end{proof}
\subsubsection{From \texorpdfstring{$ko_{*}(BSs(32))$}{ko*(BSs(32))} to \texorpdfstring{$\Omega_{*}^{Spin}(BSs(32))$}{spin cobordism of BSs(32)}}
As the final step we can now complete our calculation of $\Omega^{Spin}_{*}(BSs(32))$
\begin{theorem}
As the result of the ABP-splitting \eqref{ABP} we get the following \\cobordism groups
    \begin{table}[h!]
    \begin{tabular}{ c | c c c c c c c c c }
    n & 0 & 1 & 2 &  3 & 4 & 5 & 6 & 7 & 8 \\
    \midrule
     $\Omega^{Spin}_n (BSs(32))$ &$\mathbb{Z}$ & $\mathbb{Z}_2$  & $2\mathbb{Z}_2$  & 0 & $2\mathbb{Z}\oplus \mathbb{Z}_2$ & 0& $2 \mathbb{Z}_2$  &0 & 5$\mathbb{Z} \oplus \mathbb{Z}_8$ \\
    \end{tabular}
    \vspace{1cm}
    \begin{tabular}{ c | c c c c }
    n  & 9 & 10 & 11 & 12 \\
    \midrule
     $\Omega^{Spin}_n (BSs(32))$ &$5\mathbb{Z}_2$ & $10\mathbb{Z}_2$  & $3 \mathbb{Z}_2$  & $8 \mathbb{Z} \oplus 9\mathbb{Z}_2 \oplus \mathbb{Z}_8$
    \end{tabular}
    \captionof{table}{Spin cobordism groups $\Omega^{Spin}_n (BSs(32))$.} 
     \label{table: Finalspinbordism}
    \end{table}
\end{theorem}
\begin{proof}
    Let's write the ABP-splitting \eqref{ABP} in a more convenient form for us to see what we mean by $ko_{n-10}\langle 2 \rangle(X)_{(2)}$:
    \begin{equation}
    \label{ABP2}
        H^*(MSpin, \mathbb{Z}_2) \cong \mathcal{A} \otimes_{\mathcal{A}_1} (\mathbb{Z}_2 \oplus \Sigma^8 \mathbb{Z}_2 \oplus \Sigma^{10} J \oplus \dots)
    \end{equation}
    Here, $J$ denotes the so-called ``Joker"-module, which we encountered already, for example directly as the first module within the ``x-part". For a detailed account of the next terms we refer to appendix D.1 in~\cite{Freed:2016rqq}.

    From the expression above we can see that we get:
\begin{align}
    \Omega_n^{Spin}(BSs(32))_{\widehat{2}} =& \,\pi_{t-s}(ko\, \wedge\, BSs(32))_{\widehat{2}}\, \oplus \, \pi_{8+t-s}(ko\, \wedge\, BSs(32))_{\widehat{2}}\\
    &\oplus \, \pi_{10+t-s}(ko\, \wedge\, J \wedge\, BSs(32))_{\widehat{2}}\, \oplus\, \dots\,, \nonumber
\end{align}

    where $\pi_{t-s}(ko \wedge J \wedge BSs(32))_{\widehat{2}}$ is calculated by the Adams spectral sequence $Ext^{s,t}_{\mathcal{A}_1}(J \otimes H^*(BSs(32),\mathbb{Z}_2), \mathbb{Z}_2) \Rightarrow \pi_{t-s}(ko \wedge J \wedge BSs(32))_{\widehat{2}}$.
    
    We see that apart from the two ko-homology building blocks we just get $\Sigma^{10} J \otimes H^*(BSs(32)) = \Sigma^{10} J \otimes (\Sigma^2 J \oplus \dots)$ on top in the degree range we are interested in. This amounts to a full $\mathcal{A}_1$-module in degree 12~\cite{baker2020homotopy}. The rest is just utilizing our result on the connective ko-homology of $BSs(32)$.
    Finally, we can just apply our knowledge from the AHSS that $\Omega_{n \leq 12}^{Spin}(BSs(32))$ doesn't contain odd torsion, so we are done.
\end{proof}

\section[The non-vanishing cobordism groups and string theory]{\texorpdfstring{The non-vanishing cobordism groups and \\ string theory}{The non-vanishing cobordism groups and string theory}}
\label{sec:physics_interpretation}
Finally, we want to sort out the interpretation of the non-trivial cobordism groups in regards to the Cobordism Conjecture requiring some physical mechanism trivializing them.
Before we dive into the trivialization mechanism let us provide some explanation on the physics of the ko-homology building blocks under the ABP decomposition \eqref{ABP}.
\subsection{The ko-homology building blocks \texorpdfstring{$ko_{n}(BSs(32))$}{kon(BSs(32))}} 
We begin by noting that we can split $ko_{n}(BSs(32)) = ko_{n}(pt)\, \oplus \, \widetilde{ko}_{n}(BSs(32))$, where $\widetilde{ko}_{n}(X)$ are the so called reduced ko-homology groups of $X$. With $n \geq 0$ the first part under the splitting $ko_{n}(pt)$ can be nicely mapped to $KO^{-n}(pt)$ under Poincar\'e duality. These groups specifically were famously proposed to classify type I D-branes~\cite{Witten:1998cd}. 
This raises the question of how to properly understand the latter groups under the splitting.
The refinement we propose is that the $ko_n$-homology groups at hand detect the magnetic charges of the type I D$p$-branes measured in the $n$-dimensional space transverse to the D-brane worldvolume. 

Specifically, the $ko_{n}(pt)$ detects the set of gravitational magnetic charges, whereas $\widetilde{ko}_{n}(BSs(32))$ classifies the gauge-theoretic magnetic charges, which arise from open fundamental strings connecting the D-brane to the background D9-branes. 

To exemplify this let us look at the magnetic charge of the D5-brane in type I string theory. 
Let us consider its contribution to the Bianchi identity of $C_2$, which we integrate over the compact transverse space:
\begin{equation}
\label{tadpole_D5}
    N_{grav} = N_{gauge} + N_{D5} \,,
\end{equation}
where $N_{grav} = - \frac{p_1}{2} = \frac{1}{16 \pi^2} \int_{M_4} \mathrm{Tr}\, R \wedge R$ denotes the curvature contribution and $N_{gauge} = - p_1(E) = \int_{M_4} \mathrm{Tr}\, F \wedge F$ stands for the gauge instanton contribution.
This puts the gravitational and the gauge instanton on the same footing as the number of D5-branes making them really indistinguishable from one another~\cite{Douglas:1995bn}.
This behavior should not surprise us, if our refinement to the K-theory - D-brane relationship holds true.
And precisely, $\frac{1}{16 \pi^2} \int_{M_4} \mathrm{Tr}\, R \wedge R$ is just a multiple of $\hat{A}_4$, i.e.\ the index of the Dirac operator, and detects non-triviality of $\mathbb{Z} \cong ko_{4}(pt)$. Fittingly, $\int_{M_4} \mathrm{Tr}\, F \wedge F$ uplifts to the cobordism invariant detecting the other free Abelian piece in codimension 4: $\mathbb{Z} \subseteq \widetilde{ko}_{4}(BSs(32))$.
After carefully implementing electric-magnetic duality (Poincar\'e duality) on a fixed background space there is a further refinement for the electric charges. This involves a few more steps, which will be worked through in~\cite{Basile:2024}. 

Of course, we should ask for an interpretation of the same ko-homology building blocks for the $SemiSpin(32)$ heterotic string as well.
Actually, we can see this as the K-theoretic realization of Hull's proposal for the non-perturbative sector of the HO string~\cite{Hull:1998he}, which is based on the description of the $SemiSpin(32)$ heterotic string as a composite orientifold of type IIB~\cite{Hull:1997kt}
\begin{equation}
    \widetilde{\Omega} = S \, \Omega \, S^{-1}\,,
\end{equation}
consisting of the $S$, the S-duality $SL(2, \mathbb{Z})$-transformation of type IIB, and $\Omega$ the standard type IIB orientifold giving rise to type I string theory. Consistency of the theory requires now 32 NS9-branes giving rise to the background gauge field just like their S-dual twins. Still, why does string perturbation theory look so fundamentally different for the $SemiSpin(32)$ heterotic string as compared to the type I string?

The resolution lies in the fact that the fundamental open string on the type I side gets turned into a D-string on the heterotic side, whose tension scales like $\frac{1}{g_s}$. This entails that in the perturbative limit $g_s \to 0$ the D-string, becoming infinitely heavy, retracts into the object it is ending on. The massless sector of the D-string survives this limit and for example for D-strings attached to the fundamental closed HO string these massless modes provide the worldsheet structure we are familiar with.

Our ko-homology groups of course are not sensitive to the string coupling and a heterotic interpretation entails that the groups classify the charges of the heterotic NS$p$-branes associated to its open D-string sector. In particular $\widetilde{ko}_n(BSs(32))$ assorts the charges resultant from the background NS9-branes connected to the lower dimensional NS$p$-branes through the D-strings.
Unsurprisingly, we get for example a tadpole condition for NS5-branes on a compact transverse space mirroring \eqref{tadpole_D5}. 
So S-duality between type I and HO string theory manifests itself
as self-duality of the ko-homology groups:
\begin{equation}
    ko_n(BSs(32))\, \xleftrightarrow{\;S-\text{duality}\;\,}\, ko_n({}^{L}BSs(32)) \cong ko_n(BSs(32))\,.
\end{equation}
In the following we will usually take the type I perspective on the cobordism groups. Of course, there exists a S-dual heterotic perspective along the lines outlined above, which we sometimes highlight as well.

Next we want to take a look at the torsional pieces of our spin cobordism groups up until $n = 8$.
\subsection{The torsional spin cobordism subgroups \texorpdfstring{$\tilde{\Omega}^{Spin}_{n \leq 8}(BSs(32))$}{up to d = 8}}
As we have discussed before certain $\mathbb{Z}_{2^p}$ summands of $\Omega^{Spin}_2(BSs(32))$ are in the image of the map to $\Omega^{Spin}_n(B^2\mathbb{Z}_2)$.
In particular, we want to look at the image under the map between the reduced cobordism groups.
Their associated cobordism invariants are $\int_{M_{2k}} x_2^{k}$ with $k = 1, 3, 4$, where $x_2^{k}$ are the ``generalized" Stiefel-Whitney classes referenced in~\cite{Berkooz:1996iz, Witten:1997bs} and $\frac{1}{2} \mathcal{P}_2(x_2)$ in dimension $n=4$~\cite{Wan:2018bns}, where $\mathcal{P}_2(x)$ is the Pontryagin square of $x$.
Additionally, we encounter a $\mathbb{Z}_2$ in $\Omega^{Spin}_6(BSs(32))$ detected by $x_2 y_4$.

\begin{table}[h!]
    \begin{center}
    \resizebox{\textwidth}{!}{
    \begin{tabular}{ c | c c c c c c c c c }
    n & 0 & 1 & 2 &  3 & 4 & 5 & 6 & 7 & 8 \\
    \midrule
     $\widetilde{\Omega}^{Spin}_n (BSs(32))_{tors}$ &0 & 0  & $\mathbb{Z}_2$  & 0 & $\mathbb{Z}_2$ & 0& $\mathbb{Z}_2 \oplus \mathbb{Z}_2$  &0 & $\mathbb{Z}_8$ \\
    \midrule
    invariants & - & - & $\int_{M_2} x_2$ & - & $\frac{1}{2} \mathcal{P}_2(x_2)$ & - & $\int_{M_6} x_2^3$ , $\int_{M_6} x_2 y_4$ & - & $\int_{M_8} x_2^4$ \\
    \bottomrule
    \end{tabular}
    }
    \end{center}
    \caption{Reduced torsional cobordism groups and their invariants}
\end{table}

\begin{itemize}[leftmargin = 5.5mm]
    \item n = 2: Following~\cite{Freed:2016rqq, Debray:2023yrs} we know that the nontrivial $\mathbb{Z}_2$ is detected by the corresponding cohomology class, since the $\mathbb{Z}_2$ stems from filtration $s = 0$ in the Adams Spectral sequence. The convention in the physics literature on this topic is to call the cohomology class $\widetilde{w}_2$ in analogy to the Stiefel-Whitney-classes $w_i$. However, we will stick to our convention of calling it $x_2$. First, we note that $\widetilde{\Omega}^{Spin}_2(BSs(32)) \cong \widetilde{ko}_2(BSs(32))$
    and therefore expect a D-brane interpretation.
    The simplest example we can look at is a type I compactification on a $T^2$ with full $Ss(32)$-gauge group, i.e.\ without vector structure in the language of~\cite{Witten:1997bs}. As laid out in~\cite{Uranga:2000xp} the non-trivial $\mathbb{Z}_2$ ko-theory charge arising here should be understood as the charge of a non-BPS $\widehat{\text{D}7}$-brane.
    Now, why is that? The charge of the ``conventional" $\widehat{\text{D}7}$-brane in type I was identified with $KO^{-2}(pt)$ in~\cite{Witten:1998cd}, which is of course Poincar\'e dual to $ko_2(pt)$. 
    We can draw an analogy to the D5-brane case we saw before.
    While the $\widehat{\text{D}7}$-brane does not couple to a dynamical gauge field, the K-theoretic tadpole cancellation \eqref{tadpole_D5} still arises~\cite{Uranga:2000xp}. 
    The ingredients are again invariants of the ko-homology groups in question:
    \begin{equation}
    \label{tadpole_D7}
        \int_{M_2} x_2 + \hat{A}_2 + N_{\widehat{D7}} = 0 \mod 2\,,
    \end{equation}
    where we refer to the mod 2 index of the Dirac operator on the transverse compact space $M_2$ as $\hat{A}_2$~\cite{Atiyah_V:1971}.
    Similarly to the D5-brane case the charge of the $\widehat{\text{D}7}$-brane is really indistinguishable from the gauge and gravitational contribution.
    This should not surprise us, since both invariants have been identified with the non-BPS $\widehat{\text{D}7}$-brane. The gravitational invariant showed up as a ``Berry's phase" in the system of a probe non-BPS $\widehat{D0}$-brane and the $\widehat{\text{D}7}$-brane itself~\cite{Gukov:1999yn}. The complementary statement can be found in~\cite{Uranga:2000xp}: The $\widehat{D7}$ shows up as a toron gauge-field configuration in the background presence of D9-brane(s), matching precisely our expectation that this charge would not be realized without the background gauge field provided by the D9-branes. 
    The torsional tadpole constraint \eqref{tadpole_D7} can be fulfilled in a couple of different ways. The simplest case is, if every term vanishes separately (mod 2). In the language of~\cite{McNamara:2019rup}, the higher form symmetries are gauged, i.e.\ the physical system is in the trivial cobordism configuration.
    
    However, this makes the opposite configuration, i.e.\ both the gauge and gravitational charge contributions are odd and $N_{\widehat{D7}} = 0 \mod 2$, more subtle. 
    While cancellation against each other ensures string theoretic tadpole cancellation, none of the two nontrivial global symmetries are resolved by gauging.
    
    So the remaining pathway to quantum gravitational consistency is breaking both of those by codimension 3 defects.
    Remarkably, there seem to exist supergravity solutions describing these codimension 3 defects. 
    The defect 6-brane breaking $\widetilde{\Omega}^{Spin}_2(BSs(32))$ was identified in~\cite{Kaidi:2023tqo} as the extremal limit of the non-supersymmetric heterotic black 6-brane solution~\cite{Horowitz:1991cd}, which looks like a 4-dimensional magnetic black hole with 6 flat dimensions added:
    \begin{equation}
        ds^2 = dx^{\mu}dx_{\mu} + dy^2 + r_0 d\Omega^2_2, \quad e^{-2\phi} = g_s^{-2} e^{y/r_0}\,,
    \end{equation}
    additionally the $S^2$ horizon comes equipped with non trivial $\int_{S_2} x_2$ charge.
    
    Moreover, the authors of~\cite{Kaidi:2023tqo} propose a detailed worldsheet description of this 6-brane through a $(SU(16)/\mathbb{Z}_4)_1$ spin-CFT with $c_L = 15$ on the $S^2$-part, such that in total one gets:
    \begin{equation}
        \mathbb{R}^{(1,6)} \, \times \, \mathbb{R}_{\text{ linear dilaton}} \, \times \, \text{CFT}((SU(16)/\mathbb{Z}_4)_1)\,. 
    \end{equation}
    The 6-brane, which has to accompany the non-BPS $\widehat{\text{D}7}$-brane and breaks $\Omega^{Spin}_2(pt)$, is expected to have a supergravity solution as well. It should look somewhat similar to its twin, the ETW-7-brane, arising when studying the consistency of the backreacted geometry of a single non-BPS $\widehat{\text{D}8}$-brane~\cite{Blumenhagen:2022mqw}.
    There, one gets a spontaneously compactified dimension, in form of a $S^1$, in the space transversal to the $\widehat{\text{D}8}$-brane, which matches that the generator of $\widetilde{\Omega}^{Spin}_1(pt)$ is precisely the circle with periodic boundary condition for fermions. As expected one finds that the topology of the space transverse to the ETW-7-brane solution is that of a disk.
    The ETW-6-brane subsequently would feature a transverse topology $S^1 \times D^2$ bounding the 2-torus generating $\Omega^{Spin}_2(pt)$\footnote{We have not commented on the interval in the geometry created by the dynamical tadpole associated to the non-BPS $\widehat{D8}$/$\widehat{\text{D}7}$-brane. However, it turns out, when properly accounting for the precise topology, that only the boundary of the interval has a nontrivial cobordism group. Therefore, we can think of the ETW-brane as the defect necessary to cap off the $S^1$ boundaries of the cylinder with disks.}. 
    
    While the geometry of the respective solutions reflects nicely the expected properties from the cobordism viewpoint, the purely bosonic Lagrangian utilized to construct these solutions does not detect the periodic fermionic boundary conditions.
    One option would be to use the aforementioned $\widehat{D0}$-brane probe~\cite{Gukov:1999yn}. It should be emphasized however that this $\mathbb{Z}_2$ arises from the necessary spin-structure and not from the background D9-branes introducing the background gauge theory and as already discussed above has to be combined with $\widetilde{ko}_2(BSs(32)) \cong \mathbb{Z}_2$ to describe the full K-theoretic charges associated to the non-BPS $\widehat{\text{D}7}$-brane.  
    
    Finally, we can ask, if we can understand the transformation of these defects under duality. In particular the $Ss(32)$-heterotic string is T-dual to the $(E_8 \times E_8)\rtimes \mathbb{Z}_2$ heterotic string.
    In particular it is known that the $Ss(32)$-string on a $T^2$-compactification with nontrivial $x_2$ is T-dual to the $(E_8 \times E_8)\rtimes \mathbb{Z}_2$ heterotic string, where the two $E_8$\,s are exchanged, when going around one circle in the 2-torus~\cite{Witten:1997bs}.
    Precisely this exchange symmetry is detected by $\Omega^{Spin}_1(B\mathbb{Z}_2)$ stemming from the fibration $E_8 \times E_8 \to (E_8 \times E_8) \rtimes \mathbb{Z}_2 \to \mathbb{Z}_2$ as explained in~\cite{Debray:2023rlx}.
    While we work at a different level of structural refinement, i.e.\ not yet incorporating the twisted string structure of the $Ss(32)$ heterotic string, these classes would also be present, if one would work with $\Omega^{Spin}_n(B((E_8 \times E_8)\rtimes \mathbb{Z}_2))$ instead, matching our level of refinement.

    The T-duality between two NS5-branes and the non-supersymmetric heterotic 6-brane proposed in~\cite{DeFreitas:2024yzr} is not in reach at this level of refinement as the T-duality appears to map defects breaking nontrivial cobordism groups. The NS5-brane first shows up as a breaking defect at String-structure breaking $\Omega^{String}_3(pt) \cong \mathbb{Z}_{24}$ generated by a $S^3$ with H-flux~\cite{McNamara:2019rup}. Therefore, we should expect to find the $S^3$ as the generator of the cobordism with twisted string structure for the HO-string, probably even \\ $\Omega^{String-Ss(32)}_3(pt) \cong \Omega^{String}_3(pt)$.

    \item n = 4: The natural setup to look at here is type I/HO string compactifications without vector structure on a four dimensional manifold. In particular the orientifold of type IIB on the $T^4/\mathbb{Z}_2$ orbifold limit of K3 studied in~\cite{Gimon:1996rq} comes to mind. Tadpole cancellation requires us to introduce 8 dynamical D5-branes with its corresponding collective Chan-Paton index value being 32 due the orientifold and the orbifold.
    There are a multitude of different solutions based on the position of the D5-branes, i.e.\ whether they reside at one of the 16 fixed points or not.
    Now, does this orbifold construction impact our analysis as we are not working with $\mathbb{Z}_2$ equivariant ko-homology/spin cobordism necessary to properly take the orbifold into account?
    
    The answer turns out to be no. As demonstrated in~\cite{Berkooz:1996iz} the fixed points of this precise orbifold\footnote{This does not hold true for other orbifold limits of K3~\cite{Polchinski:1996ry}.} can actually all be blown up and the spectra fully agree with the smooth K3. More importantly for our discussion here, the fixed points were shown to each carry a hidden instanton and when blown up the associated gauge bundle supported on the $S^2$ replacing the singularity indeed stems from a Ss(32)-bundle. This can be seen as follows: The construction in~\cite{Gimon:1996rq} necessitates a nontrivial twist acting on the Chan-Paton label by a matrix M in $16 \times 16$ block form:
    \begin{equation}
        M = 
        \begin{pmatrix}
            0 & I\\
            -I & 0
        \end{pmatrix}\,.
    \end{equation}
    Since this does not square to 1, it would be causing an inconsistency, if the gauge group were $SO(32)$. But we can take $M^2 = w$, where $w \in Z(Spin(32))$ and $w = - 1$ in the vector and one spinor representation of $Spin(32)$, but $w = 1$ in the second spinor representation. So for the actual $Ss(32)$ gauge group we have exactly $M^2 = 1$ and topologically trivial paths around the fixed point are well defined.
    
    Now, upon blowing up the singularity the authors of~\cite{Berkooz:1996iz} showed that the resulting two-sphere $S$ with self-intersection $-2$ supports a gauge field obeying Dirac quantization for the adjoint or spinor, but not for the vector. Subsequently the first Chern number of the gauge bundle on $S$ has to be normalized as $\int_S \frac{F}{2 \pi} = \frac{1}{2}$.

    After blowing up the singularities the final instanton number on the Eguchi-Hanson space $X$\footnote{Here $X$ is used to approximate the space close to $S$ and can be treated as the total space of the line bundle $\mathcal{O}(−2)$, when we regard the two-sphere $S$ as the complex space $\mathbb{P}_1$.} turns out to be precisely the one matching the requirement that the 16 fixed points provide the missing \contour{black}{16} to the K3 tadpole cancellation condition $24 = n_{D5}$ + \contour{black}{16}. Therefore, the GP model after singularities have been blown up can be understood from a non-equivariant cobordism perspective.
    
    However, while this example is very instructive to construct an instanton exhibiting the key physical feature of absence of vector structure as expected for a $Ss(32)$-instanton, it is not yet what we should aim for. It turns out that the ``background" gauge group provided by the D9-branes is not $Ss(32)$ anymore, but broken to $U(16)/\mathbb{Z}_2$.
    Therefore, we will pivot to a setup, where the full $Ss(32)$-group is preserved. A detailed account of this setup as the general F-theory construction of type I/HO K3 compactifications with fully preserved 9-brane gauge group is provided by~\cite{Aspinwall:1996vc}.
    It turns out that the instanton construction of~\cite{Berkooz:1996iz} has to be adapted just slightly.
    In general we can write the integral over the curvature of the gauge bundle as
    \begin{equation}
    \label{curv_instanton}
        \int_{C_i} \frac{F}{2 \pi} = \frac{1}{2} (\widetilde{w}_2 \cdot C_i) + k\,,
    \end{equation}
    where we adapted our notation to the one in~\cite{Aspinwall:1996vc}, denoting with $C_i$ one of the 16 exceptional divisors associated to the 16 fixed points of the $\mathbb{Z}_2$ orbifold limit of K3 and $\widetilde{w}_2 \in H^2(C_i, \mathbb{Z}_2)$ arising from the classifying map $f: X \to BSs(32)$.
    
   \cite{Aspinwall:1996vc}~now argues that since the curvature of the instanton should arise from the local geometry the ``generalized" Stiefel-Whitney class $\widetilde{w}_2$ has to be proportional to $C_i$. 
    While the instanton of~\cite{Berkooz:1996iz} is defined by $\widetilde{w}_2 = \frac{1}{2} C_i$, the instanton in the case of unbroken $Ss(32)$ is constructed through $\widetilde{w}_2 = C_i$~\cite{Aspinwall:1996vc}.
    Now, the obstruction can not be detected by $C_i$ itself, as $C_i \cdot C_i = 0 \mod 2$. Nevertheless, we can consider a dual exceptional divisor $C'_i$, which can detect the obstruction, i.e.\ $C'_i \cdot \widetilde{w}_2 = C'_i \cdot C_i = 1$.
    Therefore, this results in a contribution of four per instanton to the instanton number on K3.
    Based on the dual description of the type I/HO string on K3 in terms of F-theory compactified on a Calabi-Yau threefold $X$ with elliptic fibration $f: X \to \mathbb{F}_n$, where $\mathbb{F}_n$ denotes a Hirzebruch surface,~\cite{Aspinwall:1996vc} further derived a precise correspondence between the number of these unconventional instantons present and the integer $n$ defining the Hirzebruch surfaces as $\mathbb{F}_n = \mathbb{P}(\mathcal{O}_{\mathbb{P}_1} \oplus \mathcal{O}_{\mathbb{P}_1}(n))$, which are $\mathbb{P}_1$-fibrations over $\mathbb{P}_1$~\cite{hirzebruch1951klasse}.
    The tadpole cancellation on the K3 for the $Ss(32)$-heterotic string then becomes 
    \begin{equation}
        \sum_i^{\mu} k_i + 4(4-n) = 24\,,
    \end{equation}
    where the first contribution comes from $\mu$ groups of conventional heterotic instantons and the second contributions from the special instantons of instanton number 4 associated to the $4-n$ collisions between the discriminant and the zero section of $\mathbb{F}_n$ in the F-theory description.
    Equally, we can phrase this as a contribution from 
    \begin{equation}
    \widetilde{w}_2 \cdot \widetilde{w}_2 = 2(n-4)   
    \end{equation}
    through the $\widetilde{w}_2$ correction to the integral over the curvature of the gauge bundle \eqref{curv_instanton}.
    This entails that there are three succinct equivalence classes\footnote{To emphasize the difference between models ``with" and ``without vector structure", the case $\widetilde{w}_2 = 0$ is conventionally split off from the $\widetilde{w}_2^2 = 0 \mod 4$ equivalence class. Although, just by looking at $\widetilde{w}_2^2$ we cannot detect this difference.} of the type I/HO string on K3 enumerated by elements in integral homology $H_4(B^2\mathbb{Z}_2, \mathbb{Z})$, namely 
    \begin{align}
        &\widetilde{w}_2 = 0\,, \\
        &\widetilde{w}_2 \neq 0 \text{ and } \widetilde{w}_2^2 = 0 \mod 4\,,\\
        &\widetilde{w}_2 \neq 0 \text{ and } \widetilde{w}_2^2 = 2 \mod 4\,.
    \end{align}
    Now we will argue that the charge of the instanton can be understood from gauging $\mathbb{Z}_2 \cong \widetilde{\Omega}^{Spin}_4(B^2\mathbb{Z}_2) \subset \widetilde{\Omega}^{Spin}_4(BSs(32))$, i.e.\ the only physically acceptable configurations are the ones with $\frac{1}{2}\mathcal{P}(x_2) = 0 \mod 2$, where $\frac{1}{2}\mathcal{P}(x_2)$ is the cobordism invariant associated to the nontrivial $\mathbb{Z}_2$.

    By employing the associated Adams spectral sequence 
    \begin{equation}
        E_{2}^{s,t} = \mathrm{Ext}^{s,t}_{\mathcal{A}_0}(H^*(B^2\mathbb{Z}_2, \mathbb{Z}_2), \mathbb{Z}_2) \Rightarrow H_{t-s}(B^2\mathbb{Z}_2, \mathbb{Z})_{\widehat{2}}\,.
    \end{equation}
    we can examine the close connection between $H_4(B^2\mathbb{Z}_2, \mathbb{Z}) = \mathbb{Z}_4$ and\\ $\widetilde{ko}_4(B^2\mathbb{Z}_2) = \mathbb{Z}_2$, that both arise from the same $d_2$ cutting down the initial $h_0$ tower going from the 2nd to the 3rd page in their respective Adams spectral sequence.
    \begin{sseqdata}[
        name = ko4,
        Adams grading, classes = fill,
        x range = {2}{6}, y range = {0}{8},
        x tick step = 1,
        run off differentials = {->},
        xscale = 0.75,
        yscale = 0.80,
        class pattern = linearnew
        ]
        \class[MPP_blue_light](2,0)
        \class[MPP_blue_light](4,1)
        \DoUntilOutOfBoundsThenNMore{2}{
        \class[MPP_blue_light](\lastx,\lasty+1)
        \structline[MPP_blue_light]
        }
        \class[MPP_blue_light](5,0)
        \d2
        \DoUntilOutOfBounds{
        \class[MPP_blue_light](\lastx,\lasty+1)
        \structline[MPP_blue_light]
        \d2
        }
    \end{sseqdata}
    \begin{sseqdata}[
        name = H4,
        Adams grading, classes = fill,
        x range = {2}{6}, y range = {0}{8},
        x tick step = 1,
        run off differentials = {->},
        xscale = 0.75,
        yscale = 0.80,
        class pattern = linearnew
        ]
        \class[MPP_green_dark](2,0)
        \class[MPP_green_dark](4,0)
        \DoUntilOutOfBoundsThenNMore{2}{
        \class[MPP_green_dark](\lastx,\lasty+1)
        \structline[MPP_green_dark]
        }
        \class[MPP_green_dark](5,0)
        \d2
        \DoUntilOutOfBounds{
        \class[MPP_green_dark](\lastx,\lasty+1)
        \structline[MPP_green_dark]
        \d2
        }
    \end{sseqdata}

    \begin{figure}[H]
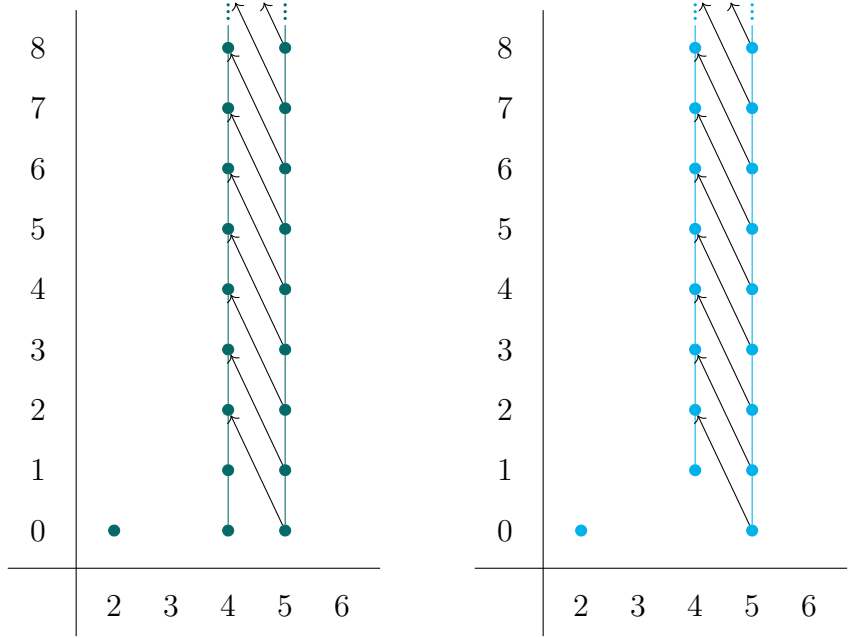

        \centering
        \printpage[ name = H4, page = 2]
        \hspace{1cm}
        \printpage[ name = ko4, page = 2]
        \caption{Second pages $E_2$ for $H_{4}(B^2 \mathbb{Z}_2, \mathbb{Z})_{\widehat{2}}$ and $ko_{4}(B^2 \mathbb{Z}_2)_{\widehat{2}}$}
        \label{fig:H4_ko4_comp}
    \end{figure}
    We lose the class in filtration zero, when going from $H_{4}(B^2 \mathbb{Z}_2, \mathbb{Z})_{\widehat{2}}$ to \\
    $\widetilde{ko}_{4}(B^2 \mathbb{Z}_2)_{\widehat{2}}$, because of a non-trivial second Steenrod Square $Sq^2$ enforcing the spin condition\footnote{This $Sq^2$ is the first non-trivial differential $d_2$ for the corresponding AHSS $H_{p}(B^2 \mathbb{Z}_2,ko_{q}(pt)) \Rightarrow ko_{p+q}(B^2 \mathbb{Z}_2)$ and can be understood as an instability due to D-instantons as discussed in~\cite{Maldacena:2001xj}.}. 
    At this point we can exploit the fact that the two $\mathbb{Z}_2$ in degree 4 and filtration 0 and 1, let's call them $a$ and $b$, are connected by a $h_0$ and therefore detect corresponding classes $\alpha, \beta \in H_{4}(B^2 \mathbb{Z}_2, \mathbb{Z})$ with $\beta = 2\, \alpha$. This means that $b$ exactly detects the equivalence classes $\widetilde{w}_2^2 = 0, 2 \mod 4$ we are interested in. Now, the map $\mathcal{A}_{1} \to \mathcal{A}_0$ induces a map 
    \begin{equation}
        \phi: \mathrm{Ext}^{1,3}_{\mathcal{A}_1}(H^*(B^2\mathbb{Z}_2, \mathbb{Z}_2), \mathbb{Z}_2) \to \mathrm{Ext}^{1,3}_{\mathcal{A}_0}(H^*(B^2\mathbb{Z}_2, \mathbb{Z}_2), \mathbb{Z}_2)\,,
    \end{equation}
    which is an isomorphism due to mapping a $\mathbb{Z}_2$ to another $\mathbb{Z}_2$(see appendix D. of~\cite{Debray:2023yrs} for a very similar map).
    Therefore, our cobordism invariant detects the same equivalence classes as $\widetilde{w}_2^2 = 0, 2 \mod 4$.
    Invoking gauging the nontrivial $\mathbb{Z}_2 \subset \Omega^{Spin}_4(BSs(32))$ we need even numbered magnetic charges $\frac{1}{2}\mathcal{P}(x_2) = 0 \mod 2$ carried by the NS5-branes/D5-branes in the respective descriptions of the theory. 
    Going backwards through the chain we can map this to $\widetilde{w}_2^2 = 0 \mod 4$ or expressed differently the contribution of the NS5-branes/D5-branes, i.e.\ the \\ $SemiSpin(32)$-instantons in the full F-theory framework, to the instanton number measured within $H_{4}(BSs(32), \mathbb{Z})$ has to be a multiple of 4, such that $\widetilde{w}_2^2 = 0 \mod 4$. As we reviewed above this is exactly the charge such an instanton contributes~\cite{Aspinwall:1996vc}.

    \item n = 6: Additionally to $\int_{M_6} x_2^3 \in \mathbb{Z}_2$ we also encounter $\int_{M_6} x_2 y_4 \in \mathbb{Z}_2$, which originates from the ``x-y-part" in the Adams spectral sequence. Both stem from a $\mathbb{Z}_2$ in filtration zero and are therefore detected by their cohomological counterparts.
    Curiously, we encounter a somewhat unexpected asymmetry between $\Omega^{Spin}_6(pt) \cong ko_6(pt) = 0$ and $\widetilde{\Omega}^{Spin}_6(BSs(32)) \cong \widetilde{ko}_6(BSs(32)) = 2\, \mathbb{Z}_2$.
    We have already seen generally and for $ko_2(BSs(32))$ in particular that the splitting principle $ko_n(BSs(32)) \cong ko_n(pt) \oplus \widetilde{ko}_n(BSs(32))$ assorts the abelian groups into different open string sectors.
    Consequently, this suggests that our nontrivial $\widetilde{ko}_6(BSs(32))$ groups capture the charges of a non-BPS $\widehat{\text{D}3}$-brane arising because of open fundamental strings connecting the $\widehat{\text{D}3}$-brane to the background stack of 32 D9-branes and an O9-plane. \\
    Thus, our expectation is that this non-BPS D3-brane would only exist as a boundary condition for DN strings supplying it with the necessary gauge degrees of freedom to stabilize it, while for the DD sector there is no topological obstruction for a decay. 
   \cite{Frau:1999qs}~provides a general construction of the type I non-BPS $\widehat{Dp}$-branes as type IIB D$p$-$\overline{\text{D}p}$-brane bound states, where the generic tachyon can be projected out by the orientifold.
    By analyzing the DD sector of the open string amplitude precisely this anticipated tachyonic instability was found. One gets the following condition for the absence of tachyons
    \begin{equation}
        \mu_p = 2 \sin(\frac{\pi}{4}(9-p))\, >\, 0\,,
    \end{equation}
    which is only true for all the non-BPS $\widehat{Dp}$-branes classified by $ko_n(pt)$, i.e.\ $p = -1, \, 0,\, 7,\,8$.
    However, things change once we look at the DN sector. \cite{Frau:1999qs}~also give a criterion for tachyonic instability in the DN sector, namely:
    \begin{equation}
        a_{NS} = \frac{1}{2} - \frac{\nu}{8}\, <\, 0\,.
    \end{equation}
    Here, $\nu$ is the number of Dirichlet-Neumann directions available to the open string. Concretely, we see that our non-BPS $\widehat{\text{D}3}$-brane would be free of tachyons in this sector. Interestingly, this is completely orthogonal to the non-BPS $\widehat{D7}$ and $\widehat{\text{D}8}$-branes, which are unstable in this sector while they are stable in the DD sector. The consequences of this DN sector instability was explored in~\cite{Loaiza-Brito:2001yer}.
    It would be very interesting to further explore the DD instability in the presence of the D9-branes/O9-plane background with $\int_{M_6} x_2 y_4$ or $\int_{M_6}x_2^3$ non-trivial, saving the configuration from completely decaying to the vacuum.
    In~\cite{Witten:1997bs} the case of $\int_{T^6} x_2^3 \neq 0$ realized as a type IIB orientifold on $T^6/\mathbb{Z}_2$ was briefly mentioned. Of course this raises the question if this charge is cancelled. This is a bit different from the $n=4$ case, where the cohomology class $x_2^2$ does not detect the spin cobordism class.
    In the aforementioned orientifold model there are 2 D3-brane pairs present on the type IIB side. It seems reasonable to expect that the charge $\int_{T^6} x_2^3$ can be associated to a non-BPS $\widehat{\text{D}3}$-brane on the type I side and triviality is achieved by a tadpole cancelling configuration of them.
    
    From the perspective of the cobordism conjecture one might reasonably expect that the non-triviality is resolved by uplifting from spin-cobordism to a twisted string structure by properly taking the Bianchi identity into account. It might happen that both $\int_{M_6} x_2^3$ and $\int_{M_6} x_2 y_4$ do not detect any cobordism classes after the uplift.
    While this scenario provides a satisfying answer to the non-triviality of the spin cobordism classes, we would naively violate K-theoretic charge cancellation in the DN-sector.
    In particular in accordance to the non-BPS $\widehat{\text{D}7}$-brane case~\cite{Uranga:2000xp} we would expect some form of non-perturbative inconsistency for the type I/HO dual of the aforementioned type IIB model. 
    Besides the gauging of the charge there is also the pathway of breaking the corresponding global symmetries by suitable codimension 7 defects left to be explored. 

    So far we have not commented on the close connection between discrete flux choices and the cohomology classes we have encountered up until this point. In particular $\int_{M_2} x_2 \neq 0$ and $\int_{M_2} x_2^2 \neq 0$ are dual to type IIB orientifold configurations with discrete values for $B_2$ flux (or $B_2^2$ respectively)~\cite{Sen:1997pm, Bianchi:1997rf, Bianchi:1991eu}.
    Consequently, one might be able to explore compactifications of type I with nontrivial $\int_{M_6} x_2^3$ or $\int_{M_6} x_2 y_4$ from a dual type IIB perspective corresponding to specific discrete fluxes turned on. In~\cite{Morrison:2001ct} the authors explored a possible discrete 6-form flux and remarked the absence of a degree 6 ``generalized" Stiefel-Whitney class in the cohomology of $BSs(32)$, i.e.\ there is no $x_6$, complicating the identification of the discrete flux with a cohomology class.

    \item n = 8: Alike the $n=6$ case we attain a torsional cobordism group detected by a cohomological invariant $\int_{M_8} x_2^4$, which can also be tracked from the ko-homological viewpoint.
    It would be very interesting to investigate whether such compactifications with nontrivial torsional K-theory charge arise from F-theory fivefold compactifications~\cite{Schafer-Nameki:2016cfr}, presumably alike the other type I/HO compactifications ``without vector structure" in its frozen phase~\cite{Witten:1998cd, Tachikawa:2015wka, Oehlmann:2024cyn}.
    
\end{itemize}    

\subsection[The remaining torsional subgroups \texorpdfstring{$\Omega^{Spin}_{n > 8}(BSs(32))$}{above d = 8}]{\texorpdfstring{The remaining torsional subgroups \texorpdfstring{$\Omega^{Spin}_{n > 8}(BSs(32))$}{above d = 8}}{The remaining torsional subgroups of \texorpdfstring{$\Omega^{Spin}_{n > 8}(BSs(32))$}{above d = 8}}}

As we have seen in table \ref{table: Finalspinbordism}, the groups $\Omega^{Spin}_{n>8}(BSs(32))$ are crowded with torsional subgroups, which makes it tough to decipher the physical meaning of each one of them. Therefore, we will just make some general remarks. First let us mention $\Omega^{Spin}_{11}(BSs(32))$. Based on its non-triviality one would expect global anomalies~\cite{Witten:1985bt}. While the same calculation for the $(E_8 \times E_8) \rtimes \mathbb{Z}_2$ heterotic string approximated to $\Omega^{Spin}_{11}(BE_8)$ lead to a vanishing group, this is not the end of the story as the full computation has to take the Bianchi identity of the heterotic string into account and results in a non-trivial group~\cite{Debray:2023rlx}.
Still, by relying on the Segal-Stolz-Teichner conjecture~\cite{Stolz_Teichner_2004} the authors of~\cite{Tachikawa:2021mby} provide a general proof for absence of global symmetries in both supersymmetric heterotic string theories. Also, the expected absence of global anomalies in type I string theory has been confirmed in~\cite{Freed:2000ta} by utilizing KO-theory. 

The ko-homology subsector of the cobordism groups of dimension 9 and 10 once again can be linked to D$p$-branes. In dimension 9 we expect a correspondence to particle-like defects, whereas dimension 10 should classify instanton effects. Specifically, we meet the gravitational and gauge-theoretic instanton of~\cite{Witten:1985xe} classified by $ko_{10}(pt) \subset \Omega^{Spin}_{10}(BSs(32))$ and $\pi_{10}(BSs(32)) \subset \Omega^{Spin}_{10}(BSs(32))$ respectively. Those groups are detected by the mod 2 index of the Dirac operator and $\int_{S^{10}} trF^5 \in \mathbb{Z}_2$. 
The first one is argued to be gauged in~\cite{Witten:1985xe} as this particular mod 2 index is always even in string theory.
The latter one would lead to a similar anomaly (in 9d) as the prime example $\pi_4(SU(2))$~\cite{Witten:1982fp}.
The corresponding heterotic instantons are instrumental to resolving the origin of Shenker's $1/g_s$-effects in heterotic string theory~\cite{HamburgLondonMunich}. 

\section{Conclusions}
The conjectured incompatibility of quantum gravity with global symmetries leaves distinct imprints in its topological sector. The Cobordism Conjecture is a recent formalization of this general principle.
In particular it relates non-trivial cobordism groups to higher-form global symmetries in an effective field theory coupled to gravity, whose physics gets encoded in both the tangential structure and the background space of the relevant cobordism groups.
In this work we specifically took a look at the consequences of the Cobordism Conjecture for type I and its S-dual formulation as the $Spin(32)/\mathbb{Z}_2$, i.e.\ $Ss(32)$ in an unambiguous language, heterotic string theory.

To this end we computed the mod 2 cohomology of the classifying space for $Ss(32)$ $H^n(BSs(32), \mathbb{Z}_2)$ via the Eilenberg-Moore spectral sequence in order to feed the Adams Spectral sequence to reach our final goal  of calculating $\Omega^{Spin}(BSs(32))$. 
The physics behind the nontrivial spin cobordism groups can be nicely tracked through its ko-homology building blocks. Here, we observe that the splitting principle $ko_n(BSs(32)) = ko_n(pt) \oplus \widetilde{ko}_n(BSs(32))$ splits the charges classified by the K-theory groups into the ones arising from the Dirichlet-Dirichlet (Neumann-Neumann for $ko_0(pt)$) sector and Dirichlet-Neumann sector respectively. It should be stressed that at this level of refinement we specifically observe the open string sector of both type I/HO string as we account for the objects the endpoint(s) of the type I fundamental open string and the HO D-string live on. This matches nicely with~\cite{Hull:1998he} and will be explored further in~\cite{HamburgLondonMunich}. 
The close relation to ko-homology also carries over to the observation that the interpretation connects nicely to type I/HO string compactifications known as compactifications without vector structure.
Furthermore, we have also seen that for string theory setups with multiple simultaneously non-trivial cobordism groups experience stronger constraints than just from (K-theoretic) tadpole cancellation as it can only account for cancellation of a diagonal component of the groups. 
In the future, clearly the twisted string cobordism groups encoding the Bianchi identity of type I/HO string theory directly in the tangential structure have to be calculated. Also,
extending the computations of~\cite{Basile:2023knk} to $\Omega_n^{String-((Spin(16) * Spin(16))\rtimes \mathbb{Z}_2)}$\footnote{$(Spin(16) * Spin(16))\rtimes \mathbb{Z}_2$ was brought forward in~\cite{McInnes:1999pt} as the refinement to $O(16) \times O(16)$ to account for one of the many subtleties involving the non-supersymmetric heterotic string. Concretely, it is a cover for both $(Ss(16) \times Ss(16))\rtimes \mathbb{Z}_2$, which embeds into $(E_8 \times E_8)\rtimes\mathbb{Z}_2$, and the other quotient $(\frac{Spin(16) \times Spin(16)}{\mathbb{Z}_2 \times \mathbb{Z}_2})\rtimes \mathbb{Z}_2$, which embeds into $Ss(32)$. Hopefully, the calculation we have detailed in this work can facilitate the computations for all three groups.} would be very interesting as it could shed some light on non-perturbative objects of the unique non-supersymmetric heterotic string theory.

\section*{Acknowledgments}
I would like to thank R.~\'Alvarez-Garc\'ia, I.~Basile, R.~Blumenhagen, N.~Cribiori, A.~Debray, J.~Leedom, J.~McNamara and N.~Righi for many insightful discussions. Finally, I want to thank A.~Makridou for very helpful discussions and comments on the draft.

\clearpage

\appendix

\section{The Steenrod algebra}
\label{app_steenrodsq}

In this appendix, we collect some useful facts about Steenrod squares and the Steenrod algebra. For a nice, pedagogical review with a particular emphasis on the Adams spectral sequence, we refer the reader to~\cite{Beaudry:2018ifm}.
A very useful, general discussion on the topic of cohomology operations and specifically the Steenrod algebra can be found in~\cite{mosher2008cohomology}.
Let's start with an axiomatic definition of the Steenrod squares.
A cohomology operation of degree $i$ is a map
\begin{equation}
H^{n}(X;\mathbb{Z}_2)\to H^{n+i}(X;\mathbb{Z}_2).
\end{equation}
It is said to be stable if it commutes with the suspension isomorphism. Steenrod squares, $Sq^i$, are stable cohomology operations of degree $i$ satisfying the following properties, for any $i\geq 0$:
\begin{itemize}
\item[\bf 0.] $Sq^i$ is a natural homomorphism $H^{n}(X;\mathbb{Z}_2)\to H^{n+i}(X;\mathbb{Z}_2)$;
\item[\bf 1.] For $i<j$ $Sq^i(x)=0$, for all $x \in H^j(X;\mathbb{Z}_2)$; 
\item[{\bf 2.}] $Sq^i(x)=x \cup x$, for all $x \in H^i(X;\mathbb{Z}_2)$;
\item[{\bf 3.}] $Sq^0={\rm Id}$;
\item[{\bf 4.}] $Sq^1= \beta$ is the Bockstein homomorphism associated to the short exact sequence $ 0 \to \mathbb{Z}_2 \to\mathbb{Z}_4 \to \mathbb{Z}_2 \to 0$;
\item[{\bf 5.}] Cartan formula: $Sq^i(x\cup y) = \displaystyle\sum_{m+n=i} Sq^m(x) \cup Sq^n(y)$.
\item[{\bf 6.}] Adem relation: $Sq^i \circ Sq^j = \displaystyle \sum_{k=0}^{\lfloor i/2 \rfloor} \left(\begin{array}{c}j-k-1\\i-2k\end{array}\right)_{\mod 2} Sq^{i+j-k}\circ Sq^k$, \\ for $0<i < 2j$\,.
\end{itemize}
Then the Steenrod algebra $\mathcal{A}$ is a $\mathbb{Z}_2$ tensor algebra, whose elements are polynomials in $Sq^{i}$ satisfying both $Sq^0 = {\rm Id}$ and the Adem relations.
Importantly, $\mathcal{A}$ is generated by just $Sq^{2^n}$ with $n \geq 0$.
Furthermore, even though the Steenrod algebra is infinitely generated, in each degree it is finitely generated. These subalgebras, denoted $\mathcal{A}_n$,
are then generated by $Sq^1$, \dots, $Sq^{2^n}$.

\bibliography{references}  

\providecommand{\href}[2]{#2}\begingroup\raggedright\begin{thebibliography}{100}

\bibitem{Gross:1984dd}
D.~J. Gross, J.~A. Harvey, E.~J. Martinec, and R.~Rohm, ``{The Heterotic String},'' {\em Phys. Rev. Lett.} {\bf 54} (1985) 502--505.

\bibitem{Witten:1998cd}
E.~Witten, ``{D-branes and K-theory},'' {\em JHEP} {\bf 12} (1998) 019, \href{http://www.arXiv.org/abs/hep-th/9810188}{{\tt hep-th/9810188}}.

\bibitem{Witten:1995ex}
E.~Witten, ``{String theory dynamics in various dimensions},'' {\em Nucl. Phys. B} {\bf 443} (1995) 85--126, \href{http://www.arXiv.org/abs/hep-th/9503124}{{\tt hep-th/9503124}}.

\bibitem{Polchinski:1995df}
J.~Polchinski and E.~Witten, ``{Evidence for heterotic - type I string duality},'' {\em Nucl. Phys. B} {\bf 460} (1996) 525--540, \href{http://www.arXiv.org/abs/hep-th/9510169}{{\tt hep-th/9510169}}.

\bibitem{Goddard:1976qe}
P.~Goddard, J.~Nuyts, and D.~I. Olive, ``{Gauge Theories and Magnetic Charge},'' {\em Nucl. Phys. B} {\bf 125} (1977) 1--28.

\bibitem{Montonen:1977sn}
C.~Montonen and D.~I. Olive, ``{Magnetic Monopoles as Gauge Particles?},'' {\em Phys. Lett. B} {\bf 72} (1977) 117--120.

\bibitem{BAUM1965305}
P.~F. Baum and W.~Browder, ``{The cohomology of quotients of classical groups},'' {\em Topology} {\bf 3} (1965), no.~4, 305--336.

\bibitem{McInnes:1999va}
B.~McInnes, ``{The Semispin groups in string theory},'' {\em J. Math. Phys.} {\bf 40} (1999) 4699--4712, \href{http://www.arXiv.org/abs/hep-th/9906059}{{\tt hep-th/9906059}}.

\bibitem{McInnes:1999pt}
B.~McInnes, ``{Gauge spinors and string duality},'' {\em Nucl. Phys. B} {\bf 577} (2000) 439--460, \href{http://www.arXiv.org/abs/hep-th/9910100}{{\tt hep-th/9910100}}.

\bibitem{McNamara@Swamplandia}
J.~McNamara, ``{Cobordism, ER = EPR, and the Sum Over Topologies}.''
\newblock \href{https://eventos.uam.es/file_manager/getFile/145965.html}{Talk} at Swamplandia 2023.

\bibitem{Blumenhagen:2022bvh}
R.~Blumenhagen, N.~Cribiori, C.~Kneissl, and A.~Makridou, ``{Dimensional Reduction of Cobordism and K-theory},'' {\em JHEP} {\bf 03} (2023) 181, \href{http://www.arXiv.org/abs/2208.01656}{{\tt 2208.01656}}.

\bibitem{Vafa:2005ui}
C.~Vafa, ``{The String landscape and the swampland},'' \href{http://www.arXiv.org/abs/hep-th/0509212}{{\tt hep-th/0509212}}.

\bibitem{Palti:2019pca}
E.~Palti, ``{The Swampland: Introduction and Review},'' {\em Fortsch. Phys.} {\bf 67} (2019), no.~6, 1900037, \href{http://www.arXiv.org/abs/1903.06239}{{\tt 1903.06239}}.

\bibitem{Grana:2021zvf}
M.~Gra\~na and A.~Herr\'aez, ``{The Swampland Conjectures: A Bridge from Quantum Gravity to Particle Physics},'' {\em Universe} {\bf 7} (2021), no.~8, 273, \href{http://www.arXiv.org/abs/2107.00087}{{\tt 2107.00087}}.

\bibitem{vanBeest:2021lhn}
M.~van Beest, J.~Calder\'on-Infante, D.~Mirfendereski, and I.~Valenzuela, ``{Lectures on the Swampland Program in String Compactifications},'' {\em Phys. Rept.} {\bf 989} (2022) 1--50, \href{http://www.arXiv.org/abs/2102.01111}{{\tt 2102.01111}}.

\bibitem{Agmon:2022thq}
N.~B. Agmon, A.~Bedroya, M.~J. Kang, and C.~Vafa, ``{Lectures on the string landscape and the Swampland},'' \href{http://www.arXiv.org/abs/2212.06187}{{\tt 2212.06187}}.

\bibitem{McNamara:2019rup}
J.~McNamara and C.~Vafa, ``{Cobordism Classes and the Swampland},'' \href{http://www.arXiv.org/abs/1909.10355}{{\tt 1909.10355}}.

\bibitem{Misner:1957mt}
C.~W. Misner and J.~A. Wheeler, ``{Classical physics as geometry: Gravitation, electromagnetism, unquantized charge, and mass as properties of curved empty space},'' {\em Annals Phys.} {\bf 2} (1957) 525--603.

\bibitem{Banks:1988yz}
T.~Banks and L.~J. Dixon, ``{Constraints on String Vacua with Space-Time Supersymmetry},'' {\em Nucl. Phys. B} {\bf 307} (1988) 93--108.

\bibitem{Banks:2010zn}
T.~Banks and N.~Seiberg, ``{Symmetries and Strings in Field Theory and Gravity},'' {\em Phys. Rev. D} {\bf 83} (2011) 084019, \href{http://www.arXiv.org/abs/1011.5120}{{\tt 1011.5120}}.

\bibitem{Gaiotto:2014kfa}
D.~Gaiotto, A.~Kapustin, N.~Seiberg, and B.~Willett, ``{Generalized Global Symmetries},'' {\em JHEP} {\bf 02} (2015) 172, \href{http://www.arXiv.org/abs/1412.5148}{{\tt 1412.5148}}.

\bibitem{Harlow:2018jwu}
D.~Harlow and H.~Ooguri, ``{Constraints on Symmetries from Holography},'' {\em Phys. Rev. Lett.} {\bf 122} (2019), no.~19, 191601, \href{http://www.arXiv.org/abs/1810.05337}{{\tt 1810.05337}}.

\bibitem{Harlow:2018tng}
D.~Harlow and H.~Ooguri, ``{Symmetries in quantum field theory and quantum gravity},'' {\em Commun. Math. Phys.} {\bf 383} (2021), no.~3, 1669--1804, \href{http://www.arXiv.org/abs/1810.05338}{{\tt 1810.05338}}.

\bibitem{Debray:2021vob}
A.~Debray, M.~Dierigl, J.~J. Heckman, and M.~Montero, ``{The anomaly that was not meant IIB},'' \href{http://www.arXiv.org/abs/2107.14227}{{\tt 2107.14227}}.

\bibitem{Dierigl:2022reg}
M.~Dierigl, J.~J. Heckman, M.~Montero, and E.~Torres, ``{IIB string theory explored: Reflection 7-branes},'' {\em Phys. Rev. D} {\bf 107} (2023), no.~8, 086015, \href{http://www.arXiv.org/abs/2212.05077}{{\tt 2212.05077}}.

\bibitem{Debray:2023yrs}
A.~Debray, M.~Dierigl, J.~J. Heckman, and M.~Montero, ``{The Chronicles of IIBordia: Dualities, Bordisms, and the Swampland},'' \href{http://www.arXiv.org/abs/2302.00007}{{\tt 2302.00007}}.

\bibitem{Debray:2023rlx}
A.~Debray, ``{Bordism for the 2-group symmetries of the heterotic and CHL strings},'' \href{http://www.arXiv.org/abs/2304.14764}{{\tt 2304.14764}}.

\bibitem{Basile:2023knk}
I.~Basile, A.~Debray, M.~Delgado, and M.~Montero, ``{Global anomalies \& bordism of non-supersymmetric strings},'' {\em JHEP} {\bf 02} (2024) 092, \href{http://www.arXiv.org/abs/2310.06895}{{\tt 2310.06895}}.

\bibitem{Andriot:2022mri}
D.~Andriot, N.~Carqueville, and N.~Cribiori, ``{Looking for structure in the cobordism conjecture},'' {\em SciPost Phys.} {\bf 13} (2022), no.~3, 071, \href{http://www.arXiv.org/abs/2204.00021}{{\tt 2204.00021}}.

\bibitem{Basile:2024}
I.~Basile and C.~Kneissl, ``{to appear.},''.

\bibitem{Garcia-Etxebarria:2018ajm}
I.~Garc\'\i{}a-Etxebarria and M.~Montero, ``{Dai-Freed anomalies in particle physics},'' {\em JHEP} {\bf 08} (2019) 003, \href{http://www.arXiv.org/abs/1808.00009}{{\tt 1808.00009}}.

\bibitem{dieck2008algebraic}
T.~Dieck, {\em Algebraic Topology}.
\newblock EMS textbooks in mathematics. European Mathematical Society, 2008.

\bibitem{McCleary_2000}
J.~McCleary, {\em A User’s Guide to Spectral Sequences}.
\newblock Cambridge Studies in Advanced Mathematics. Cambridge University Press, 2~ed., 2000.

\bibitem{hatcher2004spectral}
A.~Hatcher, ``{\href{https://pi.math.cornell.edu/~hatcher/AT/ATch5.pdf}{Spectral sequences}},'' {\em preprint} (2004).

\bibitem{mosher2008cohomology}
R.~Mosher and M.~Tangora, {\em Cohomology Operations and Applications in Homotopy Theory}.
\newblock Dover Books on Mathematics Series. Dover Publications, 2008.

\bibitem{Beaudry:2018ifm}
A.~Beaudry and J.~A. Campbell, ``{A Guide for Computing Stable Homotopy Groups},'' \href{http://www.arXiv.org/abs/1801.07530}{{\tt 1801.07530}}.

\bibitem{Witten:1985bt}
E.~Witten, ``{Topological Tools in Ten-dimensional Physics},'' {\em Int. J. Mod. Phys. A} {\bf 1} (1986) 39.

\bibitem{Stong:1985vj}
R.~E. Stong, ``{Calculation of $\Omega_{11}^{Spin}(K(\mathbb{Z}, 4))$},'' in {\em {Workshop on Unified String Theories}}.
\newblock 1985.

\bibitem{Freed:2016rqq}
D.~S. Freed and M.~J. Hopkins, ``{Reflection positivity and invertible topological phases},'' {\em Geom. Topol.} {\bf 25} (2021) 1165--1330, \href{http://www.arXiv.org/abs/1604.06527}{{\tt 1604.06527}}.

\bibitem{Wan:2018bns}
Z.~Wan and J.~Wang, ``{Higher anomalies, higher symmetries, and cobordisms I: classification of higher-symmetry-protected topological states and their boundary fermionic/bosonic anomalies via a generalized cobordism theory},'' {\em Ann. Math. Sci. Appl.} {\bf 4} (2019), no.~2, 107--311, \href{http://www.arXiv.org/abs/1812.11967}{{\tt 1812.11967}}.

\bibitem{Davighi:2020bvi}
J.~Davighi and N.~Lohitsiri, ``{Anomaly interplay in $U(2)$ gauge theories},'' {\em JHEP} {\bf 05} (2020) 098, \href{http://www.arXiv.org/abs/2001.07731}{{\tt 2001.07731}}.

\bibitem{Davighi:2020uab}
J.~Davighi and N.~Lohitsiri, ``{The algebra of anomaly interplay},'' {\em SciPost Phys.} {\bf 10} (2021), no.~3, 074, \href{http://www.arXiv.org/abs/2011.10102}{{\tt 2011.10102}}.

\bibitem{Lee:2020ewl}
Y.~Lee and Y.~Tachikawa, ``{Some comments on 6D global gauge anomalies},'' {\em PTEP} {\bf 2021} (2021), no.~8, 08B103, \href{http://www.arXiv.org/abs/2012.11622}{{\tt 2012.11622}}.

\bibitem{Davighi:2020kok}
J.~Davighi and N.~Lohitsiri, ``{Omega vs. pi, and 6d anomaly cancellation},'' {\em JHEP} {\bf 05} (2021) 267, \href{http://www.arXiv.org/abs/2012.11693}{{\tt 2012.11693}}.

\bibitem{Debray:2021rik}
A.~Debray, ``{Invertible phases for mixed spatial symmetries and the fermionic crystalline equivalence principle},'' \href{http://www.arXiv.org/abs/2102.02941}{{\tt 2102.02941}}.

\bibitem{Lee:2022spd}
Y.~Lee and K.~Yonekura, ``{Global anomalies in 8d supergravity},'' {\em JHEP} {\bf 07} (2022) 125, \href{http://www.arXiv.org/abs/2203.12631}{{\tt 2203.12631}}.

\bibitem{Davighi:2022icj}
J.~Davighi, B.~Gripaios, and N.~Lohitsiri, ``{Anomalies of non-Abelian finite groups via cobordism},'' {\em JHEP} {\bf 09} (2022) 147, \href{http://www.arXiv.org/abs/2207.10700}{{\tt 2207.10700}}.

\bibitem{McNamara:2022lrw}
J.~McNamara and M.~Reece, ``{Reflections on Parity Breaking},'' \href{http://www.arXiv.org/abs/2212.00039}{{\tt 2212.00039}}.

\bibitem{Davighi:2023luh}
J.~Davighi, N.~Lohitsiri, and A.~Debray, ``{Toric 2-group anomalies via cobordism},'' {\em JHEP} {\bf 07} (2023) 019, \href{http://www.arXiv.org/abs/2302.12853}{{\tt 2302.12853}}.

\bibitem{Debray:2023tdd}
A.~Debray and M.~Yu, ``{Adams spectral sequences for non-vector-bundle Thom spectra},'' \href{http://www.arXiv.org/abs/2305.01678}{{\tt 2305.01678}}.

\bibitem{eilenberg1945axiomatic}
S.~Eilenberg and N.~E. Steenrod, ``Axiomatic approach to homology theory,'' {\em Proceedings of the National Academy of Sciences} {\bf 31} (1945), no.~4, 117--120.

\bibitem{may1999concise}
J.~P. May, {\em A concise course in algebraic topology}.
\newblock University of Chicago press, 1999.

\bibitem{Maldacena:2001xj}
J.~M. Maldacena, G.~W. Moore, and N.~Seiberg, ``{D-brane instantons and K theory charges},'' {\em JHEP} {\bf 11} (2001) 062, \href{http://www.arXiv.org/abs/hep-th/0108100}{{\tt hep-th/0108100}}.

\bibitem{Freed:1999vc}
D.~S. Freed and E.~Witten, ``{Anomalies in string theory with D-branes},'' {\em Asian J. Math.} {\bf 3} (1999) 819, \href{http://www.arXiv.org/abs/hep-th/9907189}{{\tt hep-th/9907189}}.

\bibitem{Adams1957/58}
J.~Adams, ``On the Structure and Applications of the Steenrod Algebra,'' {\em Commentarii mathematici Helvetici} {\bf 32} (1957/58) 180--214.

\bibitem{ABP}
D.~W. Anderson, E.~H. Brown, and F.~P. Peterson, ``The Structure of the Spin Cobordism Ring,'' {\em Annals of Mathematics} {\bf 86} (1967), no.~2, 271--298.

\bibitem{Stong63}
R.~E. Stong, ``Determination of $H^*(BO(k,...,\infty), \mathbb{Z}_2)$ and $H^*(BU(k,...,\infty), \mathbb{Z}_2)$,'' {\em Transactions of the American Mathematical Society} {\bf 107} (1963), no.~3, 526--544.

\bibitem{moore1959algebre}
J.~C. Moore, ``Algebre homologique et homologie des espaces classifiants,'' {\em S{\'e}minaire Henri Cartan} {\bf 12} (1959), no.~1, 1--37.

\bibitem{rothenberg1965cohomology}
M.~Rothenberg and N.~E. Steenrod, ``The cohomology of classifying spaces of H-spaces,''.

\bibitem{eilenberg1966homology}
S.~Eilenberg and J.~C. Moore, ``Homology and fibrations. I. Coalgebras, cotensor product and its derived functors,'' {\em Comment. Math. Helv} {\bf 40} (1966) 199--236.

\bibitem{TachikawaMO}
Y.~Tachikawa and T.~Nishimoto, ``Cohomology of the classifying space of $ss(4m)$.''
\newblock \href{https://mathoverflow.net/questions/143235/cohomology-of-the-classifying-space-of-ss4m}{Question and answer} on Mathoverflow.

\bibitem{Ishitoya:1976pf}
K.~Ishitoya, A.~Kono, and H.~Toda, ``{Hopf Algebra Structure of mod-2 Cohomology of Simple Lie Groups},'' {\em Publ. Res. Inst. Math. Sci. Kyoto} {\bf 12} (1976), no.~1, 141--167.

\bibitem{may1964cohomology}
J.~P. May, {\em The cohomology of restricted Lie algebras and of Hopf algebras: application to the Steenrod algebra}.
\newblock PhD thesis, Princeton University, 1964.

\bibitem{Kono75}
A.~Kono, M.~Mimura, and N.~Shimada, ``Cohomology mod 2 of the classifying space of the compact connected lie group of type E6,'' {\em Journal of Pure and Applied Algebra} {\bf Volume 6, Issue 1} (1975) 61--81.

\bibitem{Kono76}
A.~Kono, M.~Mimura, and N.~Shimada, ``On the cohomology mod 2 of the classifying space of the 1-connected exceptional Lie group E7,'' {\em Journal of Pure and Applied Algebra} {\bf Volume 8, Issue 3} (1976) 267--283.

\bibitem{wilson1973new}
S.~Wilson, ``A new relation on the Stiefel-Whitney classes of spin manifolds,'' {\em Illinois Journal of Mathematics} {\bf 17} (1973), no.~1, 115--127.

\bibitem{Margolis74}
H.~R. Margolis, ``{Eilenberg-Mac Lane Spectra},'' {\em Proceedings of the American Mathematical Society} {\bf 43} (1974), no.~2, 409--415.

\bibitem{francisintegrals}
J.~Francis, ``{Integrals on spin manifolds and the K-theory of K (Z, 4)}.''

\bibitem{douady11suite}
A.~Douady, ``{La suite spectrale d’Adams: structure multiplicative, exp. 19},'' {\em S{\'e}minaire Henri Cartan, Ecole Normale Sup{\'e}rieure, vol} {\bf 11} (1959), no.~2, 1958--1959.

\bibitem{atiyah1961vector}
M.~F. Atiyah and F.~Hirzebruch, ``Vector bundles and homogeneous spaces,''.

\bibitem{clement2002integral}
A.~Cl{\'e}ment, {\em Integral cohomology of finite Postnikov towers}.
\newblock PhD thesis, Universit{\'e} de Lausanne, Facult{\'e} des sciences, 2002.

\bibitem{baker2020homotopy}
A.~Baker, ``Homotopy theory of modules over a commutative $ S $-algebra: some tools and examples,'' {\em arXiv preprint arXiv:2003.12003} (2020).

\bibitem{Douglas:1995bn}
M.~R. Douglas, ``{Branes within branes},'' {\em NATO Sci. Ser. C} {\bf 520} (1999) 267--275, \href{http://www.arXiv.org/abs/hep-th/9512077}{{\tt hep-th/9512077}}.

\bibitem{Hull:1998he}
C.~M. Hull, ``{The Nonperturbative SO(32) heterotic string},'' {\em Phys. Lett. B} {\bf 462} (1999) 271--276, \href{http://www.arXiv.org/abs/hep-th/9812210}{{\tt hep-th/9812210}}.

\bibitem{Hull:1997kt}
C.~M. Hull, ``{Gravitational duality, branes and charges},'' {\em Nucl. Phys. B} {\bf 509} (1998) 216--251, \href{http://www.arXiv.org/abs/hep-th/9705162}{{\tt hep-th/9705162}}.

\bibitem{Berkooz:1996iz}
M.~Berkooz, R.~G. Leigh, J.~Polchinski, J.~H. Schwarz, N.~Seiberg, and E.~Witten, ``{Anomalies, dualities, and topology of D = 6 N=1 superstring vacua},'' {\em Nucl. Phys. B} {\bf 475} (1996) 115--148, \href{http://www.arXiv.org/abs/hep-th/9605184}{{\tt hep-th/9605184}}.

\bibitem{Witten:1997bs}
E.~Witten, ``{Toroidal compactification without vector structure},'' {\em JHEP} {\bf 02} (1998) 006, \href{http://www.arXiv.org/abs/hep-th/9712028}{{\tt hep-th/9712028}}.

\bibitem{Uranga:2000xp}
A.~M. Uranga, ``{D-brane probes, RR tadpole cancellation and K theory charge},'' {\em Nucl. Phys. B} {\bf 598} (2001) 225--246, \href{http://www.arXiv.org/abs/hep-th/0011048}{{\tt hep-th/0011048}}.

\bibitem{Atiyah_V:1971}
M.~F. Atiyah and I.~M. Singer, ``{The Index of Elliptic Operators: V},'' {\em Annals of Mathematics} {\bf 93} (1971), no.~1, 139--149.

\bibitem{Gukov:1999yn}
S.~Gukov, ``{K theory, reality, and orientifolds},'' {\em Commun. Math. Phys.} {\bf 210} (2000) 621--639, \href{http://www.arXiv.org/abs/hep-th/9901042}{{\tt hep-th/9901042}}.

\bibitem{Kaidi:2023tqo}
J.~Kaidi, K.~Ohmori, Y.~Tachikawa, and K.~Yonekura, ``{Nonsupersymmetric Heterotic Branes},'' {\em Phys. Rev. Lett.} {\bf 131} (2023), no.~12, 121601, \href{http://www.arXiv.org/abs/2303.17623}{{\tt 2303.17623}}.

\bibitem{Horowitz:1991cd}
G.~T. Horowitz and A.~Strominger, ``{Black strings and P-branes},'' {\em Nucl. Phys. B} {\bf 360} (1991) 197--209.

\bibitem{Blumenhagen:2022mqw}
R.~Blumenhagen, N.~Cribiori, C.~Kneissl, and A.~Makridou, ``{Dynamical cobordism of a domain wall and its companion defect 7-brane},'' {\em JHEP} {\bf 08} (2022) 204, \href{http://www.arXiv.org/abs/2205.09782}{{\tt 2205.09782}}.

\bibitem{DeFreitas:2024yzr}
H.~P. De~Freitas, ``{T-duality for non-critical heterotic strings},'' \href{http://www.arXiv.org/abs/2407.12923}{{\tt 2407.12923}}.

\bibitem{Gimon:1996rq}
E.~G. Gimon and J.~Polchinski, ``{Consistency conditions for orientifolds and D-manifolds},'' {\em Phys. Rev. D} {\bf 54} (1996) 1667--1676, \href{http://www.arXiv.org/abs/hep-th/9601038}{{\tt hep-th/9601038}}.

\bibitem{Polchinski:1996ry}
J.~Polchinski, ``{Tensors from K3 orientifolds},'' {\em Phys. Rev. D} {\bf 55} (1997) 6423--6428, \href{http://www.arXiv.org/abs/hep-th/9606165}{{\tt hep-th/9606165}}.

\bibitem{Aspinwall:1996vc}
P.~S. Aspinwall, ``{Point - like instantons and the spin (32) / Z(2) heterotic string},'' {\em Nucl. Phys. B} {\bf 496} (1997) 149--176, \href{http://www.arXiv.org/abs/hep-th/9612108}{{\tt hep-th/9612108}}.

\bibitem{hirzebruch1951klasse}
F.~Hirzebruch, ``{\"U}ber eine Klasse von einfach-zusammenh{\"a}ngenden komplexen Mannigfaltigkeiten,'' {\em Mathematische Annalen} {\bf 124} (1951) 77--86.

\bibitem{Frau:1999qs}
M.~Frau, L.~Gallot, A.~Lerda, and P.~Strigazzi, ``{Stable nonBPS D-branes in type I string theory},'' {\em Nucl. Phys. B} {\bf 564} (2000) 60--85, \href{http://www.arXiv.org/abs/hep-th/9903123}{{\tt hep-th/9903123}}.

\bibitem{Loaiza-Brito:2001yer}
O.~Loaiza-Brito and A.~M. Uranga, ``{The Fate of the type I nonBPS D7-brane},'' {\em Nucl. Phys. B} {\bf 619} (2001) 211--231, \href{http://www.arXiv.org/abs/hep-th/0104173}{{\tt hep-th/0104173}}.

\bibitem{Sen:1997pm}
A.~Sen and S.~Sethi, ``{The Mirror transform of type I vacua in six-dimensions},'' {\em Nucl. Phys. B} {\bf 499} (1997) 45--54, \href{http://www.arXiv.org/abs/hep-th/9703157}{{\tt hep-th/9703157}}.

\bibitem{Bianchi:1997rf}
M.~Bianchi, ``{A Note on toroidal compactifications of the type I superstring and other superstring vacuum configurations with sixteen supercharges},'' {\em Nucl. Phys. B} {\bf 528} (1998) 73--94, \href{http://www.arXiv.org/abs/hep-th/9711201}{{\tt hep-th/9711201}}.

\bibitem{Bianchi:1991eu}
M.~Bianchi, G.~Pradisi, and A.~Sagnotti, ``{Toroidal compactification and symmetry breaking in open string theories},'' {\em Nucl. Phys. B} {\bf 376} (1992) 365--386.

\bibitem{Morrison:2001ct}
D.~R. Morrison and S.~Sethi, ``{Novel type I compactifications},'' {\em JHEP} {\bf 01} (2002) 032, \href{http://www.arXiv.org/abs/hep-th/0109197}{{\tt hep-th/0109197}}.

\bibitem{Schafer-Nameki:2016cfr}
S.~Sch\"afer-Nameki and T.~Weigand, ``{F-theory and 2d $(0, 2)$ theories},'' {\em JHEP} {\bf 05} (2016) 059, \href{http://www.arXiv.org/abs/1601.02015}{{\tt 1601.02015}}.

\bibitem{Tachikawa:2015wka}
Y.~Tachikawa, ``{Frozen singularities in M and F theory},'' {\em JHEP} {\bf 06} (2016) 128, \href{http://www.arXiv.org/abs/1508.06679}{{\tt 1508.06679}}.

\bibitem{Oehlmann:2024cyn}
P.-K. Oehlmann, F.~Ruehle, and B.~Sung, ``{The Frozen Phase of Heterotic F-theory Duality},'' \href{http://www.arXiv.org/abs/2404.02191}{{\tt 2404.02191}}.

\bibitem{Stolz_Teichner_2004}
S.~Stolz and P.~Teichner, {\em What is an elliptic object?}, p.~247–343.
\newblock London Mathematical Society Lecture Note Series.
\newblock Cambridge University Press, 2004.

\bibitem{Tachikawa:2021mby}
Y.~Tachikawa and M.~Yamashita, ``{Topological Modular Forms and the Absence of All Heterotic Global Anomalies},'' {\em Commun. Math. Phys.} {\bf 402} (2023), no.~2, 1585--1620, \href{http://www.arXiv.org/abs/2108.13542}{{\tt 2108.13542}}. [Erratum: Commun.Math.Phys. 402, 2131 (2023)].

\bibitem{Freed:2000ta}
D.~S. Freed, ``{Dirac charge quantization and generalized differential cohomology},'' \href{http://www.arXiv.org/abs/hep-th/0011220}{{\tt hep-th/0011220}}.

\bibitem{Witten:1985xe}
E.~Witten, ``{Global Gravitational Anomalies},'' {\em Commun. Math. Phys.} {\bf 100} (1985) 197.

\bibitem{Witten:1982fp}
E.~Witten, ``{An SU(2) Anomaly},'' {\em Phys. Lett. B} {\bf 117} (1982) 324--328.

\bibitem{HamburgLondonMunich}
R.~\'Alvarez-Garc\'ia, C.~Kneissl, J.~M. Leedom, and N.~Righi, ``{The Theory of Open Heterotic Strings and Their Non-Perturbative Effects},''. {to appear.}

\end{thebibliography}\endgroup
\bibliographystyle{utphys}

\end{document}